\documentclass[a4paper,fleqn]{cas-sc}

\usepackage[authoryear,longnamesfirst]{natbib}

\usepackage{lastpage}
\usepackage[utf8]{inputenc} 
\usepackage[T1]{fontenc}    
\usepackage{hyperref}       
\usepackage{url}            
\usepackage{booktabs}       
\usepackage{amsfonts}       
\usepackage{nicefrac}       
\usepackage{xcolor}         
\usepackage{changepage}
\usepackage{graphicx}
\usepackage{subcaption}
\usepackage{float}
\usepackage{fancyvrb}
\usepackage{algorithm}
\usepackage{algorithmicx}
\usepackage{algpseudocode}
\usepackage{amssymb,amsxtra,amsmath}
\usepackage{amsthm}
\usepackage{latexsym}
\usepackage{mathtools}
\usepackage{ifthen}
\usepackage{multirow}
\usepackage{hhline}
\usepackage{xurl}
\usepackage{Article}
\usepackage{cleveref}
\usepackage{arydshln}
\usepackage{tikz}
\usetikzlibrary{positioning,shapes.geometric,arrows.meta}
\newtheorem{theorem}{Theorem}
\newtheorem{corollary}{Corollary}

\algnewcommand\algorithmiclet{\textbf{Let:}}
\algnewcommand\Let{\item[\algorithmiclet]}

\makeatletter
\algnewcommand{\AlgLabel}[1]{%
  \item[\textbf{#1:}]%
  \setcounter{ALG@line}{0}%
}
\makeatother

\makeatletter
\newcommand{\DisableBooktabs}{%
  \let\toprule\relax
  \let\midrule\relax
  \let\bottomrule\relax
}
\makeatother

\def\A{\mathcal{A}}
\def\V{\mathcal{V}}
\def\bv{\mathbf{v}}

\def\ba{\mathbf{a}}

\def\bb{\mathbf{b}}
\def\ub{\underline{b}}
\def\bc{\mathbf{c}}
\def\uc{\underline{c}}
\def\bf{\mathbf{f}}

\def\G{\mathcal{G}}
\def\Gtilde{\tilde{\mathcal{G}}}
\def\ftilde{\tilde{f}}
\def\bPhi{\mathbf{\Phi}}
\def\bK{\mathbf{K}}
\def\bLambda{\mathbf{\Lambda}}

\definecolor{edit}{rgb}{0.,0.,0.}
\definecolor{tbd}{rgb}{0.,0.,1.}

\begin{document}
\shorttitle{Trustworthy Fluids}
\shortauthors{Sharma et~al.}

\title[mode = title]{Fluids You Can Trust: Property-Preserving Operator Learning for Incompressible Flows}

\author[1]{Ramansh Sharma}
\ead{ramansh@cs.utah.edu}
\author[1]{Matthew Lowery}
\ead{mlowery@cs.utah.edu}
\author[2]{Houman Owhadi}
\ead{owhadi@caltech.edu}
\author[1]{Varun Shankar}
\cormark[1]
\ead{shankar@cs.utah.edu}

\cortext[cor1]{Corresponding author}

\affiliation[1]{organization={Kahlert School of Computing, University of Utah},state={UT}, country={USA}}
\affiliation[2]{organization={Department of Computing and Mathematical Sciences, California Institute of Technology},state={CA}, country={USA}}

\begin{abstract}
We present a novel property-preserving kernel-based operator learning method for incompressible flows governed by the incompressible Navier--Stokes equations. Traditional numerical solvers incur significant computational costs to respect incompressibility. Operator learning offers efficient surrogate models, but current neural operators fail to exactly enforce physical properties such as incompressibility, periodicity, and turbulence. Our kernel method maps input functions to expansion coefficients of output functions in a property-preserving kernel basis, ensuring that predicted velocity fields \textit{analytically} and \textit{simultaneously} preserve the aforementioned physical properties. We present universal approximation results and worst-case {a priori} convergence rates for our framework; empirically, the observed convergence rates exceed the pessimistic predictions across most benchmarks, motivating a formulation of more optimistic convergence rates. Beyond the property-preserving formulation, a central contribution of this work is a novel computational framework revolving around streaming construction of kernel Gramians and recursive Schur-complement algorithms to solve the large block linear systems arising from property-preserving kernel methods. This kernel-based framework makes operator learning practical at scales up to $10{,}000$ training functions sampled at up to $10{,}000$ spatial locations. We evaluate the method on challenging 2D and 3D, laminar and turbulent, incompressible flow problems. Our method achieves up to six orders of magnitude lower relative $\ell_2$ errors upon generalization and trains up to five orders of magnitude faster compared to neural operators. Moreover, while our method enforces incompressibility analytically, neural operators exhibit large divergence errors. Our results show that our method provides an accurate and efficient surrogate for incompressible flows.
\end{abstract}

\begin{keywords}
operator learning, surrogate modeling, incompressible flow, kernel methods
\end{keywords}

\maketitle

\section{Introduction}
\label{sec:intro}
Incompressible fluid flows arise in an enormous range of engineering and scientific applications, such as the study of flow past airfoils \& wings~\citep{thwaites1960incompressible}, aerodynamics~\citep{kwak1986three, chang1988numerical, kwak2009cfd}, weather prediction~\citep{castorrini2023investigations}, chemical mixing~\citep{najm2005modeling}, and hemodynamics~\citep{janela20102783, suncica2003effective, secomb2016hemodynamics, womersley1955method, mcdonald1955relation, perktold1995computer}. These flows are typically modeled with the incompressible Navier--Stokes (INS) equations given by:
\begin{align}
\label{eq:ins}
\frac{\partial \bu}{\partial t} + (\bu \cdot \nabla) \bu &= \nu \nabla^2 \bu - \frac{1}{\rho} \nabla p + \frac{1}{\rho} \mathbf{f}, \quad \texttt{on } \Omega \times (0,T], \\
\nabla\cdot\bu &= 0, \quad \texttt{on } \Omega \times (0,T], \\
\mathcal{B} \bu &= \mathbf{g}, \quad \texttt{on } \partial \Omega \times (0,T], \\
\bu(\mathbf{x}, 0) &= \bu_0(\mathbf{x}),  \quad \texttt{on } \Omega,
\end{align}
where $\bu$ is the divergence-free velocity field, $p$ is the pressure, $\nu$ is the kinematic viscosity, $\rho$ is the fluid density, $\mathbf{f}$ is the external body forces, $\mathcal{B}$ denotes the boundary operator, and $\bu_0$ and $\mathbf{g}$ are the initial and boundary conditions respectively. These equations constitute a system of nonlinear partial differential equations (PDEs) and encompass a variety of physical phenomena such as boundary layers, vortex dynamics, flow separation, and turbulence. The challenging nature of these PDEs necessitates specialized numerical discretizations of which there are four dominant classes: (1) saddle-point methods, which solve the coupled PDE system directly and reach arbitrary orders of accuracy~\citep{orszag1986boundary, karniadakis1991high, john2002higher, ahmed2018assessment}; (2) high-order methods that treat the pressure explicitly and recover it via the pressure Poisson equation (``PPE'' methods)~\citep{johnston2004accurate,john2002higher, rosales2021high}; (3) artificial compressibility methods~\citep{chorin1967numerical, guermond2015high}; and (4) temporally low-order projection-based methods~\citep{chorin1993mathematical, brown2001accurate,goda1979multistep, kan1986secondorder,bell1989secondorder, guy2005stability, guermond2006overview}. All of these methods come with their own computational bottlenecks, typically requiring problem-specific preconditioners~\citep{elman2002preconditioners, john2002higher, ahmed2018assessment}, expensive Poisson solves~\citep{johnston2004accurate,jung2013heterogenous, dick2015solving} or vector Helmholtz solves~\citep{guermond2015high,guermond2017high}. Further, explicit treatment of the nonlinear advection term leads to the Courant--Friedrichs--Lewy (CFL) stability constraint, severely limiting admissible time-step size~\citep{kress2006time}. These issues have led to a growing interest in developing \emph{surrogate models} that can reduce the cost of obtaining solutions for new problem configurations~\citep{ionita2014data, ohlberger2015reduced, tezzele2020enhancing, bhattacharya2021model}.
Surrogate models for PDEs include reduced order models~\citep{holmes2012turbulence, quarteroni2015reduced, bacho2026krom}, machine learning (ML) based methods~\citep{lee2020model, fresca2020deep,sirignano2018dgm, thuerey2020deep,sanchez2020learning, pfaff2020learning,PINNs1}. More recently, operator learning models~\citep{boulle2024mathematical} that act directly on function spaces have emerged as a competitive class of surrogates. Operator learning seeks to approximate the \textit{solution operator} of the PDE as a mapping from input functions (geometry, PDE parameters, initial conditions, and/or boundary conditions) to output functions (solutions). 
Prominent neural operator architectures include deep operator networks (DeepONets)~\citep{lu2019deeponet}, which train two separate neural networks in conjunction; Fourier neural operators (FNOs)~\citep{li2020fourier} and their geometrically-flexible variants such as the Geo-FNO~\citep{li2020fourier}, which involve kernel-based integral operators computed implicitly via the fast Fourier transform (FFT); the kernel neural operator (KNO)~\citep{lowery2024kernel} which uses explicit kernels and quadrature instead of the FFT; deep Green networks (DGN)~\citep{gin2021deepgreen, boulle2022data} and their graph neural operator (GNO) counterparts~\citep{li2020neural}, which aim to learn the Green's kernel either globally or locally; and more recently, transformer-based neural operators that leverage the attention mechanism~\citep{hao2023gnot, liu2025geometry, liu2024mitigating}. Complementary to these methods, recent work~\cite{batlle2024kernel,mora2025operator} has shown that kernel/Gaussian process (GP) regression can be competitive for operator learning as well. Kernel methods for operator learning are meshless and typically use only a single trainable parameter; note that operator-valued kernels were first introduced in~\citet{kadri2011operator, kadri2016operator}. 
Despite these advances, to the best of our knowledge, current operator learning techniques are unable to \emph{analytically} and \emph{simultaneously} satisfy multiple fluid properties such as incompressibility, periodicity, and turbulence-related power laws. Recent attempts in this direction with either soft constraints~\citep{goswami2022physics, li2023phase, li2024physics, zhang2025floating} or hard constraints~\cite{richter2022neural,khorrami2024physics} have only succeeded in approximate enforcement of incompressibility. Unfortunately, a surrogate that fails to analytically satisfy these properties can produce physically-inconsistent predictions, accumulate spurious divergence, or be unable to replicate key flow features even if its pointwise errors are small. \textcolor{edit}{Furthermore, existing kernel-based operator learning methods typically rely on dense kernel matrices whose storage and solution costs grow rapidly with the number of training functions, limiting their scalability in large data regimes.} 
\begin{figure}[pos=!tpb]
    \includegraphics[width=\linewidth]{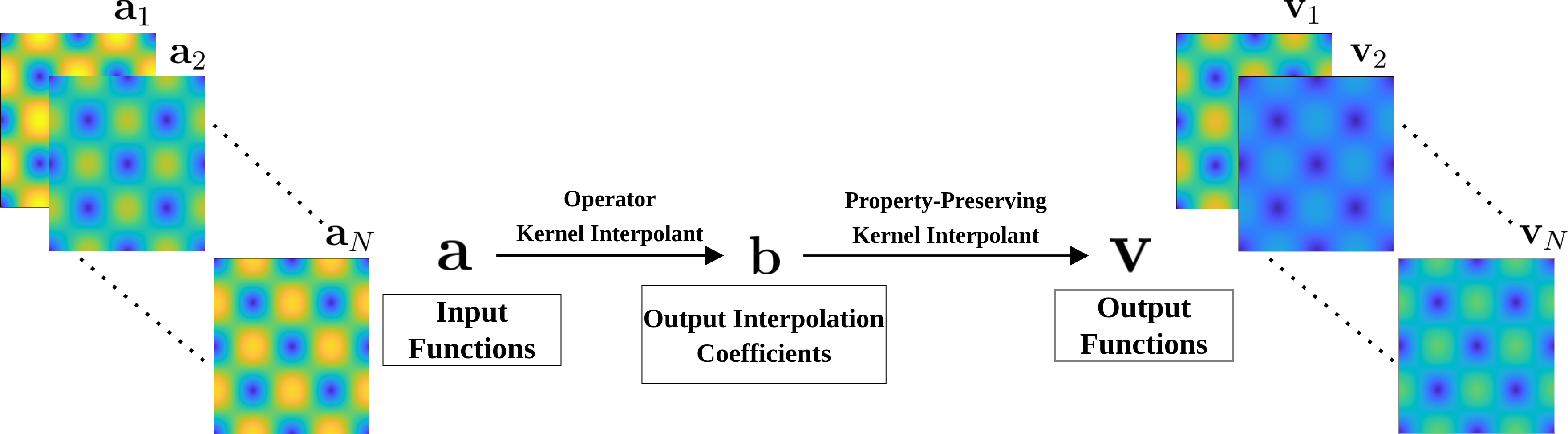}
    \caption{Schematic diagram of the proposed property-preserving kernel method.}
    \label{fig:schematic}
\end{figure}

In this work, we address this gap by introducing a novel property-preserving kernel-based operator learning method for incompressible flows. Unlike the kernel method in~\citet{batlle2024kernel}, we learn a map from input function samples to \emph{interpolation coefficients} associated with output functions, where these coefficients are associated with a kernel basis that is analytically property-preserving. Our framework uses two kernel interpolants: (1) one that interpolates the output functions using a property-preserving kernel basis, and (2) one that consumes input functions and fits the interpolation coefficients from (1). The resulting method uses only two trainable parameters. This approach decouples pointwise generalization errors in our surrogate models from their ability to respect physical constraints; while the former still depends on the training data, the latter is always guaranteed to machine precision; see Figure~\ref{fig:schematic} for an illustration of our method. \textcolor{edit}{Importantly, the novelty of our method extends beyond the property-preserving formulation: a second central contribution is the computational framework required to make it practical at scale. Specifically, because our method uses dense kernel matrices, the block linear systems that arise from the use of property-preserving kernels become prohibitively expensive or outright infeasible to solve naively. We overcome this bottleneck by employing a streaming construction of the operator-kernel and recursive Schur-complement algorithms for solving the (spacetime) block linear systems. This allows us to train on up to $10{,}000$ training functions sampled at up to $10{,}000$ spatial locations.}

We illustrate the capabilities of our method with extensive numerical experiments on a diverse suite of 2D and 3D problems involving incompressible flows. We construct surrogates that are all analytically incompressible on generalization, but that also analytically satisfy periodicity or \textcolor{edit}{asymptotically match the Richardson cascade through} turbulence-related power laws when applicable. Across these problems, we compare our method against the vanilla kernel method in~\citet{batlle2024kernel}, and two state-of-the-art neural operator baselines, Geo-FNO~\citep{li2023fourier} and Transolver~\citep{wu2024transolver}, \textcolor{edit}{each with and without an added soft constraint for incompressibility}. The resulting method is meshless; interpretable; typically orders of magnitude more accurate than existing neural operator approaches even in pointwise generalization errors; analytically property-preserving (if the properties are known in advance); more efficient than neural operators in terms of trainable parameter counts (two parameters for our method vs millions for neural operators); and orders of magnitude faster in terms of training \textcolor{edit}{and inference time}. \textcolor{edit}{We also compare the empirical error convergence rates of our method with the theoretical estimates provided in Section~\ref{sec:approx-thms}. Across most problems, particularly those with higher-dimensional input functions, the observed errors decay substantially faster than the worst-case estimates (prompting us to fit an effective dimension parameter in more optimistic error estimates), while on a few problems they exhibit slow or flat error trends that may not reflect the asymptotic behavior predicted by the theorems. We discuss these comparisons in Section~\ref{sec:error_analysis}.}

The remainder of the paper is organized as follows. In Section~\ref{sec:methods}, we present an overview of operator learning and both the high-level algorithm and the mathematical formulation of the property-preserving kernel method, with a focus on incompressibility, periodicity, and turbulence-related power laws. In Section \ref{sec:approx-thms}, we then present approximation theorems regarding universal approximation and convergence rates of the method. We present results of numerical experiments in Section~\ref{sec:results}, comparing our method against competitors in terms of accuracy, divergence, and runtime. Finally, in Section~\ref{sec:discussion}, we discuss the implications of the results in detail and discuss future directions.
%

\section{A Property-Preserving Kernel Method for Operator Learning}
\label{sec:methods}
In this section, we present an overview of operator learning (Section \ref{sec:overview}), then present the high-level algorithm used in this work (Section \ref{sec:contrib}), and then describe the mathematical formulation and implementation details of the property-preserving kernel method (Section \ref{sec:spatial_kernel}--\ref{sec:training_details}).
\subsection{Overview of Operator Learning}
\label{sec:overview}
We now briefly describe kernel-based operator learning. Let $\A$ and $\V$ be two separable Banach spaces of functions and let $\G: \A \rightarrow \V$ be an operator that maps functions $a \in \A$ (input functions) to functions $v \in \V$ (output functions), i.e. $\G(a) = v$. The goal of operator learning is to learn an approximation $\Gtilde \approx \G$ given $N$ pairs of input and output functions, $\{(a_i, v_i)\}, i=1,\dots,N$. Let $\Omega_a$ and $\Omega_v$ be the domains for the input and output functions respectively. The choice of these domains is arbitrary; they can be spatial, spatiotemporal, or parametric, for instance. In practice, the functions are provided as a finite number of samples on their respective domains; let $n$ and $m$ denote the number of samples provided for the input and output functions, respectively. Further, in the most general setting and the setting of this work, the functions can be vector-valued so that $a: \Omega_a \rightarrow \mathbb{R}^p$ and $v: \Omega_v \rightarrow \mathbb{R}^d$. In this work, we will focus on a kernel-based parametrization of $\Gtilde$, a generalization of the work by~\citet{batlle2024kernel}.
The kernel method for operator learning (first proposed by~\citet{batlle2024kernel}) formulates $\Gtilde$ as
\begin{align}
v \approx \Gtilde(a) = \chi \circ \ftilde \circ \varphi(a), \label{eq:kernel_method}
\end{align}
where $\chi: \mathbb{R}^{m} \rightarrow \V$ reconstructs $v$ by interpolating its $m$ samples, $\varphi: \A \rightarrow \mathbb{R}^{n}$ generates $n$ samples of $a$, and $\ftilde: \mathbb{R}^{n} \rightarrow \mathbb{R}^{m}$ maps from $n$ samples of $a$ to $m$ samples of $v$. Following~\citet{batlle2024kernel}, we endow both $\A$ and $\V$ with reproducing kernel Hilbert spaces (RKHSs) and choose kernel maps for $\chi$ and $\ftilde$. Specifically, we choose to work with positive-definite radial kernels (also called radial basis functions or RBFs) so that, assuming pairwise distinct interpolation points, the linear systems in~\eqref{eq:spatial_kernel_linearsystem},~\eqref{eq:spacetime_gramian}, and~\eqref{eq:operator_kernel_linearsystem} admit unique interpolants~\citep{chen2005scientific, fasshauer2007meshfree, baxter2010interpolation}.

\subsection{Our contribution}
\label{sec:contrib}
\begin{algorithm}[!htpb]
\caption{Training and evaluation procedure for the property-preserving kernel method.}
\label{alg:full_method}
\begin{algorithmic}[1]
    \Require \\
        \textbf{Training data set}: $N$ pairs of input and output function evaluations, $\{(\ba_i, \bv_i)\}_{i=1}^{N}$. \\
        \textbf{Inference example}: $(\ba^\star, \bv^\star)$ is an unseen pair of input and output function evaluations. \\
        Property-preserving kernel $\Phi$ and operator kernel $\lambda$.
        Ridge parameter $\theta$.
    \Let $\ba \in \mathbb{R}^{np}$ and $\bv \in \mathbb{R}^{md}$ denote the column vectors of the training input and output function evaluations respectively.
    \Let $\bv^{\star} \in \mathbb{R}^{m^{\star} d}$ denote the column vector of the inference output function evaluations at $m^{\star}$ points.
    \Let $\bPhi$ denote the $md \times md$ Gramian matrix associated with $\Phi$, see Section~\ref{sec:spatial_kernel} for details.
    \Let $\bK$ denote the $N \times N$ Gramian matrix associated with $\lambda$, see Section~\ref{sec:operator_kernel} for details.
    \Let $\bK^{\star}$ denote the $1 \times N$ operator kernel evaluation matrix and $\bPhi^{\star}$ denote the $m^{\star}d \times md$ property-preserving kernel evaluation matrix.
    \AlgLabel{Training}
    \State Solve the block linear system ~\eqref{eq:spatial_kernel_linearsystem} for $\bb_i, \ i=1,\dots,N$ using the techniques in Section \ref{sec:efficient_linear_algebra}.\label{state:spatial_solve}
    \State Find $\epsilon$ as the root of the function $\kappa(\bK_{\epsilon}) - 10^{15}$ using the approach outlined in Section~\ref{sec:shape_param}.
    \State Solve the linear system ~\eqref{eq:operator_kernel_linearsystem_ridge_epsilon} for $\bc^j, \ j=1, \dots, md$, using $\epsilon$ and $\theta$.
    \AlgLabel{Generalization}
    \State Obtain output coefficients as $\bb^{\star} \gets \begin{bmatrix} \bK^{\star} \bc^1 \\ \vdots \\ \bK^{\star} \bc^{md} \end{bmatrix} \in \mathbb{R}^{md}$.
    \State Obtain a property-preserving prediction $\tilde{\bv}^\star \gets \bPhi^{\star} \bb^{\star}$.
\end{algorithmic}
\end{algorithm}

Our method emerges from the observation that if $\chi$ is carefully constructed to be a matrix-valued kernel interpolant that \emph{analytically} enforces desirable properties (such as incompressibility, periodicity, and turbulence) when interpolating $v$, then changing $\ftilde$ to output expansion coefficients in that kernel interpolant ensures that all predictions made by $\Gtilde(a)$ for every $a$ will automatically (and analytically) satisfy those properties as well. Our method therefore has two main components; (i) the vector-valued \textbf{operator kernel} interpolant $\ftilde$, and (ii) the vector-valued \textbf{property-preserving kernel} interpolant $\chi$. We summarize the method in Algorithm~\ref{alg:full_method}; describe the two components in detail in Sections~\ref{sec:spatial_kernel} and~\ref{sec:operator_kernel}; and present relevant implementation details in Section~\ref{sec:kernel_implementation}. A schematic of our method is shown in Figure~\ref{fig:schematic}.

\textbf{A note on block vector notation}: We briefly describe the vector notation used throughout this section for input and output function evaluations. Let $v: \mathbb{R}^d \rightarrow \mathbb{R}^d$ be a vector valued function that is evaluated at a set of points $\{y_j\}_{j=1}^{m} \subset \mathbb{R}^d$. We denote these evaluations by the block vector
\begin{align}
\bv = \begin{bmatrix} \bv^1 \\ \vdots \\ \bv^d \end{bmatrix} \in \mathbb{R}^{md},
\end{align}
where each block $\bv^k = \begin{bmatrix}v^k(y_1), \ldots, v^k(y_m) \end{bmatrix}^T \in \mathbb{R}^m$ contains the evaluations of the $k^{\mathrm{th}}$ spatial component of $v$. This notation is helpful when discussing vector-valued kernel interpolants.

\subsection{The property-preserving kernel interpolant}
\label{sec:spatial_kernel}
We now discuss how various properties can be analytically encoded into a kernel basis. For this work, the relevant function spaces $\V$ always consist of incompressible velocity fields which are solutions to PDEs such as~\eqref{eq:ins}. Additionally, depending on the problem, these velocity fields can also exhibit spatial periodicity (due to boundary conditions) and turbulence (due to modeling choices). In our method, we use a property-preserving kernel basis in $\chi$ such that the velocity fields recovered analytically preserve these spatial properties.

Let $Y = \{y_j\}_{j=1}^m \subset \Omega_v \subset \mathbb{R}^d$ be a set of spatial points where an output function $v: \mathbb{R}^d \rightarrow \mathbb{R}^d$ is evaluated; further, let the block vector $\bv = \begin{bmatrix} \bv^1, \ldots, \bv^d \end{bmatrix}^T \in \mathbb{R}^{md}$ denote the function evaluations of $v$ at $Y$. Letting $\Phi: \mathbb{R}^d \times \mathbb{R}^d \rightarrow \mathbb{R}^{d \times d}$ be a matrix-valued positive-definite property-preserving kernel (described in the following subsections), we explicitly write the resulting kernel interpolant as
\begin{align}
\chi(y) = \sum_{j=1}^m \Phi(y, y_j) \ \ub^j,\label{eq:chi_interp}
\end{align}
where $\ub^j \in \mathbb{R}^d$. In order to interpolate $v$, we enforce $\chi(y)= v(y)$ for all $y \in Y$. These interpolation conditions give rise to the linear system
\begin{align}
\bPhi \bb = \bv,\label{eq:spatial_kernel_linearsystem}
\end{align}
where $\bPhi_{ij} = \Phi(y_i, y_j)$ is the $md \times md$ block Gramian matrix arising from evaluations of $\Phi$, and $\bb = \begin{bmatrix} \bb^1, \ldots, \bb^d \end{bmatrix}^T \in \mathbb{R}^{md}$ contains the interpolation coefficient vectors for each spatial component. In the operator learning setting, there are multiple such target functions $v_i$; we denote by $\bb_i \in \mathbb{R}^{md}$ the coefficient vector for the $i^{\mathrm{th}}$ output function, and by $b^j_i$ its $j^{\mathrm{th}}$ component. The system \eqref{eq:spatial_kernel_linearsystem} has a unique solution if the points in $Y$ are distinct~\citep{fasshauer2007meshfree,Wendland2004}. This $md \times md$ system is computationally expensive to solve ($O(d^3m^3)$ operations for a Cholesky factorization, and $O(d^2m^2)$ operations for subsequent solves for each right-hand side). However, we never form and store $\bPhi$, but instead compute only with the individual blocks in $\bPhi$, reducing the solution cost significantly. These details are described in Section \ref{sec:kernel_implementation}. Once the interpolation coefficients in \eqref{eq:chi_interp} are computed, the interpolant can be evaluated at any location (ideally within the hull of $Y$). Let $Y^*$ be a set of $m^*$ evaluation locations. Then, the evaluation of $\chi(y)$ at $Y^*$ can be written as a matrix-vector product
\begin{align}\label{eq:spatial_kernel_eval}
    \left.\chi (y)\right|_{Y^*} = \bPhi^{*} \bb,
\end{align}
where $\bPhi^*(y_k^*,y_j)$, $k=1,\ldots,m^*$, $j=1\ldots,m$ is the $m^*d \times md$ (rectangular) evaluation matrix.

Thus far, our description of these property-preserving matrix-valued kernels has been abstract. In the following subsections, we describe in detail how to simultaneously encode incompressibility, periodicity, and turbulence in $\Phi$.
\subsubsection{Incompressibility}
The primary goal of this work is the high-fidelity surrogate modeling of incompressible fluid flows. To that end, our matrix-valued kernel $\Phi$ and the corresponding property-preserving interpolant \eqref{eq:chi_interp} always enforce incompressibility analytically.

The approximation of divergence-free (incompressible) vector fields has a long history in the numerical methods and approximation literature. Classical divergence-conforming finite element spaces, such as the Raviart--Thomas~\citep{raviartthomas1977}, Brezzi--Douglas--Marini~\citep{brezzi1985two}, and N\'ed\'elec elements~\citep{nedelec1980mixed}, enforce the divergence-free constraint by construction through their basis functions, yielding local and exactly divergence-free numerical solutions. This paradigm established the foundational principle of enforcing physical constraints by restricting the admissible solution space. However, these constructions rely on polynomial interpolation over simplices, thereby necessitating spatial meshes. Subsequent work in the meshfree approximation literature led to the construction of matrix-valued divergence-free kernels from scalar-valued positive-definite kernels~\citep{narcowich1994generalized, narcowich2007divergence, fuselier2008sobolev, fuselier2009error, wendland2009divergence}, even leading to numerical methods for meshfree Helmholtz-Hodge decompositions and incompressible flow solvers~\citep{fuselier2017radial, fuselier2016high, drake2021partition, owhadi2023gaussian}. While Algorithm \ref{alg:full_method} can easily leverage mesh-based divergence-free approximations, we choose to use kernel-based methods as they are (1) meshfree and therefore more general; and more importantly (2) can \emph{simultaneously encode multiple properties in addition to incompressibility}.

We now describe the kernel construction used in this work. Let $\phi: \mathbb{R}^d \times \mathbb{R}^d \rightarrow \mathbb{R}$ be a scalar-valued positive-definite kernel that is at least $C^4(\mathbb{R}^d)$. The associated divergence-free matrix-valued kernel is defined by
\begin{align}
\Phi(x,y) = \nabla_y \times \nabla_x \times \phi(x,y),\label{eq:curlcurl}
\end{align}
where the differential operators act componentwise, and $x=(x_1, \dots, x_d)$ and $y=(y_1, \dots, y_d)$ are points in $\mathbb{R}^d$. For $d=2$, we obtain the following matrix-valued kernel
\begin{align}
\Phi(x,y) =
\begin{bmatrix}
\partial_{x_2} \partial_{y_2} & - \partial_{x_1} \partial_{y_2} \\
- \partial_{x_2} \partial_{y_1} & \partial_{x_1} \partial_{y_1}
\end{bmatrix} \phi(x, y).\label{eq:2ddfblocks}
\end{align}
Whereas for $d=3$ we obtain
\begin{align}
\Phi(x,y) =
\begin{bmatrix}
\partial_{x_3} \partial_{y_3} + \partial_{x_2} \partial_{y_2} & - \partial_{x_1} \partial_{y_2} & - \partial_{x_1} \partial_{y_3} \\
- \partial_{x_2} \partial_{y_1} & \partial_{x_3} \partial_{y_3} + \partial_{x_1} \partial_{y_1} & - \partial_{x_2} \partial_{y_3} \\
- \partial_{x_3} \partial_{y_1} & - \partial_{x_3} \partial_{y_2} & \partial_{x_2} \partial_{y_2} + \partial_{x_1} \partial_{y_1}
\end{bmatrix} \phi(x, y).\label{eq:3ddfblocks}
\end{align}
By construction, each column of $\Phi$ is divergence free with respect to the evaluation coordinate $y$, and the resulting kernel is positive-definite; a simplified form is also readily available for radial kernels~\citep{fuselier2016high}. An important consequence of this construction is that any kernel interpolant of the form~\eqref{eq:chi_interp} lies in a divergence-free space. Consequently, the interpolant $\chi$ satisfies $\nabla \cdot \chi = 0$ pointwise throughout the domain, independently of the sampling locations or interpolation coefficients. This provides an analytically exact enforcement of incompressibility.
\subsubsection{Periodicity}
\label{sec:periodicity}
In addition to incompressibility, many applications require velocity fields to satisfy periodic boundary conditions, typically for enabling the modeling of larger physical domains at lower computational cost but also for modeling truly periodic domain geometries. Given that we began with the divergence-free kernel $\Phi$, this now opens up the issue of additionally incorporating periodicity into $\Phi$.

In the kernel and RBF literature, periodic structure is imposed in settings ranging from interpolation on periodic domains~\citep{jacob2002sampling, fasshauer2007meshfree} to the numerical treatment of periodic boundary conditions~\citep{nguyen2012imposing} and closed curves~\citep{boissonnat2006effective}. Existing approaches fall into three broad categories: Fourier-based constructions, direct kernel and RBF modifications on periodic domains, and embedding-based techniques that enforce periodicity by construction. Fourier spectral and pseudospectral methods are classical tools for approximating periodic functions~\citep{canuto2006spectral, boyd2001chebyshev, fornberg_1996}, while kernel and RBF methods impose periodicity via kernel periodization and lattice-point methods~\citep{xiao2014periodized, cools2019lattice, kaarnioja2022fast}. Embedding-based methods encode periodicity through an explicit embedding of Euclidean coordinates into a compact manifold, most commonly the circle or torus. This idea has been employed in numerical methods for PDEs on the sphere~\citep{flyer2007transport, flyer2009radial}; for elastic surfaces in fluid-structure interaction problems in~\citet{shankar2015augmenting, shankar2015radial, kassen2021immersed}; and for problems posed on tori in~\citet{fuselier2015order, owhadi2023gaussian}.

Adopting the embedding-based approach, we enforce periodicity by composing kernels with an explicitly periodic embedding into a higher-dimensional Euclidean space. With this approach, periodicity is enforced by construction, allowing standard kernel constructions to be used without requiring kernel periodization. Specifically, we define an embedding $h: \mathbb{T}^d \rightarrow \mathbb{R}^{2d}$ as
\begin{align}
h(y) = (\cos(y_1), \sin(y_1), \dots, \cos(y_d), \sin(y_d)).\label{eq:periodic_embedding}
\end{align}
Composition with $h$ induces $2\pi$-periodicity in each spatial coordinate. Let $\phi^\pi: \mathbb{R}^{2d} \times \mathbb{R}^{2d} \rightarrow \mathbb{R}$ be a scalar-valued positive-definite kernel defined on the embedding space, \emph{i.e.}, $\phi^{\pi}(x,y) = \phi(h(x),h(y))$, where $\phi$ is some positive-definite kernel. Applying the standard curl-curl construction from~\eqref{eq:curlcurl} on $\phi^\pi$ gives the \emph{periodic, divergence-free, matrix-valued kernel}
\begin{align}
\Phi^\pi(x,y) = \nabla_y \times \nabla_x \times \phi^\pi(h(x),h(y)),
\end{align}
which inherits periodicity directly from the embedding and is therefore also $2\pi$-periodic in each spatial coordinate; note that this kernel is no longer radial. If periodicity must be imposed in only one coordinate, we modify $h(y)$ to only incorporate the torus embedding into that coordinate.
\subsubsection{Turbulence through power laws}
Finally, we turn our attention to creating reliable surrogates for turbulent incompressible flows. Turbulence in incompressible flow is a high-Reynolds-number regime of the Navier-Stokes equations characterized by nonlinear vortex stretching, broadband energy cascades, and anomalous dissipation~\citep{Frisch1995,Pope2000}. It is ubiquitous in geophysical and engineering flows, including atmospheric and oceanic boundary layers~\citep{Stull1988,Vallis2017}, riverine and open-channel flows~\citep{Nezu2005}, and internal and external aerodynamic flows~\citep{Schlichting2000}. Turbulent flows arise when the Reynolds number $Re = UL/\nu \gg 1$ and exhibit multiscale structure~\citep{Kolmogorov1941,Pope2000}. Modeling therefore requires statistical closures (RANS)~\citep{alfonsi2009reynolds}, large-eddy simulation (LES)~\citep{spalart2000strategies}, or direct numerical simulation (DNS)~\citep{spalart2000strategies} depending on the resolved range of scales \citep{Pope2000}; a range of other approaches also exist~\citep{spalart2000strategies,alam2015hybrid,cambon1999linear,durbin2018some}. In the high--Reynolds-number regime, kinetic energy injected at large scales is transferred to progressively smaller eddies through nonlinear interactions—Richardson’s cascade—until viscous dissipation dominates at the Kolmogorov scale $\eta \sim (\nu^3/\varepsilon)^{1/4}$, yielding the inertial-range spectrum $E(k)\sim \varepsilon^{2/3} k^{-5/3}$ \citep{Richardson1922,Kolmogorov1941,Frisch1995}.

Due to the ubiquity of turbulence in practical applications, we also incorporate turbulence preservation into our property-preserving kernels by leveraging the fact that Kolmogorov’s inertial-range prediction $E(k)\sim \varepsilon^{2/3}k^{-5/3}$ implies that second-order statistics scale as power laws in separation $r$ (e.g., $S_2(r)=\mathbb{E}|u(x+r)-u(x)|^2\sim (\varepsilon r)^{2/3}$), so that velocity fields exhibit self-similar, scale-invariant correlations~\citep{Kolmogorov1941,Frisch1995}. More specifically, we follow~\citet[Section 5.1]{owhadi2023gaussian} and define an additive multiscale kernel $\Phi^{\eta}$ of the form
\begin{align}
\Phi^\eta = \sum_{s=1}^q \alpha_s \Phi_s,
\end{align}
where $\sigma_s = \frac{\sigma_0}{2^s}$ is the shape parameter for $\Phi_s$ (see Section~\ref{sec:kernel_implementation} for details on how $\sigma_0$ is picked), $\alpha_s = \sigma_s^\gamma$, and
\begin{align}
\gamma = \begin{cases}
4, & \mathrm{if } \ d = 2, \\
\frac{2}{3} + 2, & \mathrm{if } \ d = 3.
\end{cases}
\end{align}
We set $q=5$ for all problems involving turbulence. Here, $q$ denotes the number of modes in the additive kernel, with each mode associated with a characteristic length scale corresponding to eddies at that scale~\citep{owhadi2021kernel}. In general, increasing $q$ improves the approximation accuracy of the multiscale kernel~\citep[Section 7]{owhadi2023gaussian}. The key idea here is that $\Phi^{\eta}$ can analytically mimic the Richardson cascade; in the continuum limit, this corresponds to representing $\Phi^{\eta}$ as a scale mixture whose spectral density matches the Kolmogorov target.

\textbf{Simultaneous property-preservation:} Though we described $\Phi$, $\Phi^{\pi}$, and $\Phi^{\eta}$ separately, our approach allows for the simultaneous property-preserving kernel $\Phi^{\pi,\eta}$, which is  matrix-valued, positive-definite, and more importantly, analytically divergence-free, analytically periodic, and analytically approximating the Richardson Cascade through power laws. We use this kernel within the property-preserving interpolant for turbulent, periodic, incompressible flows. In general, our work involves utilizing known problem features and the analytic enforcement of these features in our surrogates. We defer \emph{detection} of these features to future work. 
\subsubsection{Spacetime kernel interpolation}
In problems where $\Omega_v$ is a spacetime domain, the interpolant $\chi$ must now interpolate spacetime data. As an example, let $\Omega_v = \Omega \times \Gamma$, where $\Omega$ denotes the spatial domain and $\Gamma$ denotes the temporal domain. Consider a function $v(x,t)$, where $x \in \Omega \subset\mathbb{R}^d$ and $t \in \Gamma \subset \mathbb{R}$. Assume that we are provided with a set of spatial snapshots of $v$ on $\Omega$ at discrete time instances $\{t_1, \dots, t_T\}$ as part of our training data.  We now describe how to incorporate the ability to produce such spacetime predictions (upon generalization) into our surrogate.

We draw from the literature on spacetime kernel design, which employs products of spatial and temporal kernels~\citep{posa1993simple, cressie01121999, myers2002space, stein01032005, romero2017kernel, jing2024modified, yue2019novel, ku2020space, ku2022space,li2016estimation, li2011global}. Let $\psi: \mathbb{R} \times \mathbb{R} \rightarrow \mathbb{R}$ be a scalar-valued positive-definite kernel; we select $\psi = \phi$, but other choices are possible. We then define a spacetime interpolant using the product kernel $\Phi \psi$; this interpolant takes the form
\begin{align}
\chi(y, t) = \sum_{j=1}^{mT} \bigl(\Phi(y, y_j) \psi(t, t_j)\bigr) \ \ub^j,
\label{eq:chi_interp_spacetime}
\end{align}
where $(y_j, t_j)$ ranges over all space-time evaluations and the coefficients. As in the spatial case, we enforce the interpolation conditions $\chi(y,t) = v(y,t)$ at all $y \in Y$ and $t \in \{t_1, \dots, t_T\}$, thereby obtaining a block linear system of the form
\begin{align}
\begin{bmatrix}
\bPhi \ \psi_{1,1} & \bPhi \ \psi_{1,2} & \cdots & \bPhi \ \psi_{1,T} \\
\bPhi \ \psi_{2,1} & \bPhi \ \psi_{2,2} & \cdots & \bPhi \ \psi_{2,T} \\
\vdots & \vdots & \ddots & \vdots \\
\bPhi \ \psi_{T,1} & \bPhi \ \psi_{T,2} & \cdots & \bPhi \ \psi_{T,T}
\end{bmatrix}
\begin{bmatrix}
\bb^1 \\
\bb^2 \\
\vdots \\
\bb^T
\end{bmatrix}
=
\begin{bmatrix}
\bv^1 \\
\bv^2 \\
\vdots \\
\bv^T
\end{bmatrix},\label{eq:spacetime_gramian}
\end{align}
where $\psi_{ij} = \psi(t_i, t_j)$, $\bPhi$ is the Gramian matrix obtained by evaluating the matrix-valued spatial kernel $\Phi$ on $Y$, and the superscripts on $\bb$ and $\bv$ denote the timestep. Since $\Phi$ and $\psi$ are both positive-definite kernels and the spatial and temporal sampling locations are distinct, the kernel interpolant $\chi$ is uniquely determined~\citep{fasshauer2007meshfree,Wendland2004}. It is important to note that $\Phi$ here can be replaced by $\Phi^{\pi}$ or $\Phi^{\pi,\eta}$, as needed. Much as in the case of \eqref{eq:spatial_kernel_linearsystem}, \eqref{eq:spacetime_gramian} requires the use of efficient linear algebra to contend with the $d+1$ dimensionality of the problem. We describe these techniques in Section \ref{sec:efficient_linear_algebra}.

%
\subsubsection{Approximate kernel Fekete points}
\label{sec:fekete_points}
\begin{figure}[pos=!b]
    \centering
    \begin{subfigure}{0.35\linewidth}
        \centering
        \includegraphics[width=\linewidth]{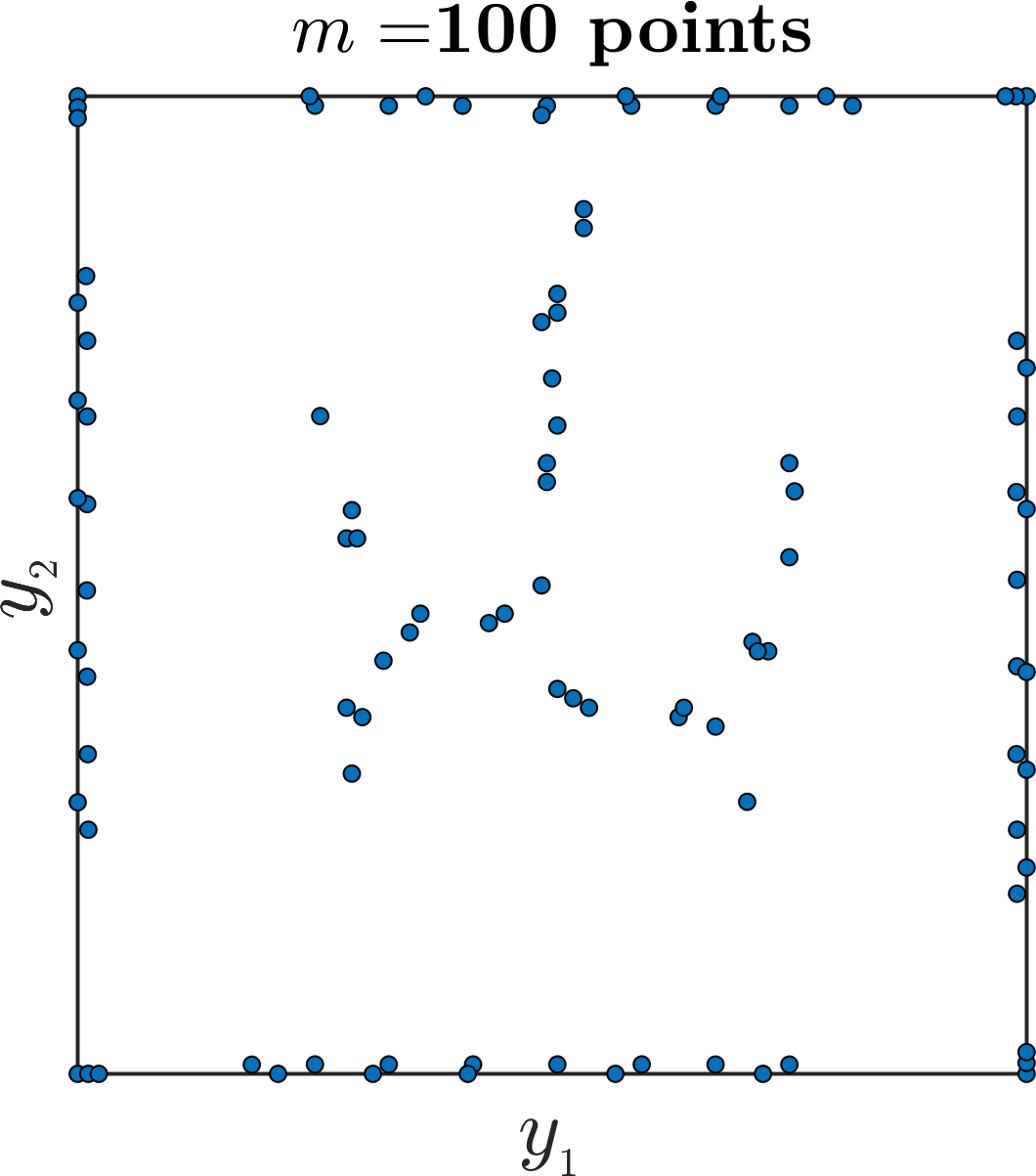}
        \caption{}
    \end{subfigure}
    \begin{subfigure}{0.35\linewidth}
        \centering
        \includegraphics[width=\linewidth]{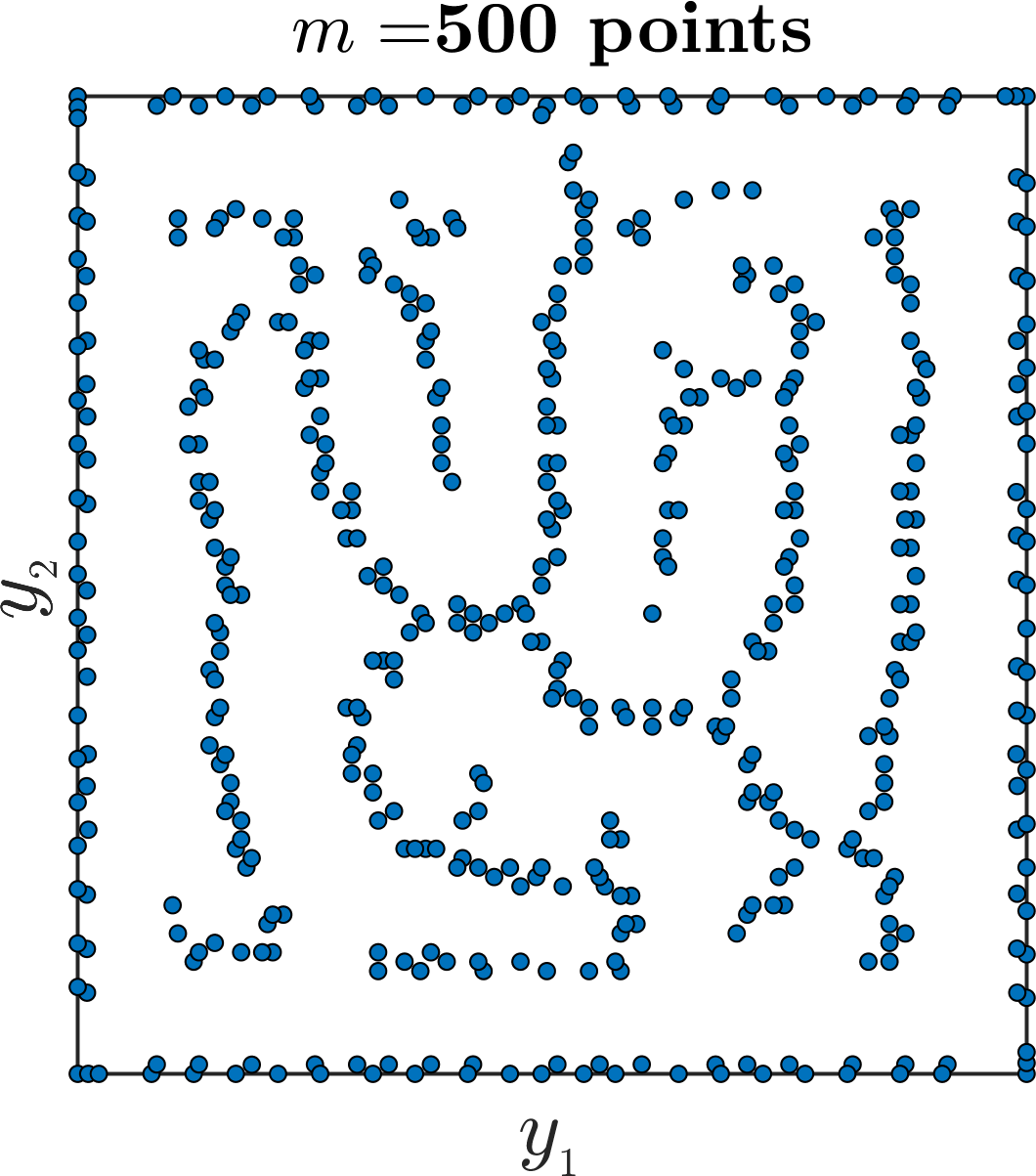}
        \caption{}
    \end{subfigure} \\
    \begin{subfigure}{0.35\linewidth}
        \centering
        \includegraphics[width=\linewidth]{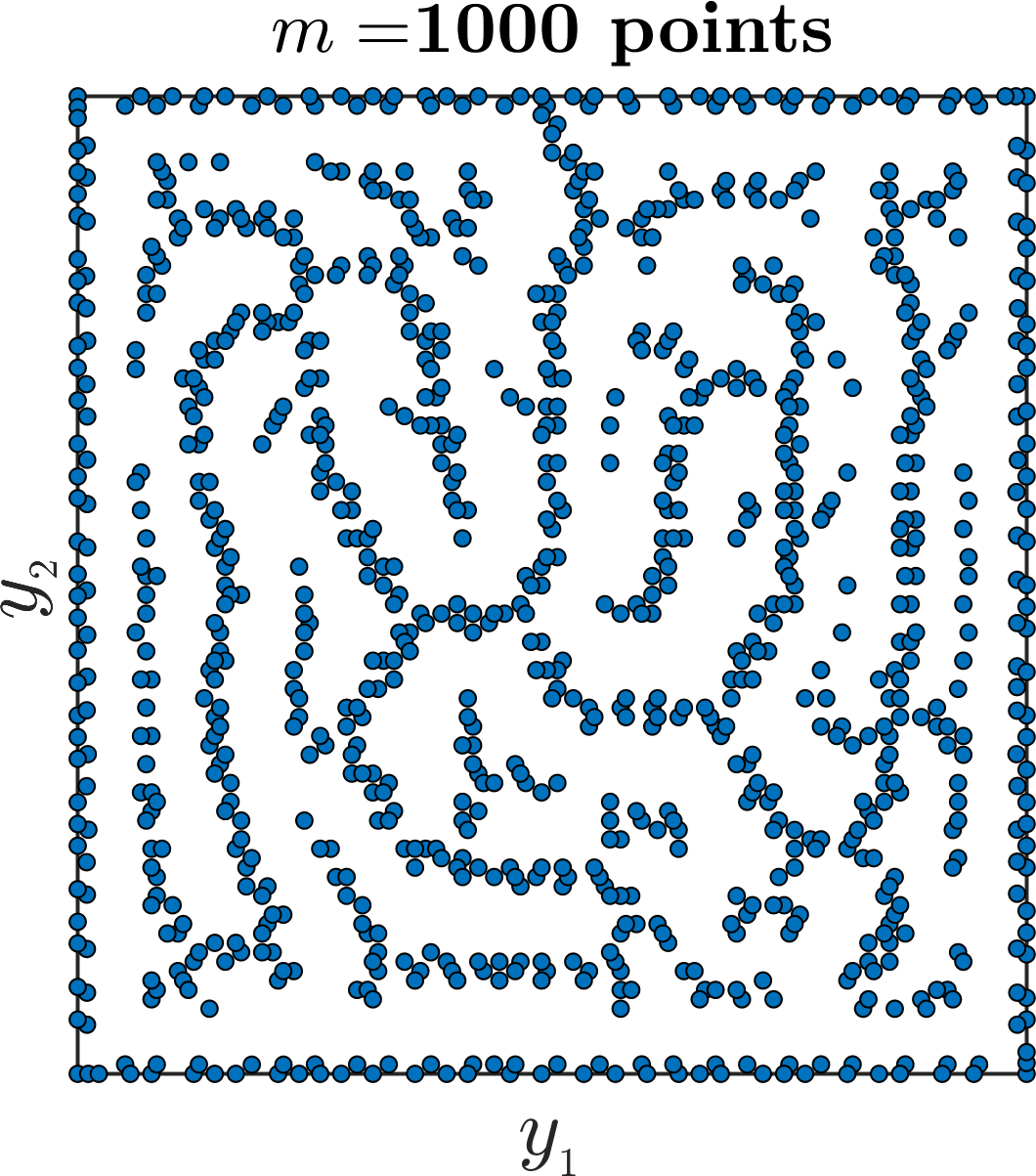}
        \caption{}
    \end{subfigure}
    \caption{(\textbf{A}) $m=100$, (\textbf{B}) $m=500$, and (\textbf{C}) $m=1000$ approximate kernel Fekete points in the domain $\Omega =[0, 1]^2$.}
    \label{fig:fekete_points}
\end{figure}
The global kernel interpolation performed throughout this work in~\eqref{eq:chi_interp} and~\eqref{eq:chi_interp_spacetime} is susceptible to the Runge phenomenon, a well known instability affecting global interpolation of smooth target functions~\citep{Runge1901, epperson1987runge}. This issue has been observed both for global polynomial interpolation and for RBF methods~\citep{platte2005polynomials, fornberg2007runge}. Since the velocity fields considered here are smooth and defined on bounded domains, naive global kernel interpolation on equally-spaced or quasi-uniform points may exhibit large oscillations near the domain boundary.

Most work in the literature addresses this difficulty through careful choice of interpolation points or bases. Approaches include boundary-clustered points~\citep{berrut2004barycentric}, spatially varying shape parameters in RBF interpolation~\citep{fornberg2007runge}, and extension-based techniques such as RBF extension~\citep{boyd2010six}, which are closely related to Fourier extension methods~\citep{boyd2002comparison, boyd2009exponentially}. Within the RBF literature, greedy point selection algorithms~\citep{schaback2000adaptive, de2005near} and inducing-point strategies~\citep{galy2021adaptive} are commonly employed to identify subsets of points that give improved stability and approximation accuracy.

A complementary and well-established strategy in polynomial interpolation is the use of Fekete points~\citep{sommariva2009computing, bos2008calculation}. These points are defined as the maximizers of the absolute value of the determinant of the polynomial Vandermonde matrix~\citep{de2012polynomial}, a property that is closely connected to favorable conditioning and reduced Lebesgue constants, and hence improved interpolation stability. However, exact Fekete points are expensive to compute and are known analytically only in very restricted settings. In practice, \emph{approximate} Fekete points are typically computed via rank-revealing, column-pivoted QR factorizations of Vandermonde matrices~\citep{sommariva2009computing, sommariva2010approximate}. Recent work has begun extending these ideas to kernel Gramian matrices in the univariate setting~\citep{karvonen2021kernel}. Motivated by these developments, we adopt an analogous approach for kernel-based interpolation by computing approximate Fekete points directly from our kernel Gramian matrices. We describe this approach below.

Let $Y$ be the set of candidate points from which we wish to pick approximate Fekete points; typically, this is the set of all data locations available to us. Let $\phi$ be a scalar-valued positive-definite kernel and let $\mathbf{A}_{ij} = \phi(y_i, y_j)$ be its associated Gramian matrix on $Y$. Since $\mathbf{A}$ is positive-definite, it admits the Cholesky factorization $\mathbf{A} = L L^T$ where $L$ is lower triangular. The columns of $L$ can be interpreted as a set of feature vectors induced by the kernel $\phi$ at the points $Y$. We seek to find a subset of points whose feature vectors are maximally linearly independent. To this end, we apply a rank revealing QR factorization with column pivoting to $L$
\begin{align}
L P = Q R,
\end{align}
where $Q$ has orthogonal columns, $R$ is upper triangular, and $P$ is a permutation matrix containing column pivot indices. The column pivoting orders the columns of $L$ according to their contribution to the numerical rank. We take the points corresponding to the first $m$ entries of $P$ as the \emph{approximate kernel Fekete points}. These points are used by the property-preserving kernel in all experiments. Figure~\ref{fig:fekete_points} shows example figures of $m=100$, $m=500$, and $m=1000$ approximate kernel Fekete points in the domain $\Omega =[0, 1]^2$. The figures show very clearly the greater relative importance the algorithm puts on the points closer to the boundary, which helps alleviate the Runge phenomenon and improves the regularity of the interpolation coefficients. In cases where $\Omega_a = \Omega_v$, the same set of points is used for both input and output fields; $n = m$.
\subsection{The operator kernel interpolant}
\label{sec:operator_kernel}
Once the interpolation coefficients in \eqref{eq:chi_interp} or \eqref{eq:chi_interp_spacetime} are obtained, the next step in our framework is to treat those coefficients as a function of the input functions and fit them in turn (via the object $\ftilde$ mentioned previously); this enables the prediction of new coefficients on generalization (when an unseen input function is supplied). We now describe the construction of $\ftilde$, which is an operator-valued kernel interpolant.
Let $a: \mathbb{R}^d \rightarrow \mathbb{R}^p$ be an input function and $X = \{x_1, \dots, x_n\} \subset \Omega_a \subset \mathbb{R}^d$ the set of points on its domain where it is sampled; we denote the block vector of evaluations of $a$ at $X$ by $\ba = \begin{bmatrix} \ba^1, \ldots, \ba^p \end{bmatrix}^T \in \mathbb{R}^{np}$. The operator-valued kernel interpolant (or simply operator kernel interpolant) requires a matrix-valued positive-definite kernel $\Lambda: \mathbb{R}^{np} \times \mathbb{R}^{np} \rightarrow \mathbb{R}^{md \times md}$ that is capable of consuming these $\mathbb{R}^{np}$ block vectors as arguments and outputting a matrix in $\mathbb{R}^{md \times md}$. Equipped with this matrix-valued kernel, we may now write the operator kernel interpolant $\ftilde$ as
\begin{align}
\ftilde(\ba) = \sum_{i=1}^N \Lambda(\ba, \ba_i) \ \uc_i,\label{eq:f_interp}
\end{align}
In order to reduce the dimensionality of the associated interpolation problem, we first select $\Lambda$ to be a diagonal kernel of the form $\Lambda = \lambda I$, where $I$ is the $md\times md$ identity matrix and $\lambda: \mathbb{R}^{np} \times \mathbb{R}^{np} \to \mathbb{R}$ is some scalar-valued positive definite kernel. This choice decouples all solves for the components of the coefficients $\uc_i$. After making this choice, we enforce the interpolation conditions $\ftilde(\ba_i) = \bb_i$, where $i = 1,\ldots,N$ indexes training functions and $\bb_i$ is the coefficient vector associated with the property-preserving interpolant $\chi$ for the $i$-th output function. This results in the following linear system with $N$ right-hand sides and the shared Gramian of $\lambda$:
\begin{align}\label{eq:operator_kernel_linearsystem}
\underbrace{\begin{bmatrix}
\lambda_{1,1} & \lambda_{1,2} & \ldots & \lambda_{1,N} \\
\lambda_{2,1} & \lambda_{2,2} & \ldots & \lambda_{2,N} \\
\vdots & \vdots & \ddots & \vdots \\
\lambda_{N,1} & \lambda_{N,2} & \ldots & \lambda_{N,N}
\end{bmatrix}}_{\bK}
\underbrace{\begin{bmatrix}
c^1_1 & c^2_1 & \ldots & c^{md}_1 \\
c^1_2 & c^2_2 & \ldots & c^{md}_2 \\
\vdots & \vdots & \ddots & \vdots\\
c^1_N & c^2_N & \ldots & c^{md}_N\\
\end{bmatrix}}_{C} =
\underbrace{\begin{bmatrix}
b^1_1 & b^2_1 & \ldots & b^{md}_1 \\
b^1_2 & b^2_2 & \ldots & b^{md}_2 \\
\vdots & \vdots & \ddots & \vdots\\
b^1_N & b^2_N & \ldots & b^{md}_N\\
\end{bmatrix}}_{B}
\end{align}
The system \eqref{eq:operator_kernel_linearsystem} has a unique solution if the input function sample vectors are all distinct~\citep{fasshauer2007meshfree,Wendland2004,batlle2024kernel}.
To account for possible irregularities in the input function samples, we include a ridge parameter $\theta$ on the Gramian's diagonal for regularization~\citep{batlle2024kernel, hastie2009elements}. This ensures positive-definiteness and serves to control the magnitude of the columns of $C$. This results in the following linear system:
\begin{align}
\bigl(\bK_\epsilon + \theta I\bigr) \ \bc^j = \bb^j, \ j=1, \dots, md,\label{eq:operator_kernel_linearsystem_ridge_epsilon}
\end{align}
where $\bc^j$ are the columns of $C$, $\bb^j$ are the columns of $B$, $\epsilon$ is the shape parameter of the operator kernel, and $I$ is the $N \times N$ identity matrix.

\textbf{Generalization:} Analogous to the evaluation of the property-preserving kernel interpolant in \eqref{eq:spatial_kernel_eval}, $\ftilde$ can be evaluated at an unseen test function $\ba^{\star}$ to obtain the property-preserving kernel interpolation coefficients for the corresponding output function $\bv^{\star}$. Denote by $\bK^*$ the $1\times N$ evaluation matrix for the kernel $\lambda$ so that $(\bK^*)_{1j} = \lambda(\ba^*, \ba_j)$, $j=1,\ldots,N$. Then,
\begin{align}\label{eq:operator_kernel_eval}
    \ftilde(\ba^*) = \bK^* C.
\end{align}
For $N^*$ test functions, $\bK^*$ is simply an $N^* \times N$ matrix. However, note that $\ftilde(\ba^*)$ merely returns an approximation to the output coefficient vector $\bb^*$ (in a property-preserving kernel expansion) to the output function $\bv^*$. Samples of the output function $\bv^*$ are in turn obtained by \eqref{eq:spatial_kernel_eval}.

\subsubsection{Operator-valued Gaussian processes for uncertainty quantification}
\label{sec:uq}
The operator kernel interpolant map also provides us with a natural method for estimating function-level uncertainty. Uncertainty quantification involves calculating the statistical uncertainty associated with a surrogate model's prediction~\citep{sudret2017surrogate}. The applications of this field vary widely, such as aerodynamic shape design~\citep{asouti2023radial}, materials science~\citep{kovachki2022multiscale}, inverse problems~\citep{martin2012stochastic, huang2022iterated}, and seismic wave propagation~\citep{martin2012stochastic}. In particular, kernel interpolants (regularized or otherwise) automatically form the mean function for a Gaussian process (GP); alternatively, GPs may be thought of as a straightforward probabilistic generalization of kernel interpolation. GP methods have been used to approximate uncertainty in machine learning model predictions~\citep{williams2006gaussian, zhu2018bayesian}. GPs have also been used in kernel based methods~\citep{tharakan2015scalable, csaji2019distribution, narayan2021optimal, mora2025operator}.

While not a major focus of this work, we briefly describe how to leverage the operator kernel interpolant $\ftilde$ to quantify uncertainty. 
The GP naturally induced by the operator kernel is given by $\xi \sim \mathcal{N}(\widehat{\bb}, \Lambda)$, where $\widehat{\bb}$ is the prior mean and $\Lambda$ the prior covariance kernel of the GP (and also the kernel of the interpolant \eqref{eq:f_interp}. Conditioning the GP on the training observations yields a posterior mean and posterior covariance kernel $\Sigma^{\star}$ given by
\begin{align}
\overline{\bb} &= \bLambda^\star \bigl(\bLambda + \theta I\bigr)^{-1} \bb_{\mathrm{train}}, \label{eq:gp_posterior_mean} \\
\Sigma^\star &= \bLambda^{\star\star} - \bLambda^\star \bigl(\bLambda + \theta I\bigr)^{-1} (\bLambda^\star)^T, \label{eq:gp_posterior_covariance}
\end{align}
where $\bb_{\mathrm{train}} = \begin{bmatrix} \bb_1 \\ \vdots \\ \bb_N \end{bmatrix} \in \mathbb{R}^{Nmd}$, $\bLambda$ is the $Nmd \times Nmd$ matrix denoting the prior covariance on the training data set, $\bLambda^\star$ is the $md \times Nmd$ cross-covariance matrix, and $\bLambda^{\star \star} = \Lambda(\ba^{\star}, \ba^{\star})$ is the $md \times md$ matrix denoting the prior covariance on the test data set. The uncertainty in the predicted output field can be obtained by perturbing the posterior mean using the Cholesky factorization of the posterior covariance. Specifically, if $\Sigma^\star = L L^T$, then samples from the predictive distribution within one standard deviation are given by
\begin{align}
\widetilde{\bb}^\star = \overline{\bb} \pm L z, \ z \sim \mathcal{N}(0, I),\label{eq:predictive_posterior_coeffs}
\end{align}
where $z \in \mathbb{R}^{md}$. The perturbed coefficient vector $\widetilde{\bb}^{\star}$ is then used in the property-preserving kernel interpolant in~\eqref{eq:spatial_kernel_linearsystem} to obtain a sample velocity field from the posterior distribution (via \eqref{eq:spatial_kernel_eval}). Hence, our method provides a framework for \emph{property-preserving uncertainty quantification}. While we leave a detailed exploration of the GP associated with our operator kernel interpolant for future work, we present preliminary results in Section \ref{sec:gp_results} for a turbulent flow problem in 3D.

%
\subsection{Implementation and training details}
\label{sec:kernel_implementation}
We now describe important implementation details for the two kernel interpolants used in our framework.
\subsubsection{Kernels, shape parameters, and nuggets}
\label{sec:common_details}
\paragraph{Kernels} We use the Euclidean distance metric in all kernel arguments; $\phi(x, y) = \phi(\|x-y\|_2)$, thereby making all non-periodic kernels radial. In this special case, the divergence-free kernel can be written as $\Phi(x, y) = (- \Delta I + \nabla \nabla^T) \ \phi(\|x-y\|_2)$ where $\Delta$ and $\nabla$ are the $\mathbb{R}^d$ Laplacian and gradient operators respectively acting on $y$, and $I$ is a $d \times d$ identity matrix. This does not work in the periodic case because the transformation $h$ (see~\eqref{eq:periodic_embedding}) breaks the radial characteristic of $\phi$.

We tested three different choices of radial kernels (RBFs): (i) the Gaussian kernel, (ii) the compactly supported $C^4$ Wendland kernel~\citep{wendland1995piecewise}, and (iii) the $C^4$ Mat\'ern kernel~\citep{demarchi2018lectures}. The $C^4$ Mat\'ern kernel outperformed the other two kernels consistently across all experiments (the Gaussian was more ill-conditioned and the Wendland kernel was less accurate when supports were small). We therefore report results using only this kernel. It is given by $\phi_{\epsilon}(r) = e^{-(r\epsilon)} (3 + 3 r \epsilon + (r \epsilon)^2)$, where $r = \|x-y\|_2$. Specifically, \textbf{the kernels $\phi$, $\psi$, and $\lambda$ are all chosen to be the $C^4$ Mat\'ern kernel}. In Section \ref{sec:approx-thms}, we present universal approximation results and theoretical convergence rates for the Mat\'ern kernel.

\paragraph{Shape parameters}
\label{sec:shape_param}
Like most RBFs, the Mat\'ern kernel comes equipped with a \emph{shape parameter} $\epsilon$ which enormously affects accuracy and conditioning of the corresponding kernel interpolants; in our case, we have three \emph{different} shape parameters that must be selected (one corresponding to $\phi$, $\psi$, and $\lambda$ each). We empirically found that selecting these shape parameters to be as small as possible produced the most accurate generalization results; this corresponds to the so-called flat limit of the kernel. We hypothesize that this was because our training function fields output by the SU2 solver (a finite-volume solver) are naturally consistent with a finite-order Sobolev regularity class~\citep{leveque2002finite}.  For fixed smoothness $\nu$, a Mat\'ern kernel induces a Sobolev (Bessel-potential) native space, and as the shape parameter goes to zero, its reproducing kernel Hilbert space (RKHS) norm approaches the homogeneous $\dot H^{m}$ seminorm ($m=\nu+d/2$), so the interpolant converges to the polyharmonic (Duchon) spline that minimizes that Sobolev semi-norm among all interpolants~\citep{Stein1999,duchon1976splinesmr}; thus, using a near-flat Mat\'ern kernel potentially aligns the operator kernel interpolant with the natural energy class of the finite-volume data.

Algorithms for selecting shape parameters have been studied extensively in the literature. One class of methods depends on the target functions, such as the leave-one-out cross-validation (LOOCV) method~\citep{fasshauer2007choosing, cavoretto2024comparing}, which is closely related to Rippa's algorithm~\citep{rippa1999algorithm, fasshauer2007choosing, ling2022stochastic}. Another class of methods involves optimizing the shape parameter for a target condition number~\citep{wang2002optimal, chen2023selection, krowiak2019choosing, flyer2011radial, shankar2014radial}. For convenience and implementation simplicity, we leverage the approach in~\cite{shankar2014radial} Specifically, given a kernel Gramian $\mathbf{A}_{\epsilon}$ and a target condition number $\delta$, we find the root of the function $s(\epsilon) = 1/\kappa(\mathbf{A}_\epsilon) - \delta$, where $1/\kappa$ is reciprocal 1-norm condition number of $\mathbf{A}_{\epsilon}$ (which is relatively inexpensive to estimate in comparison to the true condition number). More specifically, we find shape parameters by rootfinding using the popular Brent--Dekker method~\citep{brent1973algorithms, dekker1969finding} (as implemented in Matlab's \texttt{fzero} function). 
We consistently selected a target condition number of $10^{12}$ for the property-preserving kernel Gramian (or the product kernel in the spacetime case) and $10^{15}$ for the operator kernel Gramian; we were able to use a larger target condition number for the latter as it also included a ridge regularization term (aka nugget).

\paragraph{Nuggets}
Finally, following standard practice in RBF interpolation of noisy data~\citep{williams2006gaussian, fasshauer2015kernel}, we also regularized the operator kernel by adding a small ``nugget'', commonly known as the ridge parameter, $\theta$ to the diagonal of the Gramian  in~\eqref{eq:operator_kernel_linearsystem_ridge_epsilon}. Interestingly, we found that adding a ridge parameter to the property-preserving kernel degraded the overall generalization accuracy, whereas it was crucial for the operator kernel. We show results on the effect of varying $\theta$ in Section~\ref{sec:nugget_effect} and discuss these results in Section~\ref{sec:discussion}.
\subsubsection{Efficient linear algebra}
\label{sec:efficient_linear_algebra}
Each interpolant carries its own implementation challenges: the property-preserving interpolant brought with it computational cost issues due to the $O(d^3 m^3)$ cost of the naive linear solve, especially for the large $m$ explored in this work; the operator interpolant brings challenges in terms of GPU memory, and the inability to store all $N$ training functions on GPU memory for large $N$. We now describe our specialized implementation techniques.
%
\paragraph{Property-preserving kernel interpolant}
We now describe our approach to overcome the $O(d^3 m^3)$ cost for finding the property-preserving interpolation coefficients, a necessary preprocessing step to obtaining the training data for the operator kernel interpolant. Our approach is centered around careful and recursive use of the Schur complement of the block matrix $\bPhi$ (the property-preserving kernel Gramian). In the discussion below, we will focus on $\bPhi$; the same techniques were used for the other variants.

 We first consider the case when $d=2$. Here, $\Phi$ is a $2 \times 2$ matrix-valued kernel. When evaluated at $m$ points, its Gramian is a matrix of size $2m \times 2m$. The Gramian naturally separates into four blocks, but each block acts on both vector components of the coefficients; recall that each component of the coefficients corresponds to a different spatial dimension (in our case coordinate axes). Following~\citet{fuselier2016high}, we first reorder the individual blocks \emph{spatially}; evaluations for each spatial component are grouped together and the resulting Gramian has the form
\begin{align}
\bPhi =
\begin{bmatrix}
A_{11} & A_{12} \\
A_{21} & A_{22}
\end{bmatrix},
\; A_{ij} \in \mathbb{R}^{m \times m},
\end{align}
where, for notational convenience, $A_{ij}$ is the $(i, j)$ block in~\eqref{eq:2ddfblocks} evaluated at all $m$ points. Then, the block linear system for interpolation is given by
\begin{align}
\begin{bmatrix}
A_{11} & A_{12} \\
A_{21} & A_{22}
\end{bmatrix} \begin{bmatrix}
\bb^1 \\
\bb^2
\end{bmatrix} =
\begin{bmatrix}
\bv^1 \\
\bv^2
\end{bmatrix}.
\end{align}
For the symmetric kernels $\phi$ used in this work, the mixed partial derivatives commute w.r.t. $x$ and $y$. As a result, $\Phi$ exhibits symmetry in its block structure and $A_{12} = A_{21}$. This symmetry allows us to store only the upper (or lower) triangular blocks. Rather than forming and factorizing the full $2m \times 2m$ matrix, we apply a Schur complement reduction with respect to $A_{22}$. Defining
\begin{align}
S = A_{11} - A_{12} A_{22}^{-1} A_{12},\label{eq:2d_schur}
\end{align}
the coefficients are obtained with the following two steps:
\begin{align}
\bb^1 &= S^{-1} \left(\bv^1 - A_{12} A_{22}^{-1} \bv^2\right), \label{eq:2d_bb1} \\
\bb^2 &= A_{22}^{-1} \left(\bv^2 - A_{12} \bb^1 \right).\label{eq:2d_bb2}
\end{align}
This avoids explicitly forming and storing the full Gramian and reduces the memory usage from $\mathcal{O}(m^2 d^2)$ to $\mathcal{O}(m^2)$. Additionally, since $A_{22}$ is a positive-definite matrix, we apply a Cholesky factorization $A_{22} = L L^T$, and then perform forward and backward substitution solves using the lower triangular matrix $L$ in~\eqref{eq:2d_schur},~\eqref{eq:2d_bb1}, and~\eqref{eq:2d_bb2}. This reduces the overall runtime complexity from $\mathcal{O}(m^3 d^3)$ to $\mathcal{O}(m^3)$.

In $d=3$, after a similar reorder by spatial coordinates, the symmetric block linear system takes the form
\begin{align}
\left[
\begin{array}{cc|c}
A_{11} & A_{12} & A_{13} \\
A_{21} & A_{22} & A_{23} \\ \hline
A_{31} & A_{32} & A_{33}
\end{array}
\right]
\begin{bmatrix}
\bb^1 \\
\bb^2 \\ \hline
\bb^3
\end{bmatrix}
=
\begin{bmatrix}
\bv^1 \\
\bv^2 \\ \hline
\bv^3
\end{bmatrix},
\end{align}
where $A_{ij}$ refers to the $(i,j)$ block in~\eqref{eq:3ddfblocks}, $A_{ij}=A_{ji}$, and the horizontal and vertical bars indicate a partition of the Gramian into a $2 \times 2$ system. We apply the same Schur complement strategy recursively. First, we eliminate $\bb^3$ via the Schur complement of $\Phi_{33}$, giving the following $2 \times 2$ block system
\begin{align}
\begin{bmatrix}
\tilde{A}_{11} & \tilde{A}_{12} \\
\tilde{A}_{21} & \tilde{A}_{22}
\end{bmatrix}
\begin{bmatrix}
\bb^1 \\
\bb^2
\end{bmatrix}
=
\begin{bmatrix}
\tilde{\bv}^1 \\
\tilde{\bv}^2
\end{bmatrix},\label{eq:3d_reduced}
\end{align}
where
\begin{align*}
\tilde{A}_{11} = A_{11} - A_{13} A_{33}^{-1} A_{13}, & \
\tilde{A}_{12} = A_{12} - A_{13} A_{33}^{-1} A_{23},\\
\tilde{A}_{22} = A_{22} - A_{23} A_{33}^{-1} A_{23}, & \
\tilde{\bv}^1 = \bv^1 - A_{13} A_{33}^{-1} \bv^3,\\
\tilde{\bv}^2 = \bv^2 - &A_{23} A_{33}^{-1} \bv^3.
\end{align*}
\eqref{eq:3d_reduced} is then solved using the same Schur complement procedure described previously, after which $\bb^3$ is recovered as
\begin{align}
\bb^3 = A_{33}^{-1} \left(\bv^3 - A_{13} \bb^1 - A_{23} \bb^2 \right).
\end{align}
We once again obtain an $O(m^3)$ cost for this recursive solution and require $O(m^2)$ storage. It is also likely that these systems are better conditioned in practice than the original full Gramian. We also leverage a version of these algorithms at inference/generalization time: we only store blocks of the evaluation matrix $\bPhi^*$ and apply them component-wise to the predicted coefficient (block) vectors.

\textbf{Extension to spacetime kernels}: In the spacetime experiment in Section~\ref{sec:taylor_green_spacetime}, the output consists of velocity fields at four discrete timesteps, resulting in \eqref{eq:spacetime_gramian} having a $4 \times 4$ block Gramian. We recursively apply the same Schur complement strategy to reduce the system to first a $3 \times 3$ system, then a $2\times 2$ system in turn.
\paragraph{Operator kernel interpolant}
The main challenge with the operator kernel interpolant was in forming the $N \times N$ Gramian matrix $\bK_{ij} = \lambda(\ba_i, \ba_j)$ where $\ba \in \mathbb{R}^{np}$. Much of our computation is done with the JAX library~\citep{jax2018github} in Python which provides a useful vectorizable map, \texttt{vmap}, that can be customized to execute instructions in parallel on the GPU. Fortunately, forming the entries of $\bK$ is an embarrassingly parallelizable task. However, it is not possible to use \texttt{vmap} to form the entire matrix in one call due to the cost of storing all the input functions in memory simultaneously, even on a moderately powerful desktop GPU (NVIDIA RTX 4080); \texttt{vmap} also appears to incur a significant memory overhead of its own, as do JAX's just-in-time (JIT) compilation features. In our experiments, $n$ and $N$ can each be as large as $10,000$. Consequently, we required an efficient streaming method to compute the entries of $\bK$. We used Algorithm~\ref{alg:streaming_gramian} to form $B \times B$ blocks of $\bK$ in batches of size $B$.
\begin{algorithm}
\caption{Streaming construction of the operator kernel Gramian}
\label{alg:streaming_gramian}
\begin{algorithmic}[1]
\Require Evaluations of the input functions $\{\ba_i\}_{i=1}^N$, operator kernel $\lambda$, and batch size $B$.
\State Initialize $\bK = \mathbf{0}_{N \times N}$
\State $P = \lceil N / B \rceil$
\For{$i = 1$ to $P$}
    \State $\texttt{I}_{\mathrm{start}} \gets iB$
    \State $\texttt{I}_{\mathrm{end}} \gets \min((i+1) B, \ N)$
    \For{$j = 1$ to $P$}
        \State $\texttt{J}_{\mathrm{start}} \gets jB$
        \State $\texttt{J}_{\mathrm{end}} \gets \min((j+1) B, \ N)$
        \State $\bK_{\texttt{I}_{\mathrm{start}}:\texttt{I}_{\mathrm{end}}, \texttt{J}_{\mathrm{start}}:\texttt{J}_{\mathrm{end}}} = \lambda\left(\ba_{\texttt{I}_{\mathrm{start}}:\texttt{I}_{\mathrm{end}}}, \ba_{\texttt{J}_{\mathrm{start}}:\texttt{J}_{\mathrm{end}}}\right)$ \Comment \text{\texttt{vmap} call here.}
    \EndFor
\EndFor
\end{algorithmic}
\end{algorithm}
\subsubsection{Training details}
\label{sec:training_details}
We now describe our training method. Though the computation of our interpolation coefficients only requires linear solves, we also compute the shape parameters as described in Section \ref{sec:common_details}.

\paragraph{Preprocessing}
Before we begin training, we first preprocess our inputs to enable more intuitive and comparable selection of the shape parameter and the ridge parameter across all experiments; a discussion of the magnitude of the ridge parameter $\theta$ is provided in Section~\ref{sec:nugget_effect}. Specifically, we normalized all spatial functions to each have zero mean and unit standard deviation. Recalling that $\{a_i(x)\}_{i=1}^{N + N_{\mathrm{test}}}$ are the input functions, let $\mu(x)$ be the estimated spatial mean and estimated $\zeta(x)$ be the spatial standard deviation computed from the training set at the set of points $X$. Both training and test input functions are then normalized as follows: $\frac{a_i(x) - \mu(x)}{\zeta(x)}, i=1,\dots,N+N_{\mathrm{test}}$.

We also conduct another critical preprocessing step: the property-preserving kernel interpolants to the $N$ training output functions must be precomputed, and their coefficients must be stored as they form the right-hand side to the operator kernel interpolant in \eqref{eq:operator_kernel_linearsystem_ridge_epsilon}. We conducted this preprocessing step on the CPU for all experiments simply due to the memory constraints of storing the resulting $N md$ coefficients; while a streaming approach is possible here also, we did not explore such an approach here as this was simply viewed as a preprocessing step.

\paragraph{Two-step training}
We use a simple two-step training procedure. First, we optimize for the operator kernel's shape parameter as discussed in Section \ref{sec:common_details}: we used MATLAB's \texttt{fzero} function on the CPU to find the shape parameter for the target condition number of $10^{15}$ as it proved to be significantly faster than Python's rootfinding routines; note that \texttt{fzero} (the Brent-Dekker method) is iterative and gradient-free. We then use this shape parameter to reform the Gramian $\bK$ and solve the linear system \eqref{eq:operator_kernel_linearsystem_ridge_epsilon} on the GPU via the JAX Python library. We measure the total training time as the time taken for each of these steps. We measure total inference time for our method as the sum of the times taken to form the evaluation matrix $\bK^\star$, predict the output coefficients, and reconstruct the velocity fields from those coefficients.


%
\section{Approximation theorems and convergence rates}
\label{sec:approx-thms}

We present approximation theorems for (i) approximation of coefficient maps by operator-side kernel methods, (ii) spatial reconstruction by divergence-free kernels, and (iii) the composed operator error. Corollaries record the corresponding bounds when ridge regularization is used in the operator-side regression (our experimental setting). Throughout, we specialize to the $C^4$ Mat\'ern kernel, i.e.\ $\nu=5/2$. We present two types of theorems: the first set of theorems assume that the relevant dimension of the problem is the sampling dimension of the input and output functions; the second set of theorems instead present the results in terms of an unknown intrinsic dimension. We believe the latter to be more representative of our experimental results.
In Section~\ref{sec:error_analysis}, we compare the empirical convergence rates across our benchmark problems with the theoretical estimates derived here, particularly the intrinsic-dimension rate in Theorem~\ref{th:dimension1}.

We equip $C(X;\mathbb R^p)$ with the uniform norm $\|f\|_\infty:=\sup_{x\in X}\|f(x)\|_2$, and all $L^2$ norms on $X$ and $\Omega$ are with respect to the Lebesgue measure. Here $p$ denotes the length of the coefficient vector returned by the operator learner (in the divergence-free reconstruction below, $p=Md$ after stacking vector weights).

\begin{theorem}[Universal approximation]
\label{th:ua}
Let $X\subset\mathbb R^{d_X}$ be compact, let $k_\nu(x,x')=\kappa_\nu(x-x')$ be the translation-invariant Mat\'ern kernel on $\mathbb R^{d_X}$, let $B\in\mathbb R^{p\times p}$ be symmetric positive definite, define $K_X(x,x'):=k_\nu(x,x')B$, and let $\mathcal H_{K_X}\subset C(X;\mathbb R^p)$ be the associated vector-valued RKHS. Then $\overline{\mathcal H_{K_X}}^{\|\cdot\|_\infty}=C(X;\mathbb R^p)$.
\end{theorem}
\begin{proof}
The Mat\'ern kernel is translation-invariant with strictly positive spectral density on $\mathbb R^{d_X}$, hence it is $c_0$-universal and $H_{k_\nu}$ is dense in $C(X)$ for compact $X$ \citep{sriperumbudur2011universality}. For $B\succ0$ and $K_X(x,x')=k_\nu(x,x')B$, one has $\mathcal H_{K_X}=\{B^{1/2}g:\ g\in H_{k_\nu}(X)^p\}$ with equivalent norms depending only on $B$ \citep{micchelli2005learning}. Given $f\in C(X;\mathbb R^p)$, set $g=B^{-1/2}f$ and approximate $g$ uniformly by $g_m\in H_{k_\nu}(X)^p$ componentwise; defining $f_m:=B^{1/2}g_m$ yields $f_m\in\mathcal H_{K_X}$ and $\|f_m-f\|_\infty\to0$.
\end{proof}

The next theorem gives a pessimistic worst-case rate for approximating the coefficient map from $N$ training inputs in ambient dimension $d_X$.

\begin{theorem}[Coefficient interpolation rate for vector-valued Mat\'ern kernels in $L^2(X)$]
\label{th:coeffsrate}
Let $X\subset\mathbb R^{d_X}$ be a bounded Lipschitz domain, let $k_\nu$ be the scalar Mat\'ern kernel of smoothness $\nu$, let $B\in\mathbb R^{p\times p}$ be symmetric positive definite, define $K_X(x,x'):=k_\nu(x,x')B$ with RKHS $\mathcal H_{K_X}$, and let $X_N\subset X$ be quasi-uniform with fill distance $h_X:=\sup_{x\in X}\min_{x_i\in X_N}\|x-x_i\|$. If $c^\dagger\in H^{\nu+d_X/2}(X;\mathbb R^p)$ and $\widehat c_N\in\mathcal H_{K_X}$ denotes the $K_X$-interpolant of $c^\dagger$ on $X_N$, then
\[
\|\widehat c_N-c^\dagger\|_{L^2(X)}
\le
C_X\,h_X^{\nu+d_X/2}\,\|c^\dagger\|_{H^{\nu+d_X/2}(X)}.
\]
If $X_N$ is quasi-uniform, then $h_X\asymp N^{-1/d_X}$ and
\[
\|\widehat c_N-c^\dagger\|_{L^2(X)}
\le
\widetilde C_X\,N^{-(\nu+d_X/2)/d_X}\,\|c^\dagger\|_{H^{\nu+d_X/2}(X)}.
\]
\end{theorem}
\begin{proof}
For separable $K_X(x,x')=k_\nu(x,x')B$ with $B\succ0$, $\mathcal H_{K_X}$ is (up to norm-equivalence constants depending only on $B$) isomorphic to $H_{k_\nu}(X)^p$ and interpolation corresponds to componentwise scalar interpolation \citep{micchelli2005learning}. Since $H_{k_\nu}(X)\simeq H^{\nu+d_X/2}(X)$ \citep[Thm.~10.47]{Wendland2004}, the scalar Sobolev-type estimate gives $\|I_{X_N}g-g\|_{L^2(X)}\le C h_X^{\nu+d_X/2}\|g\|_{H^{\nu+d_X/2}(X)}$, and applying it componentwise yields the result; $h_X\asymp N^{-1/d_X}$ for quasi-uniform $X_N$.
\end{proof}
\begin{corollary}[Ridge perturbation of coefficient interpolation in $L^2(X)$]
\label{cor:ridge-coeff}
Under the assumptions of Theorem~\ref{th:coeffsrate}, let $\widehat c_{N,\alpha}\in\mathcal H_{K_X}$ be the kernel ridge regression estimator
\[
\widehat c_{N,\alpha}\in\arg\min_{f\in\mathcal H_{K_X}}
\frac1N\sum_{i=1}^N \|f(x_i)-c^\dagger(x_i)\|_2^2+\alpha\|f\|_{\mathcal H_{K_X}}^2,
\]
and let $\widehat c_N$ denote the $K_X$-interpolant of $c^\dagger$ on $X_N$. Then
\[
\|\widehat c_{N,\alpha}-c^\dagger\|_{L^2(X)}
\le
C_X h_X^{\nu+d_X/2}\|c^\dagger\|_{H^{\nu+d_X/2}(X)}
+
\|\widehat c_{N,\alpha}-\widehat c_N\|_{L^2(X)}.
\]
\end{corollary}

\begin{proof}
Add and subtract $\widehat c_N$ and apply the triangle inequality, then invoke Theorem~\ref{th:coeffsrate}.
\end{proof}

We next analyze spatial reconstruction using divergence-free Mat\'ern kernels with vector weights, leveraging key work by Fuselier et al.

\begin{theorem}[Divergence-free Mat\'ern reconstruction error in $L^2(\Omega)$ with vector weights]
\label{th:generalization}
Let $\Omega\subset\mathbb R^d$ be a bounded Lipschitz domain, let $\phi_\nu$ denote the scalar Mat\'ern kernel on $\mathbb R^d$, define the divergence-free matrix-valued kernel $\Phi_{\mathrm{div}}(x):=(-\Delta I+\nabla\nabla^\top)\phi_\nu(x)$, fix quasi-uniform sites $Y_M=\{y_j\}_{j=1}^M\subset\Omega$ with fill distance $h_Y:=\sup_{y\in\Omega}\min_{y_j\in Y_M}\|y-y_j\|$, define the reconstruction operator $\mathcal R:\mathbb R^{Md}\to L^2(\Omega;\mathbb R^d)$ by
\[
(\mathcal R c)(x):=\sum_{j=1}^M \Phi_{\mathrm{div}}(x-y_j)c_j,\qquad c=[c_1;\dots;c_M],\ \ c_j\in\mathbb R^d,
\]
and define the $L^2(\Omega)$ mass matrix $\mathbf T\in\mathbb R^{Md\times Md}$ by blocks
\[
\mathbf T_{ij}:=\int_\Omega \Phi_{\mathrm{div}}(x-y_i)^\top \Phi_{\mathrm{div}}(x-y_j)\,dx\in\mathbb R^{d\times d}.
\]
If $v\in H^{\nu+d/2}(\Omega;\mathbb R^d)$ is divergence-free and $v_M=\mathcal Rc^\star$ denotes its divergence-free Mat\'ern interpolant on $Y_M$, then for any $\widehat c\in\mathbb R^{Md}$,
\[
\|v-\mathcal R\widehat c\|_{L^2(\Omega)}
\le
C_Y h_Y^{\nu+d/2}\|v\|_{H^{\nu+d/2}(\Omega)}
+
\sqrt{\lambda_{\max}(\mathbf T)}\,\|\widehat c-c^\star\|_2.
\]
If $Y_M$ is quasi-uniform, then $h_Y\asymp M^{-1/d}$ and the first term is $\mathcal O(M^{-(\nu+d/2)/d})$.
\end{theorem}

\begin{proof}
Sobolev approximation for divergence-free kernels yields $\|v-v_M\|_{L^2(\Omega)}\le C_Y h_Y^{\nu+d/2}\|v\|_{H^{\nu+d/2}(\Omega)}$ \citep{fuselier2008sobolev}. Moreover,
\[
\|\mathcal R(c^\star-\widehat c)\|_{L^2(\Omega)}^2
=
\int_\Omega \Big\|\sum_{j=1}^M \Phi_{\mathrm{div}}(x-y_j)(c^\star_j-\widehat c_j)\Big\|_2^2\,dx
=
(c^\star-\widehat c)^\top \mathbf T (c^\star-\widehat c)
\le
\lambda_{\max}(\mathbf T)\,\|\widehat c-c^\star\|_2^2.
\]
The claim follows from $v-\mathcal R\widehat c=(v-v_M)+\mathcal R(c^\star-\widehat c)$ and the triangle inequality.
\end{proof}
With these results in hand, we now present convergence rates for the operator learning problem.
\begin{theorem}[Operator error in $L^2(X;L^2(\Omega))$]
\label{th:errorbound}
Under the assumptions of Theorems~\ref{th:coeffsrate} and \ref{th:generalization}, assume $\|v\|_{L^\infty(X;H^{\nu+d/2}(\Omega))}:=\sup_{x\in X}\|v(\cdot;x)\|_{H^{\nu+d/2}(\Omega)}<\infty$, define $c^\dagger(x)\in\mathbb R^{Md}$ by $v_M(\cdot;x)=\mathcal R(c^\dagger(x))$ where $v_M(\cdot;x)$ is the divergence-free Mat\'ern interpolant of $v(\cdot;x)$ on $Y_M$, and let $\widehat c_N$ denote the $K_X$-interpolant of $c^\dagger$ on $X_N$. Defining $\widehat{\mathcal G}_{N,M}(x):=\mathcal R(\widehat c_N(x))$ and $\mathcal G(x):=v(\cdot;x)$ yields
\[
\|\widehat{\mathcal G}_{N,M}-\mathcal G\|_{L^2(X;L^2(\Omega))}
\le
\widetilde C_Y M^{-(\nu+d/2)/d}
+
\widetilde C_X \sqrt{\lambda_{\max}(\mathbf T)}\,N^{-(\nu+d_X/2)/d_X}.
\]
\end{theorem}

\begin{proof}
For each $x\in X$, apply Theorem~\ref{th:generalization} with $v=v(\cdot;x)$, $c^\star=c^\dagger(x)$, and $\widehat c=\widehat c_N(x)$; take the $L^2(X)$ norm and use $\|v\|_{L^\infty(X;H^{\nu+d/2}(\Omega))}<\infty$ to bound the spatial term, then substitute Theorem~\ref{th:coeffsrate} and $h_Y\asymp M^{-1/d}$, absorbing multiplicative constants into $\widetilde C_Y,\widetilde C_X$.
\end{proof}

\begin{corollary}[Ridge operator bound in $L^2(X;L^2(\Omega))$]
\label{cor:ridge-op}
Under the assumptions of Theorem~\ref{th:errorbound}, replace $\widehat c_N$ by the kernel ridge estimator $\widehat c_{N,\alpha}$ of Corollary~\ref{cor:ridge-coeff} and define $\widehat{\mathcal G}_{N,M,\alpha}(x):=\mathcal R(\widehat c_{N,\alpha}(x))$. Then
\[
\|\widehat{\mathcal G}_{N,M,\alpha}-\mathcal G\|_{L^2(X;L^2(\Omega))}
\le
\widetilde C_Y M^{-(\nu+d/2)/d}
+
\sqrt{\lambda_{\max}(\mathbf T)}\,\|\widehat c_{N,\alpha}-c^\dagger\|_{L^2(X)},
\]
and consequently
\[
\|\widehat{\mathcal G}_{N,M,\alpha}-\mathcal G\|_{L^2(X;L^2(\Omega))}
\le
\widetilde C_Y M^{-(\nu+d/2)/d}
+
\widetilde C_X \sqrt{\lambda_{\max}(\mathbf T)}\,N^{-(\nu+d_X/2)/d_X}
+
\sqrt{\lambda_{\max}(\mathbf T)}\,\|\widehat c_{N,\alpha}-\widehat c_N\|_{L^2(X)}.
\]
\end{corollary}

\begin{proof}
Repeat the proof of Theorem~\ref{th:errorbound} with $\widehat c_N$ replaced by $\widehat c_{N,\alpha}$, then apply Corollary~\ref{cor:ridge-coeff}.
\end{proof}

We now give intrinsic-dimension analogues for some unknown intrinsic dimension $r$. The only nontrivial additional ingredient is the mapping property of Sobolev/Bessel-potential spaces under composition with a sufficiently smooth bi-Lipschitz chart; for $s=\nu+r/2>1$, a bare bi-Lipschitz assumption is not sufficient in general, hence we assume the corresponding composition operator is bounded (cf.\ the cited references).

\begin{theorem}[Intrinsic-dimension coefficient interpolation rate in $L^2(X_{\mathrm{eff}})$]
\label{th:dimension1}
Let $X\subset\mathbb R^{d_X}$ be a bounded Lipschitz domain, let $k_\nu$ be the scalar Mat\'ern kernel of smoothness $\nu$, let $B\in\mathbb R^{p\times p}$ be symmetric positive definite, define $K_X(x,x'):=k_\nu(x,x')B$, and assume there exist an integer $r\ll d_X$, a bounded Lipschitz domain $\widetilde X\subset\mathbb R^r$, and a bi-Lipschitz diffeomorphism $\Phi:\widetilde X\to X_{\mathrm{eff}}\subset X$ such that the composition operators $u\mapsto u\circ\Phi$ and $u\mapsto u\circ\Phi^{-1}$ are bounded on $H^{\nu+r/2}$ (e.g.\ $\Phi,\Phi^{-1}$ have bounded derivatives up to order $\lceil \nu+r/2\rceil$), and $X_N=\Phi(\widetilde X_N)$ for a quasi-uniform set $\widetilde X_N\subset\widetilde X$ with fill distance $h_{\widetilde X}:=\sup_{\xi\in \widetilde X}\min_{\xi_i\in\widetilde X_N}\|\xi-\xi_i\|$. If $c^\dagger(x)=\widetilde c(\Phi^{-1}(x))$ for $x\in X_{\mathrm{eff}}$ with $\widetilde c\in H^{\nu+r/2}(\widetilde X;\mathbb R^p)$ and $\widehat c_N$ denotes the $K_X$-interpolant of $c^\dagger$ on $X_N$, then
\[
\|\widehat c_N-c^\dagger\|_{L^2(X_{\mathrm{eff}})}
\le
C_{\Phi}\, h_{\widetilde X}^{\nu+r/2}\,\|\widetilde c\|_{H^{\nu+r/2}(\widetilde X)}.
\]
If $\widetilde X_N$ is quasi-uniform, then $h_{\widetilde X}\asymp N^{-1/r}$ and the error is $\mathcal O(N^{-(\nu+r/2)/r})$.
\end{theorem}

\begin{proof}
By the boundedness of composition on $H^{\nu+r/2}$ under the stated regularity assumptions \citep{runst1996sobolev,bourdaud2023composition}, pullback by $\Phi$ yields norm equivalences between $H^{\nu+r/2}(X_{\mathrm{eff}})$ and $H^{\nu+r/2}(\widetilde X)$ (and likewise for $L^2$). Under pullback, interpolation on $X_N=\Phi(\widetilde X_N)$ corresponds to interpolation on $\widetilde X_N$ for the pulled-back target $\widetilde c$, and applying the $r$-dimensional $L^2$ Mat\'ern/Sobolev estimate yields the claim, with constants absorbed into $C_\Phi$.
\end{proof}

\begin{theorem}[Intrinsic-dimension operator bound in $L^2(X_{\mathrm{eff}};L^2(\Omega))$]
\label{th:dimension2}
Under the assumptions of Theorems~\ref{th:generalization} and \ref{th:dimension1}, defining $\widehat{\mathcal G}_{N,M}(x):=\mathcal R(\widehat c_N(x))$ and $\mathcal G(x):=v(\cdot;x)$ yields
\[
\|\widehat{\mathcal G}_{N,M}-\mathcal G\|_{L^2(X_{\mathrm{eff}};L^2(\Omega))}
\le
\widetilde C_Y M^{-(\nu+d/2)/d}
+
\widetilde C_{\Phi}\sqrt{\lambda_{\max}(\mathbf T)}\,N^{-(\nu+r/2)/r}.
\]
\end{theorem}

\begin{proof}
Repeat the proof of Theorem~\ref{th:errorbound} with $X$ replaced by $X_{\mathrm{eff}}$, then substitute Theorem~\ref{th:dimension1} and $h_Y\asymp M^{-1/d}$, absorbing constants into $\widetilde C_Y,\widetilde C_\Phi$.
\end{proof}

\section{Results}
\label{sec:results}
We now present experimental results comparing our method with operator learning baselines on a diverse set of challenging 2D and 3D operator learning problems involving incompressible flows. We demonstrate that the property-preserving kernel method (PPKM) consistently outperforms state-of-the-art neural operators by several orders of magnitude in accuracy while \emph{analytically} preserving incompressibility and other properties across the board. We compared our method against \textcolor{edit}{five} baselines: (i) the ``vanilla'' kernel method (VKM)~\citep{batlle2024kernel}, (ii) the Geo-FNO (geometry-aware FNO)~\citep{li2023fourier}, (iii) the Transolver~\citep{wu2024transolver}, \textcolor{edit}{(iv) the Geo-FNO + divergence-loss, and (v) the Transolver + divergence-loss}; see Appendix~\ref{app:neural_operator} for implementation details. Several variations of the PPKM were used in the experiments; the $\pi$-PPKM is spatially periodic, the $\eta$-PPKM incorporates turbulence power laws, and the $\pi \eta$-PPKM includes both properties.
\textcolor{edit}{For a fair hardware comparison, the PPKM (and its variants), the VKM, and the neural operators were all trained on an NVIDIA A100 GPU (the shape parameter $\epsilon$ for the kernel methods was selected on an AMD Ryzen 16-core CPU). The kernel methods were trained using double precision while the neural operators were trained using single precision. Notably, the kernel methods also scaled on a desktop NVIDIA RTX 4080 GPU and continued to train substantially faster than the neural operators (runtime results not shown).}
\textcolor{edit}{The divergence-loss neural operator variants give small improvements, if any, in the reported numerical divergence on most problems. The clearest improvements occur for the Transolver on the spatial and spacetime 2D Taylor--Green vortices problems, although the predicted fields still exhibit nonzero divergence. In contrast, the PPKM and its variants preserve incompressibility analytically across all problems.}
\textcolor{edit}{Despite our use of double precision, our method consistently trains orders of magnitude faster than the neural operators and exhibits half to two orders of magnitude faster inference times.}
\textcolor{edit}{In the two turbulent problems, the PPKM's predicted second-order velocity structure functions closely match those of the reference velocity fields across the resolved separation distances (Section~\ref{sec:velocity_turb_analysis}), indicating that the method captures their scale-dependent velocity statistics.}
We further evaluate out-of-distribution (OOD) generalization across the benchmarks in Section~\ref{sec:ood_tests}.

\textbf{Experimental details}: All data sets (\textcolor{edit}{except for the Taylor--Green vortices problems which have analytical solutions, and the 3D forced homogeneous isotropic turbulence problem which used OpenFOAM~\footnote{\url{https://openfoam.org}}}) were produced using the SU2 numerical solver~\citep{economon2016su2} for incompressible flows~\citep{economon2020simulation}, which uses a second-order accurate finite volume method on unstructured meshes. Moreover, SU2 supports a variety of boundary conditions (BCs), including but not limited to periodic BCs, no-slip BCs, and wall functions for turbulent flow. For the turbulent problems explored, we solved the incompressible Reynolds averaged Navier--Stokes (RANS) equations in SU2 with the Shear Stress Transport (SST) turbulence model. We produced domain geometry meshes using \texttt{gmsh}~\citep{geuzaine2009gmsh} and used them within SU2; these scripts will be provided in the codebase upon publication of the manuscript. Table~\ref{tab:exp_config} lists the details pertaining to the flow in each problem; the domain, relevant flow regime, forcing term $f$, initial (IC) and boundary conditions, time domain, and operator learning map. For each problem, we generated $10,200$ time dependent simulations of which $N=10,000$ were used for training and $N_{\mathrm{test}}=200$ for testing. Following the mathematical notation introduced in Section~\ref{sec:methods}, the input functions $a \in \A$ may be scalar- or vector-valued; in contrast, the output functions are always vector-valued incompressible velocity fields. We used the truncated normal distribution for randomly sampling governing parameters such as initial velocity, viscosity, leading coefficients, etc. We define it as $\mathcal{N}_{[a, b]}(\mu, \sigma)$, where $\mu$ is the mean, $\sigma^2$ is the variance, and the random variables are sampled in the range $[a, b]$. This section reports the results for a variety of important 2D and 3D experiments (two each); results for additional classical benchmark problems are reported in Appendix~\ref{app:additional_results}. In Section~\ref{sec:nugget_effect}, we also present results on the effect of the ridge parameter $\theta$ on the accuracy of our method. Section~\ref{sec:gp_results} discusses results from a GP experiment for uncertainty quantification using the 3D species transport example as a case study (see Section~\ref{sec:uq} for details). \textcolor{edit}{Section~\ref{sec:velocity_turb_analysis} compares the second-order velocity structure functions for the two turbulent problems to assess whether the kernel predictions preserve the scale-dependent structure of the reference fields.}

\textbf{Metrics reported}: Let $\widetilde{\bv} \in \mathbb{R}^m$ and $\bv \in \mathbb{R}^m$ be the vectors containing the pointwise velocity magnitudes of the predicted and true velocity fields respectively. We report the spatial relative $\ell_2$ error computed as $(\|\widetilde{\bv} - \bv\|_2)/\|\bv\|_2$, \emph{i.e.}, the relative error in the magnitudes. Further, we report the maximum pointwise interior divergence (ignoring spurious divergences due to possible edge effects from the solver or the method to compute divergence). Both metrics are averaged over the test functions. We additionally performed ablation studies over the number of training functions $N$ (uniformly sampled at random without replacement) and the number of spatial samples $m$ and report the results with the combination that yielded the lowest errors for our method. \textcolor{edit}{We report a least-squares best-fit line for the PPKM errors on each error convergence plot to estimate their empirical rate of decay with $N$. In Section~\ref{sec:error_analysis}, we compare these fitted rates with the worst-case estimates and use rate-compatible effective dimensions to interpret the faster observed rates.} For the Geo-FNO, we reported errors across all $N$ using the largest available $m$ for each problem. Additionally, at the value of $N$ yielding the lowest error for the property-preserving kernel method, we reported Geo-FNO and Transolver errors using the same spatial resolution $m$. Table~\ref{tab:best_results} reports the relative $\ell_2$ errors, maximum pointwise divergence, training and inference (on training set) times, and the $\theta$ used by the proposed method. Section~\ref{sec:training_details} describes how training times are computed for the property-preserving and vanilla kernel methods. We now present each of our benchmark problems.
\begin{table}[pos=!htpb]
\renewcommand{\arraystretch}{1.3}
\setlength{\tabcolsep}{1.4pt}
\centering
\resizebox{\linewidth}{!}{%
\begin{tabular}{l c l l c r c c r}
\toprule
\# & Problem & Domain & Flow regime & $f$ & IC & BC & Time & Operator \\
\midrule

\multirow{5}{*}{1)}
&  2D Flow Past & $[0,20]\!\times\![0,14]$ & Laminar & -- &  $\bu$ given & $\bu$ given at & $T=10$ & $\G:\bu(x,0)\!\to\!\bu(y,10)$ \\
&  a Cylinder &  & No shedding & &  & top/bottom/inlet & $\; \Delta t=10^{-3}$ &  \\
&  &  &  $\mathrm{Re}\!\in\![25,64]$  & &  & No-slip cylinder  &  &  \\
&  &  & Vortex shedding & &  & outlet, $p=0$ &  &  \\
&  &  & $\mathrm{Re}\!\in\![112,199]$ & &  &  &  &  \\
\midrule

\multirow{2}{*}{2)}
&  2D Lid--Driven & $[0,1]^2$ & Laminar & -- & $\bu$ given & $\bu^{\texttt{top}}$ given & $T=5$ & $\G:\bu(x,0)\!\to\!\bu(y,5)$ \\
&  Cavity Flow &  &  &  &  & $\bu_b=0,$ else & $\; \Delta t=10^{-3}$ &  \\
\midrule

\multirow{4}{*}{3)}
& 2D Backward & $[0,15]$ & Laminar & -- & $\; \bu^{\texttt{inlet}}=\nu x_2(0.5-x_2)$ & $\bu^{\texttt{inlet}}$ given & $T=5$ & $\G:\bu(x,0)\!\to\!\bu(y,5)$ \\
& Facing Step & $\!\times [-0.5,0.5]$ & $\mathrm{Re}\!\in\![28,900]$ &  & $\bu=0,$ else & $\bu=0$ at & $\; \Delta t=10^{-3}$ &  \\
&   &  &  &  &  & top, bottom, step &  &  \\
&   &  &  &  &  & outlet, $p=0$  &  &  \\
\midrule

\multirow{4}{*}{4)}
&  2D Buoyancy--Driven & $\; [0,1]^2$ & Laminar & $-g$ & $\bu=0$ & No-slip walls & $T=5$ & $\G:\texttt{T}(x, 0)\!\to\!\bu(y,5)$ \\
&  Cavity Flow &  & Ra$=10^6$ &  & $\texttt{T}_{\mathrm{left}}=\texttt{T}_1$ &  & $\; \Delta t=10^{-3}$ &  \\
&   &  &  &  & $\texttt{T}_{\mathrm{right}}=\texttt{T}_2$ &  &  &  \\
&   &  &  &  & $\texttt{T}=288.15\,\mathrm{K},$ else &  &  &  \\
\midrule

\multirow{3}{*}{5)}
&  2D Taylor--Green & $[0,2\pi]^2$ & Laminar & -- & $\bu=$ & $\bu^{\texttt{left}}=\bu^{\texttt{right}}$ & $T=1$ & $\G:\bu(x,0)\!\to\!\bu(y,1)$ \\
&  Vortices & & &  & $(A\sin(x_1) \cos(x_2),$ & $\bu^{\texttt{top}}=\bu^{\texttt{bottom}}$ &  &  \\
&  (\textbf{spatial}) &  & &  & $A\cos(x_1) \sin(x_2))e^{-2\nu t}$ &  &  &  \\
\midrule

\multirow{3}{*}{6)}
&  2D Taylor--Green & $[0,2\pi]^2$ & Laminar & -- & $\bu=$ & $\bu^{\texttt{left}}=\bu^{\texttt{right}}$ & $T=[0.7, 0.8$ & $\G:\bu(x,0)\!\to\!\bu(y,T)$ \\
&  Vortices &  & &  & $(A\sin(x_1) \cos(x_2),$ & $\bu^{\texttt{top}}=\bu^{\texttt{bottom}}$ & $0.9, 1]$ &  \\
&  (\textbf{spacetime}) &  & &  & $A\cos(x_1) \sin(x_2))e^{-2\nu t}$ &  &  & \\
\midrule

\multirow{3}{*}{7)}
&  2D Merging & $[0, 2\pi]^2$ & Laminar & -- & $\bu$ given & $\bu^{\texttt{left}}=\bu^{\texttt{right}}$ & $T=0.4$ & $\, \G:\omega(x, 0)\!\to\!\bu(y,0.4)$ \\
&  Vortices &  &  &  &  & $\bu^{\texttt{top}}=\bu^{\texttt{bottom}}$ & $\; \Delta t=10^{-3}$ &  \\
&   &  &  &  & $\omega$ (vorticity) given &  &  & \\
\midrule

\multirow{6}{*}{8)}
&  3D Species & See Figure~\ref{fig:species_transport} & Turbulent & -- & $\bu^{\texttt{inlet}},$ given & $\bu^{\texttt{inlet}}$ given & $T=0.5$ & $\G:\bu^{\texttt{inlet}}\!\to\!\bu(y,0.5)$ \\
&  Transport &  & $\mathrm{Re}\!\in\![10^5,\!10^6]$ &  &  $\bu=0,$ else & no-slip & $\; \Delta t=0.005$ &  \\
&  &  &  &  &  & inner wall, &  &  \\
&  &  &  &  &  & outer wall, &  &  \\
&  &  &  &  &  & \& blades &  &  \\
&  &  &  &  &  & outlet, $p=0$ &  &  \\
\midrule

\multirow{5}{*}{9)}
&  3D Forced & $[0,2\pi]^3$ & Turbulent & $\bf = [A (\sin(q_1 x_3) + \cos(q_2 x_2)),$ & $\bu$ given & $\bu^{\texttt{left}} = \bu^{\texttt{right}}$ & $T=0.5$ & $\G:\bf\!\to\!\bu(y,0.5)$ \\
&  Homogeneous Isotropic & & $\mathrm{Re}=6\times10^6$ & $B (\sin(q_1 x_1) + \cos(q_2 x_3)),$ &  &  $\bu^{\texttt{top}}= \bu^{\texttt{bottom}}$ & $\; \Delta t=0.002$ &  \\
&  Turbulence & & & $C (\sin(q_1 x_2) + \cos(q_2 x_1))]$ &  & $\bu^{\texttt{front}} = \bu^{\texttt{back}}$ &  &  \\
&  &  &  &  &  &  &  &  \\
&  &  &  &  &  &  &  &  \\
\bottomrule
\end{tabular}}
\caption{Problem configurations; flow regime, domain, forcing term ($f$), initial and boundary conditions (IC and BC), time domain, and the operator of interest. Re stands for the Reynolds number and Ra for the Rayleigh number.}
\label{tab:exp_config}
\end{table}
\begin{table}[pos=!htpb]
\renewcommand{\arraystretch}{1.03}
\setlength{\tabcolsep}{1pt}
\centering
\resizebox{\linewidth}{!}{%
\begin{tabular}{c|c|c|c|c|c|c|c|c|c|c}

\hline
Problem & $m$ & $N$ & $\theta$ & Metric & PPKM & VKM & Geo-FNO & Transolver & Geo-FNO, div. loss & Transolver, div. loss \\
\hline

\multirow{4}{*}{2D Flow Past a Cylinder} & \multirow{4}{*}{1000} & \multirow{4}{*}{100} & \multirow{4}{*}{$10^{-6}$} &
rel. $\ell_2$ error & $3.95\times10^{-5}, \Phi$ & $5.32\times10^{-5}$ & $1.618\times10^{-4}$ & $7.36\times10^{-3}$ & $1.555\times10^{-4}$ & $5.738\times10^{-3}$ \\
\hhline{~~~~-------}
& & & & Divergence & $0$ & $0.788$ & $0.7716$ & $0.7528$ & $0.7716$ & $0.709$ \\
\hhline{~~~~-------}
(no vortex shedding) & & & & Training time & $1.81$ & $1.39$ & $262.2$ & $56.2$ & $179.9$ & $60.54$ \\
\hhline{~~~~-------}
& & & & Inference time & $0.0053$ & $0.0024$ & $0.19$ & $0.09$ & $0.19$ & $0.09$ \\
\hline

\multirow{4}{*}{2D Flow Past a Cylinder} & \multirow{4}{*}{1000} & \multirow{4}{*}{10,000} & \multirow{4}{*}{$10^{-6}$} &
rel. $\ell_2$ error & $3.14\times10^{-7}, \Phi$ & $2.34\times10^{-6}$ & $4.783\times10^{-5}$ & $1.504\times10^{-4}$ & $3.737\times10^{-5}$ & $1.523\times10^{-4}$ \\
\hhline{~~~~-------}
& & & & Divergence & $0$ & $2.89$ & $4.943$ & $4.896$ & $4.942$ & $4.889$ \\
\hhline{~~~~-------}
(vortex shedding) & & & & Training time & $72$ & $72$ & $11,790$ & $4738$ & $11,870$ & $4685$ \\
\hhline{~~~~-------}
& & & & Inference time & $0.022$ & $0.018$ & $0.14$ & $0.069$ & $0.14$ & $0.069$ \\
\hline

\multirow{4}{*}{2D Lid--Driven Cavity} & \multirow{4}{*}{1000} & \multirow{4}{*}{10,000} & \multirow{4}{*}{$10^{-6}$} &
rel. $\ell_2$ error & $1.39\times10^{-6}, \Phi$ & $1.45\times10^{-6}$ & $1.509\times10^{-4}$ & $2.139\times10^{-4}$ & $1.795\times10^{-4}$ & $2.108\times10^{-4}$ \\
\hhline{~~~~-------}
& & & & Divergence & $0$ & $7.96$ & $18.14$ & $21.11$ & $18.09$ & $20.98$ \\
\hhline{~~~~-------}
Flow & & & & Training time & $72$ & $72$ & $6636$ & $4574$ & $6724$ & $4904$ \\
\hhline{~~~~-------}
& & & & Inference time & $0.021$ & $0.018$ & $0.085$ & $0.065$ & $0.085$ & $0.065$ \\
\hline

\multirow{4}{*}{2D Backward--Facing} & \multirow{4}{*}{1000} & \multirow{4}{*}{500} & \multirow{4}{*}{$10^{-6}$} &
rel. $\ell_2$ error & $1.51\times10^{-6}, \Phi$ & $2.9\times10^{-6}$ & $9.949\times10^{-5}$ & $7.330\times10^{-4}$ & $9.259\times10^{-5}$ & $7.247\times10^{-4}$ \\
\hhline{~~~~-------}
& & & & Divergence & $0$ & $0.907$ & $1.981$ & $1.987$ & $1.983$ & $1.959$ \\
\hhline{~~~~-------}
Step & & & & Training time & $1.59$ & $1.54$ & $225.7$ & $247.6$ & $216.2$ & $268.7$ \\
\hhline{~~~~-------}
& & & & Inference time & $0.0057$ & $0.0024$ & $0.056$ & $0.079$ & $0.056$ & $0.079$ \\
\hline

\multirow{4}{*}{2D Buoyancy--Driven} & \multirow{4}{*}{5000} & \multirow{4}{*}{10,000} & \multirow{4}{*}{$10^{-4}$} &
rel. $\ell_2$ error & $6.27\times10^{-6}, \Phi$ & $6.14\times10^{-6}$ & $1.421\times10^{-4}$ & $2.305\times10^{-4}$ & $1.571\times10^{-4}$ & $1.932\times10^{-4}$ \\
\hhline{~~~~-------}
& & & & Divergence & $0$ & $4.28$ & $3.546$ & $3.037$ & $3.513$ & $3.03$ \\
\hhline{~~~~-------}
Cavity Flow & & & & Training time & $92$ & $72$ & $21,160$ & $16,620$ & $21,600$ & $17,020$ \\
\hhline{~~~~-------}
& & & & Inference time & $0.12$ & $0.061$ & $0.30$ & $0.29$ & $0.30$ & $0.29$ \\
\hline

\multirow{4}{*}{2D Taylor--Green} & \multirow{4}{*}{500} & \multirow{4}{*}{5000} & \multirow{4}{*}{$10^{-8}$} &
rel. $\ell_2$ error & $8.76\times10^{-10}, \Phi^\pi$ & $1.68\times10^{-9}$ & $4.181\times10^{-4}$ & $1.531\times10^{-4}$ & $5.070\times10^{-4}$ & $6.072\times10^{-5}$ \\
\hhline{~~~~-------}
& & & & Divergence & $0$ & $1.42$ & $2.563$ & $0.6184$ & $5.037$ & $0.1035$ \\
\hhline{~~~~-------}
Vortices & & & & Training time & $25$ & $26$ & $2570$ & $2214$ & $2646$ & $2404$ \\
\hhline{~~~~-------}
& & & & Inference time & $0.0055$ & $0.0049$ & $0.066$ & $0.050$ & $0.066$ & $0.050$ \\
\hline

\multirow{4}{*}{2D Taylor--Green} & \multirow{4}{*}{500} & \multirow{4}{*}{5000} & \multirow{4}{*}{$10^{-8}$} &
rel. $\ell_2$ error & $8.65\times10^{-10}, \Phi^\pi$ & $8.67\times10^{-10}$ & $1.004\times10^{-3}$ & $3.447\times10^{-4}$ & $7.765\times10^{-4}$ & $2.177\times10^{-4}$ \\
\hhline{~~~~-------}
& & & & Divergence & $0$ & $14.58$ & $8.963$ & $1.484$ & $2.961$ & $0.1462$ \\
\hhline{~~~~-------}
Vortices, \textbf{Spacetime} & & & & Training time & $25$ & $25$ & $6092$ & $4888$ & $6204$ & $5080$ \\
\hhline{~~~~-------}
& & & & Inference time & $0.0063$ & $0.0063$ & $0.16$ & $0.17$ & $0.16$ & $0.17$ \\
\hline

\multirow{4}{*}{2D Merging Vortices} & \multirow{4}{*}{500} & \multirow{4}{*}{500} & \multirow{4}{*}{$10^{-4}$} &
rel. $\ell_2$ error & $1.09\times10^{-4}, \Phi^\pi$ & $1.09\times10^{-4}$ & $1.351\times10^{-3}$ & $3.747\times10^{-3}$ & $1.124\times10^{-3}$ & $3.397\times10^{-3}$ \\
\hhline{~~~~-------}
& & & & Divergence & $0$ & $0.56$ & $0.999$ & $0.7377$ & $0.9934$ & $0.7236$ \\
\hhline{~~~~-------}
& & & & Training time & $1.59$ & $1.42$ & $342.3$ & $209.3$ & $322.7$ & $229.6$ \\
\hhline{~~~~-------}
& & & & Inference time & $0.0030$ & $0.0022$ & $0.080$ & $0.063$ & $0.080$ & $0.063$ \\
\hline

\multirow{4}{*}{3D Species Transport} & \multirow{4}{*}{7000} & \multirow{4}{*}{10,000} & \multirow{4}{*}{$10^{-4}$} &
rel. $\ell_2$ error & $2.38\times10^{-4}, \Phi^\eta$ & $2.35\times10^{-4}$ & $4.958\times10^{-4}$ & $2.737\times10^{-3}$ & $\sim$ & $\sim$ \\
\hhline{~~~~-------}
& & & & Divergence & $0$ & $\sim$ & $\sim$ & $\sim$ & $\sim$ & $\sim$ \\
\hhline{~~~~-------}
& & & & Training time & $73$ & $72$ & $81,220$ & $26,110$ & -- & -- \\
\hhline{~~~~-------}
& & & & Inference time & $0.21$ & $0.12$ & $1.13$ & $0.43$ & -- & -- \\
\hline

\multirow{4}{*}{3D Homogeneous Forced} & \multirow{4}{*}{7000} & \multirow{4}{*}{10,000} & \multirow{4}{*}{$10^{-4}$} &
rel. $\ell_2$ error & $3.63\times10^{-4}, \Phi^{\pi, \eta}$ & $1.55\times10^{-4}$ & $1.842\times10^{-4}$ & $1.206\times10^{-3}$ & $1.486\times10^{-4}$ & $1.154\times10^{-3}$ \\
\hhline{~~~~-------}
& & & & Divergence & $0$ & $1.43$ & $1.823$ & $2.193$ & $1.812$ & $2.142$ \\
\hhline{~~~~-------}
Isotropic Turbulence & & & & Training time & $78$ & $78$ & $46,200$ & $22,620$ & $70,090$ & $23,410$ \\
\hhline{~~~~-------}
& & & & Inference time & $0.43$ & $0.34$ & $1.636$ & $0.4631$ & $1.636$ & $0.4631$ \\
\hline
\end{tabular}}
\caption{Accuracy, divergence, and runtime results. For each problem, we reported the spatial rel. (relative) $\ell_2$ error and maximum pointwise divergence on the test set, total training time, and inference time (on the training set) corresponding to $m$ points and $N$ training functions they were trained with, and the $\theta$ used by the PPKM and the VKM. The runtimes are reported in seconds. \textcolor{edit}{We report the (statistically significant) mean results over three random seeds for the Geo-FNO and the Transolver models, and their variants with an additional divergence-loss}. The choice of the PPKM is shown next to the relative $\ell_2$ error (see Section~\ref{sec:spatial_kernel} for notation). The PPKM and the VKM were trained in double precision while the neural operators were trained in single precision. $\sim$ indicates a blowup in the divergence magnitude ($> 10^{3}$).}
\label{tab:best_results}
\end{table}
\clearpage
\subsection{2D flow past a cylinder}
\label{sec:flow_cylinder}
\begin{figure}[pos=!htpb]
    \centering
    \begin{tabular}{ccc}
    &
    \begin{subfigure}{0.4\linewidth}
        \centering
        \includegraphics[width=\linewidth]{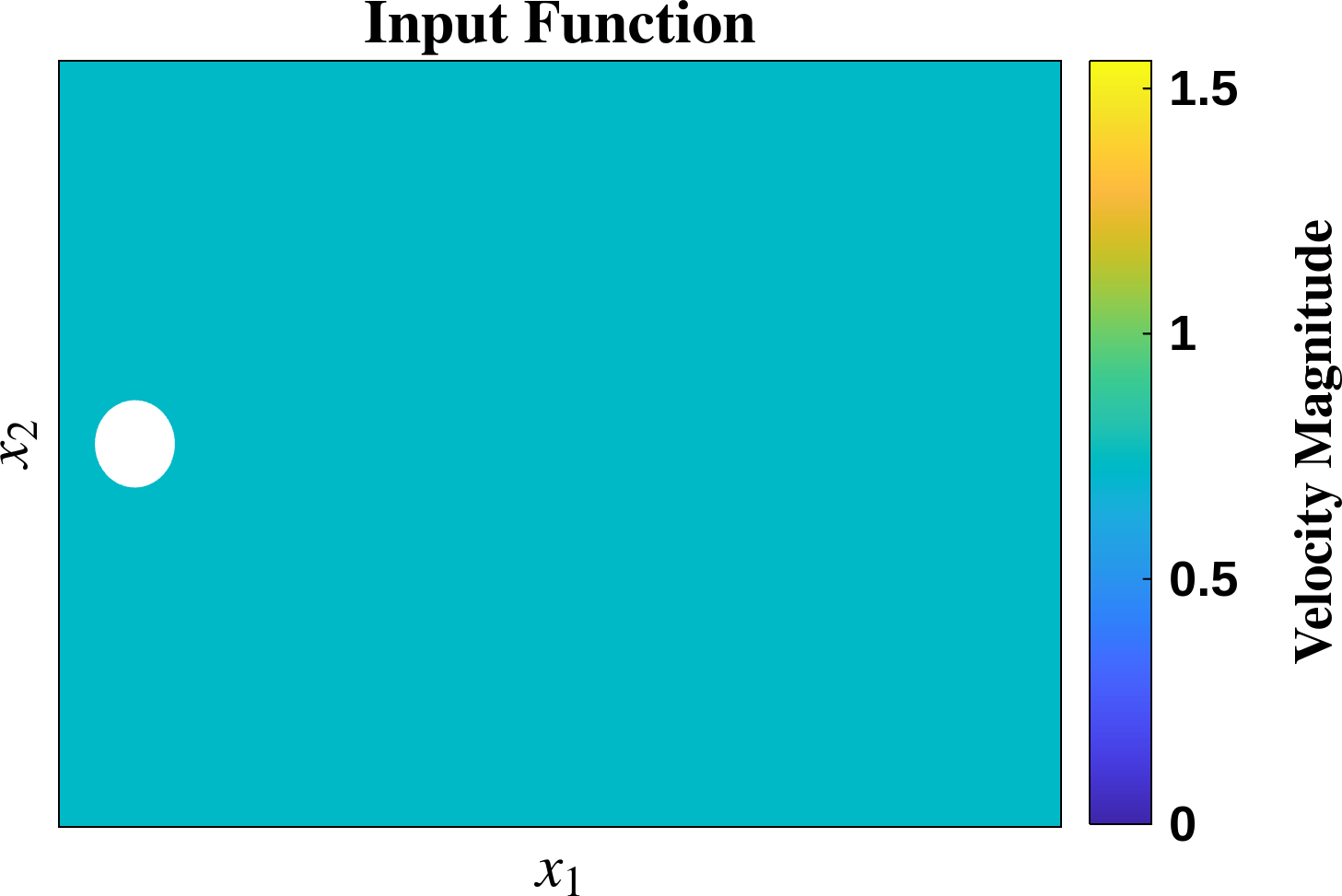}
        \caption{}
        \label{fig:flow_cylinder_input}
    \end{subfigure} &
    \begin{subfigure}{0.4\linewidth}
        \centering
        \includegraphics[width=\linewidth]{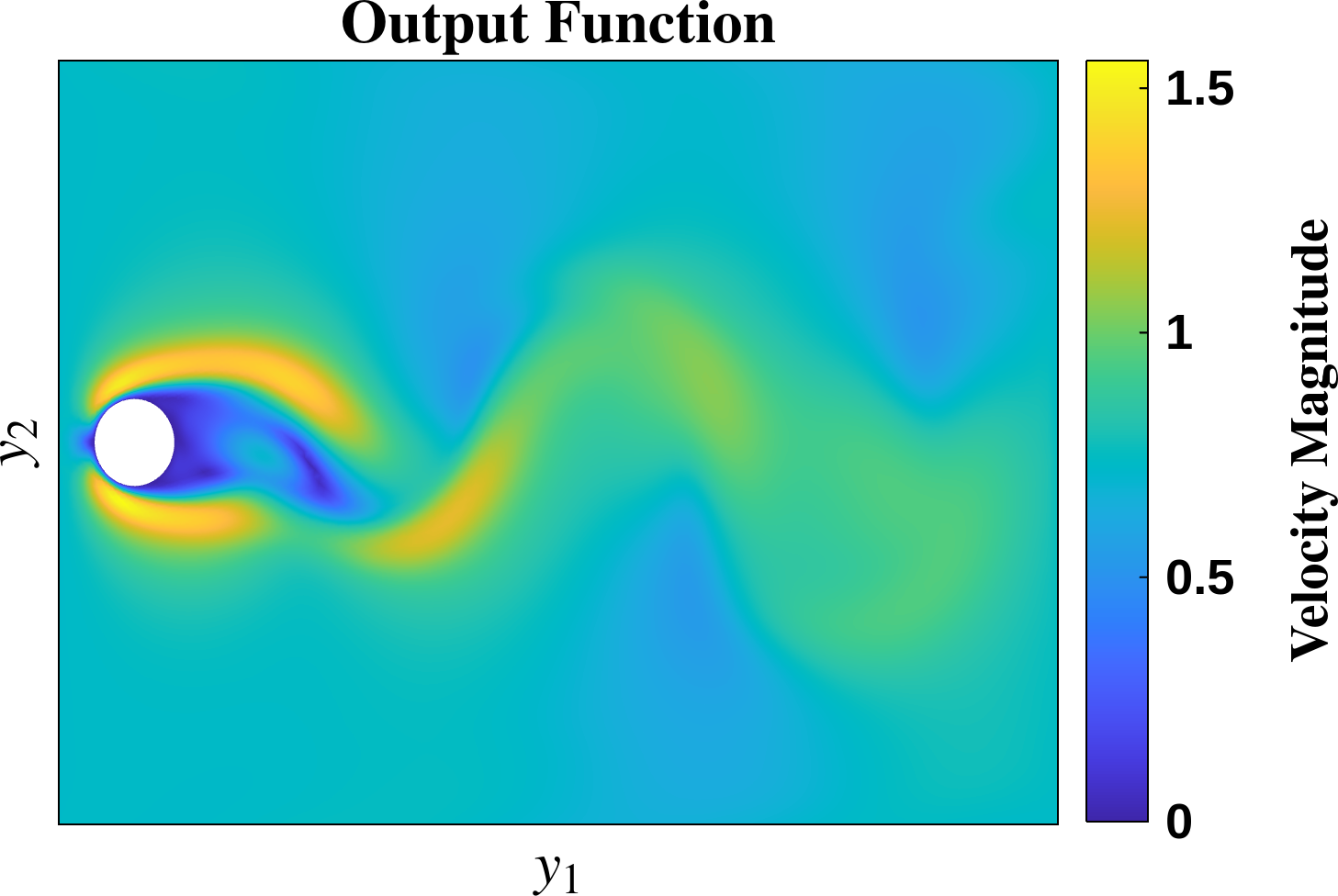}
        \caption{}
        \label{fig:flow_cylinder_output}
    \end{subfigure} \\
    \raisebox{10ex}{\rotatebox{90}{Vortex Shedding}} &
    \begin{subfigure}{0.48\linewidth}
        \centering
        \includegraphics[width=\linewidth]{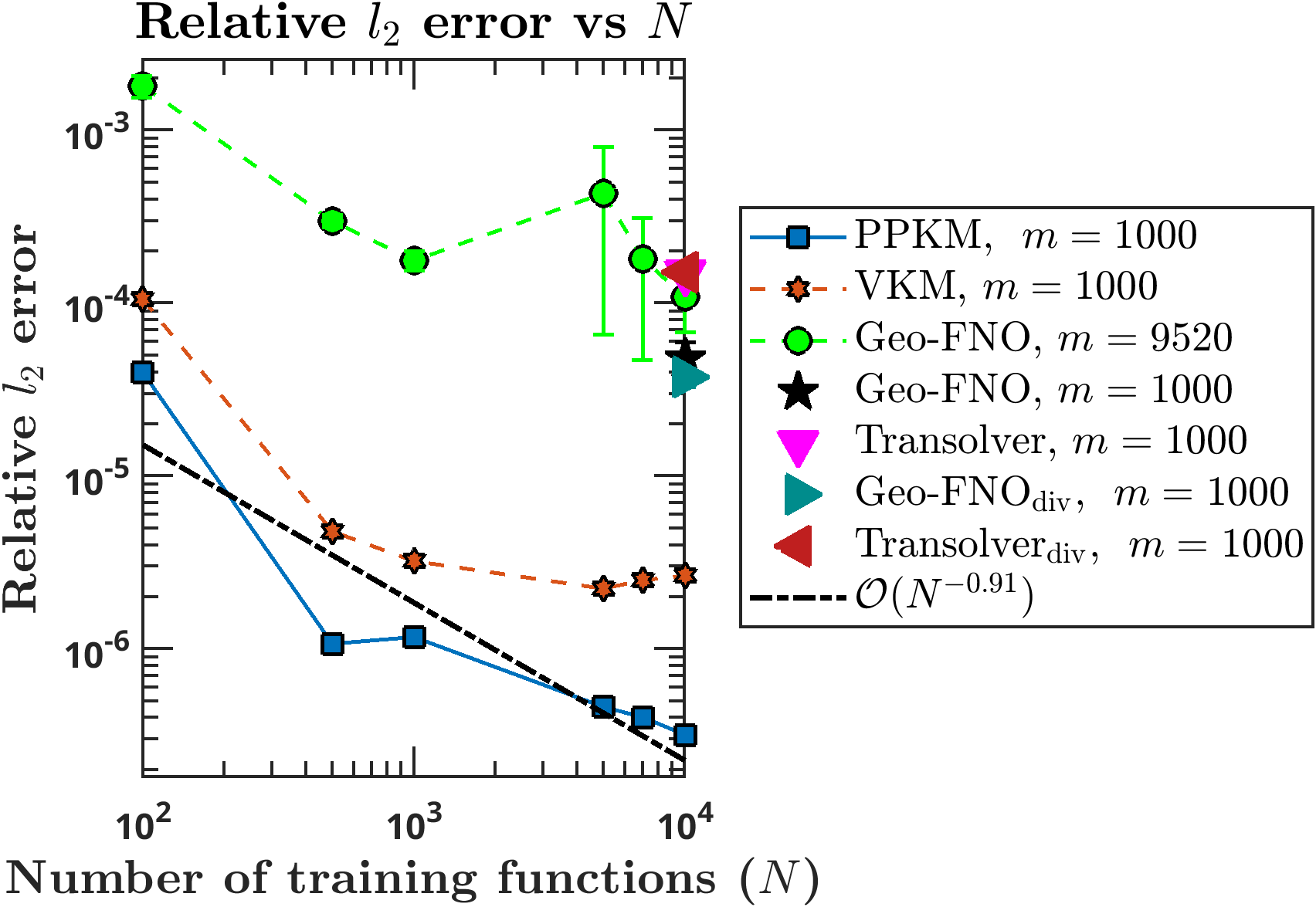}
        \caption{}
        \label{fig:flow_cylinder_shedding_errors}
    \end{subfigure} &
    \begin{subfigure}{0.48\linewidth}
        \centering
        \includegraphics[width=\linewidth]{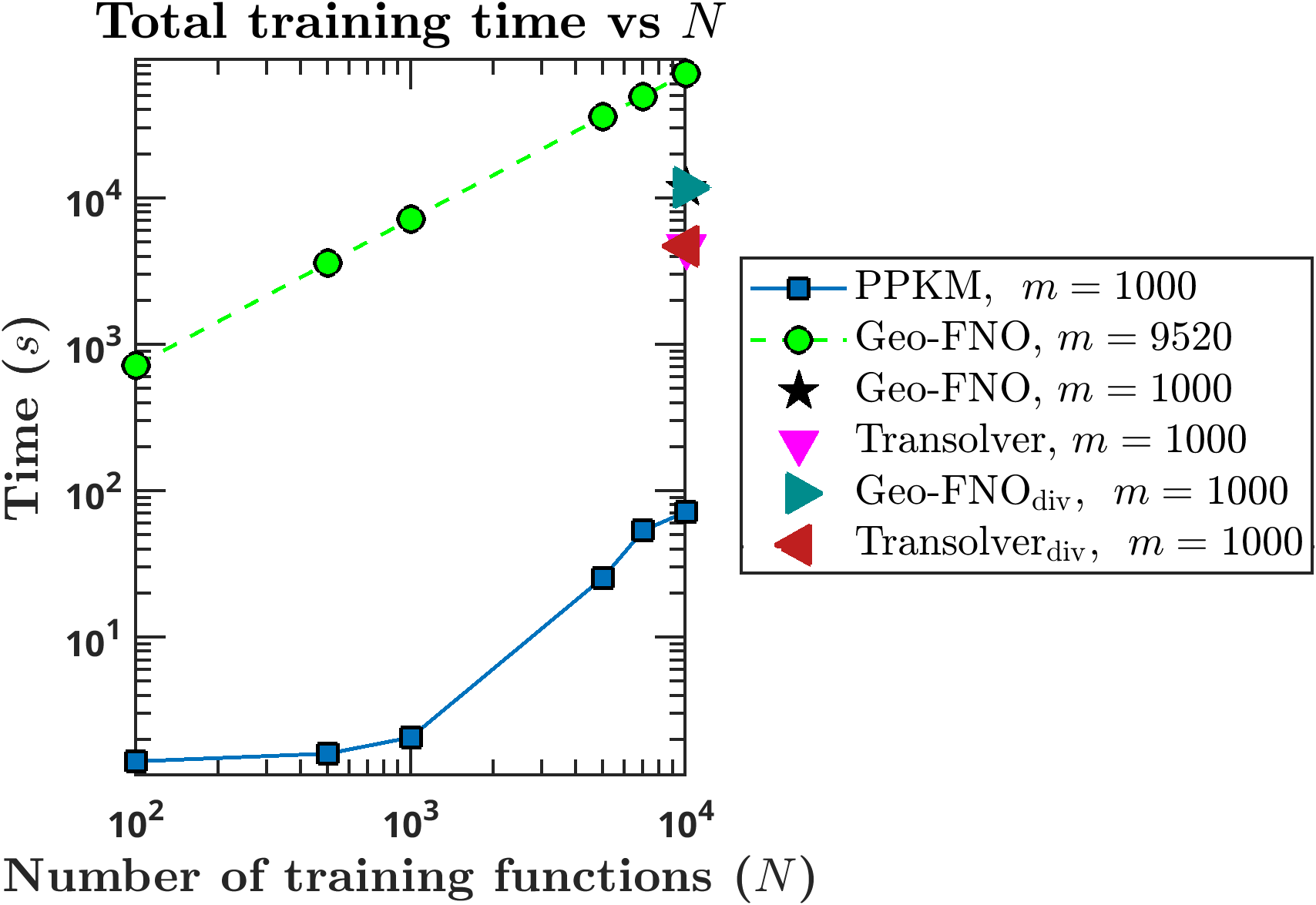}
        \caption{}
        \label{fig:flow_cylinder_shedding_times}
    \end{subfigure} \\
    \raisebox{10ex}{\rotatebox{90}{No Vortex Shedding}} &
    \begin{subfigure}{0.48\linewidth}
        \centering
        \includegraphics[width=\linewidth]{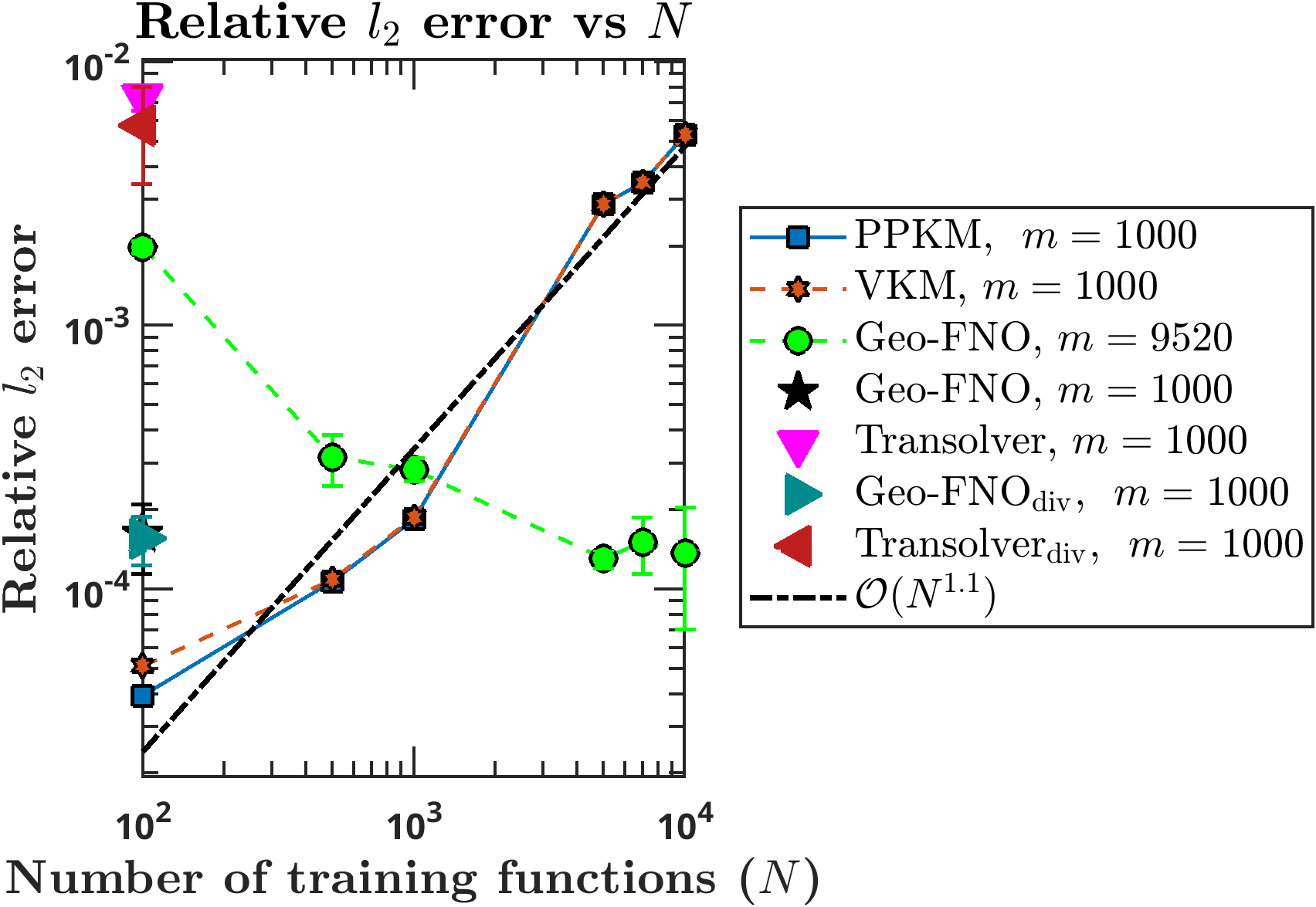}
        \caption{}
        \label{fig:flow_cylinder_laminar_errors}
    \end{subfigure} &
    \begin{subfigure}{0.48\linewidth}
        \centering
        \includegraphics[width=\linewidth]{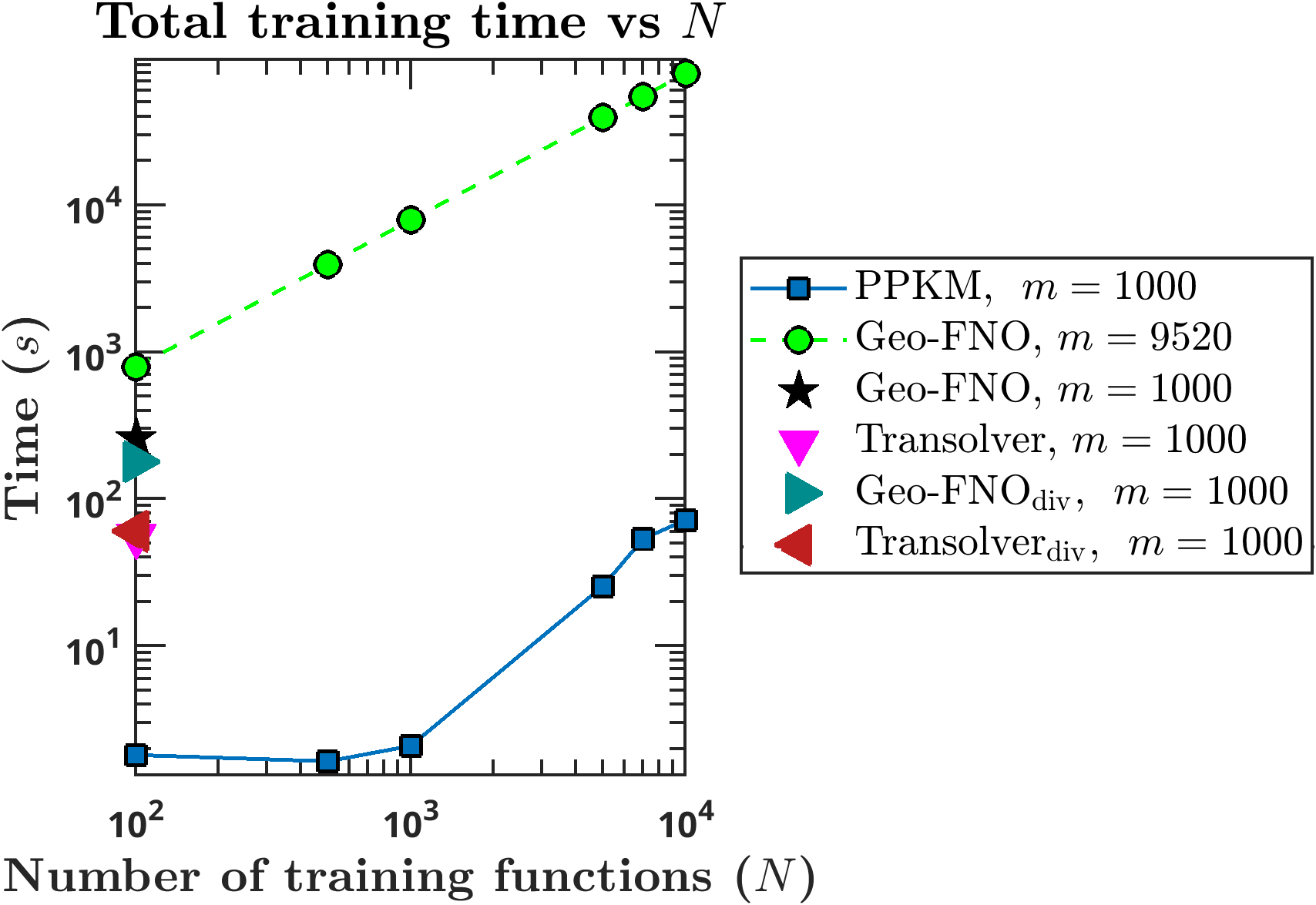}
        \caption{}
        \label{fig:flow_cylinder_laminar_times}
    \end{subfigure}
    \end{tabular}
    \caption{The 2D laminar flow past a cylinder problem. (\textbf{A}) and (\textbf{B}) show examples of an input function (the initial velocity) and an output function (the final velocity), respectively, from the vortex shedding regime. (\textbf{C}) and (\textbf{D}) show the test relative $\ell_2$ errors and training runtimes as functions of $N$ for the vortex shedding regime. (\textbf{E}) and (\textbf{F}) show the same results for the regime without vortex shedding.}
    \label{fig:flow_cylinder}
\end{figure}
This classical problem~\citep{tritton1959experiments, jackson1987finite, li1991numerical, rajani2009numerical} describes laminar fluid flow past a 2D cylinder of radius $0.8$ centered at $(1.5, 10)$ within a rectangular domain $\Omega_a = \Omega_v=[0, 20] \times [0, 14]$ representing a channel. We generated a 2D mesh with $9520$ points such that the refinement near the cylinder was at least twice as fine as it was elsewhere. We prescribed an inlet BC on the left boundary, free-stream BC on the top and bottom boundaries (set to be the same as the inlet velocity), a no-slip BC on the cylinder, and a zero pressure BC on the right boundary. We generated simulations in two flow regimes, with and without vortex shedding~\citep{provansal1987benard} in the wake of the cylinder. The IC was varied across simulations, but always set to a constant velocity everywhere sampled from $\mathcal{N}_{[0.1, 0.26]}(0.18, 0.026)$ for the case without vortex shedding; and from $\mathcal{N}_{[0.45, 0.8]}(0.625, 0.1)$ with vortex shedding. The simulations were run to $T=10$ with a time-step $\Delta t=10^{-3}$. We then learned the operator map $\G: \bu(x, 0) \rightarrow \bu(y, 10)$. The Reynolds numbers for these flow regimes and the relevant configuration details are provided in row~1 in Table~\ref{tab:exp_config}. Figures~\ref{fig:flow_cylinder_input} and~\ref{fig:flow_cylinder_output} show examples of input and output functions, respectively.

The errors reported in Figure~\ref{fig:flow_cylinder_shedding_errors} show that in the vortex shedding case, our method consistently performed the best, with lower errors than the neural operators by up to 2 orders of magnitude for the same $m$ and by up to 2.5 orders of magnitude when the neural operators use a much larger $m$. 
In the flow without any vortex shedding, our method showed errors (Figure~\ref{fig:flow_cylinder_laminar_errors}) increasing with $N$. We suspect that importance sampling of the input functions would restore the proper convergence behavior. 
\textcolor{edit}{In the no vortex shedding flow regime, the Geo-FNO's divergence loss variant improved the accuracy and numerical divergence while the Transolver's divergence loss variant did not. In the vortex shedding flow regime, neither divergence-loss variant improved the accuracy or divergence.}
However, the lowest error was still achieved by our method at $N=100$. Further, Figures~\ref{fig:flow_cylinder_shedding_times} and~\ref{fig:flow_cylinder_laminar_times} show that our method trained orders of magnitude faster than all the baseline methods.
\subsection{2D spacetime Taylor--Green vortices}
\label{sec:taylor_green_spacetime}
\begin{figure}[pos=!htpb]
    \centering
    \begin{tabular}{ccc}
    &
    \begin{subfigure}{0.32\linewidth}
        \centering
        \includegraphics[width=\linewidth]{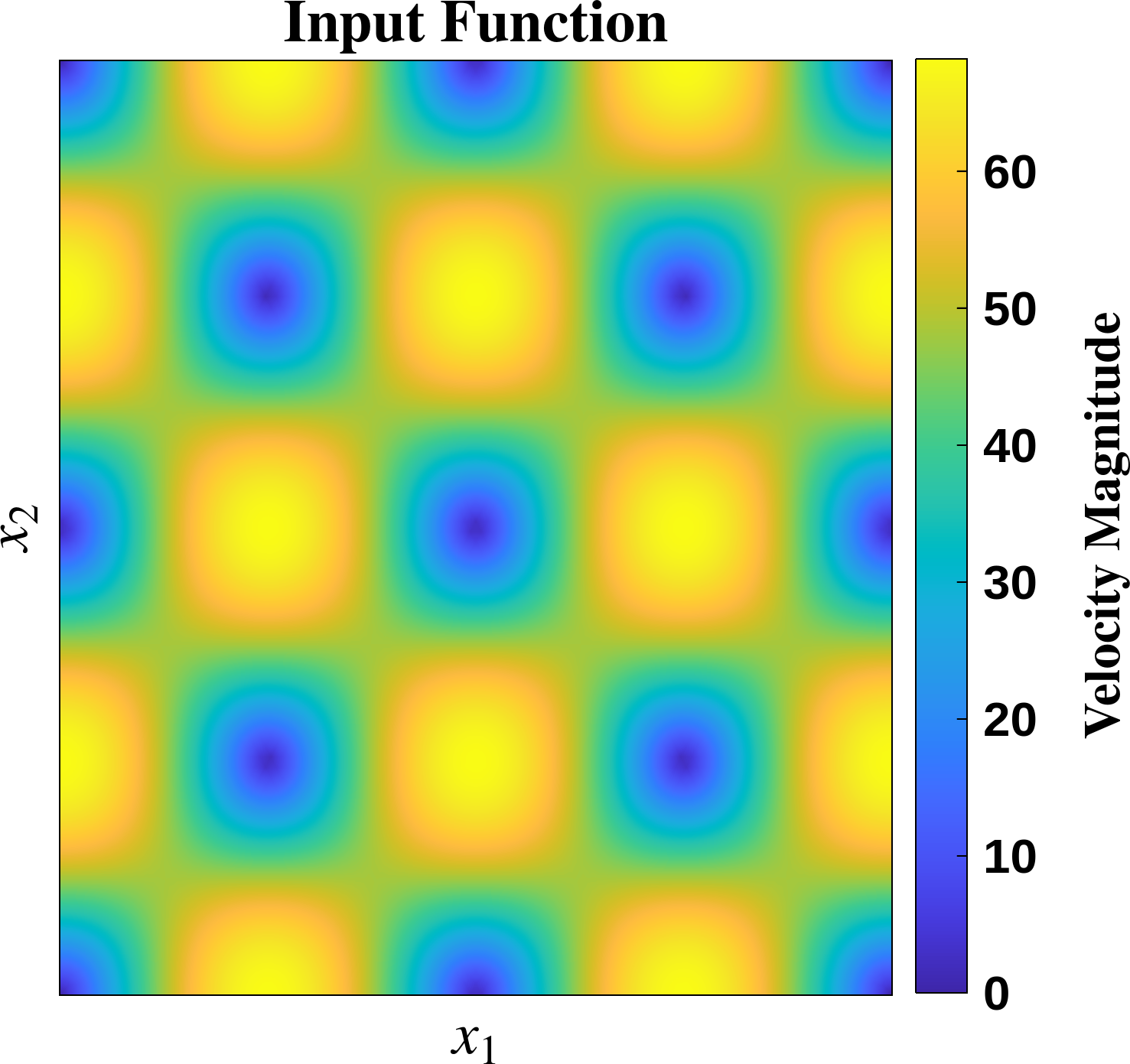}
        \caption{}
        \label{fig:taylor_green_input}
    \end{subfigure} &
    \begin{subfigure}{0.32\linewidth}
        \centering
        \includegraphics[width=\linewidth]{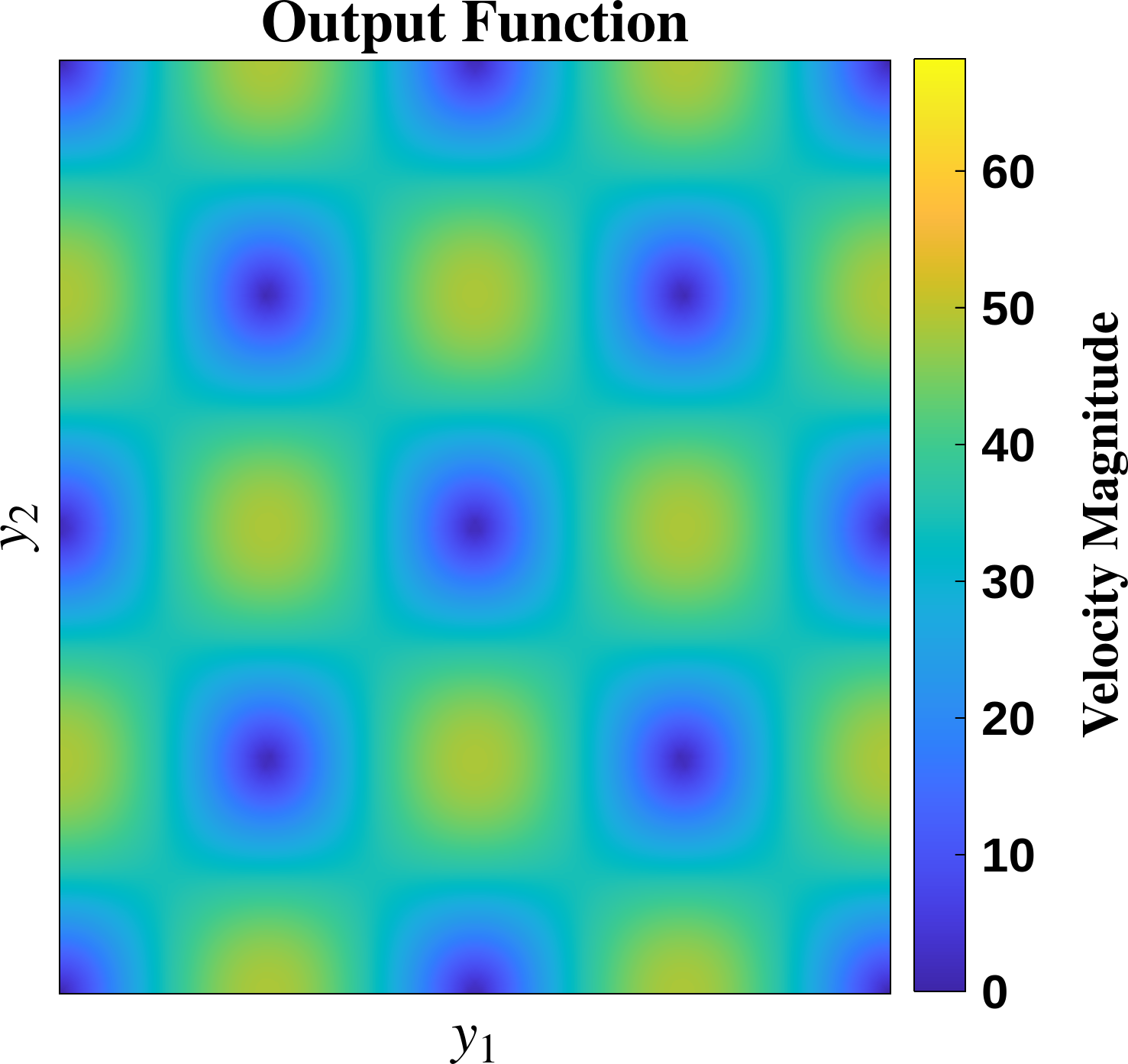}
        \caption{}
        \label{fig:taylor_green_output}
    \end{subfigure} \\
    \raisebox{3ex}{\rotatebox{90}{Initial velocity to final velocity}} &
    \begin{subfigure}{0.48\linewidth}
        \centering
        \includegraphics[width=\linewidth]{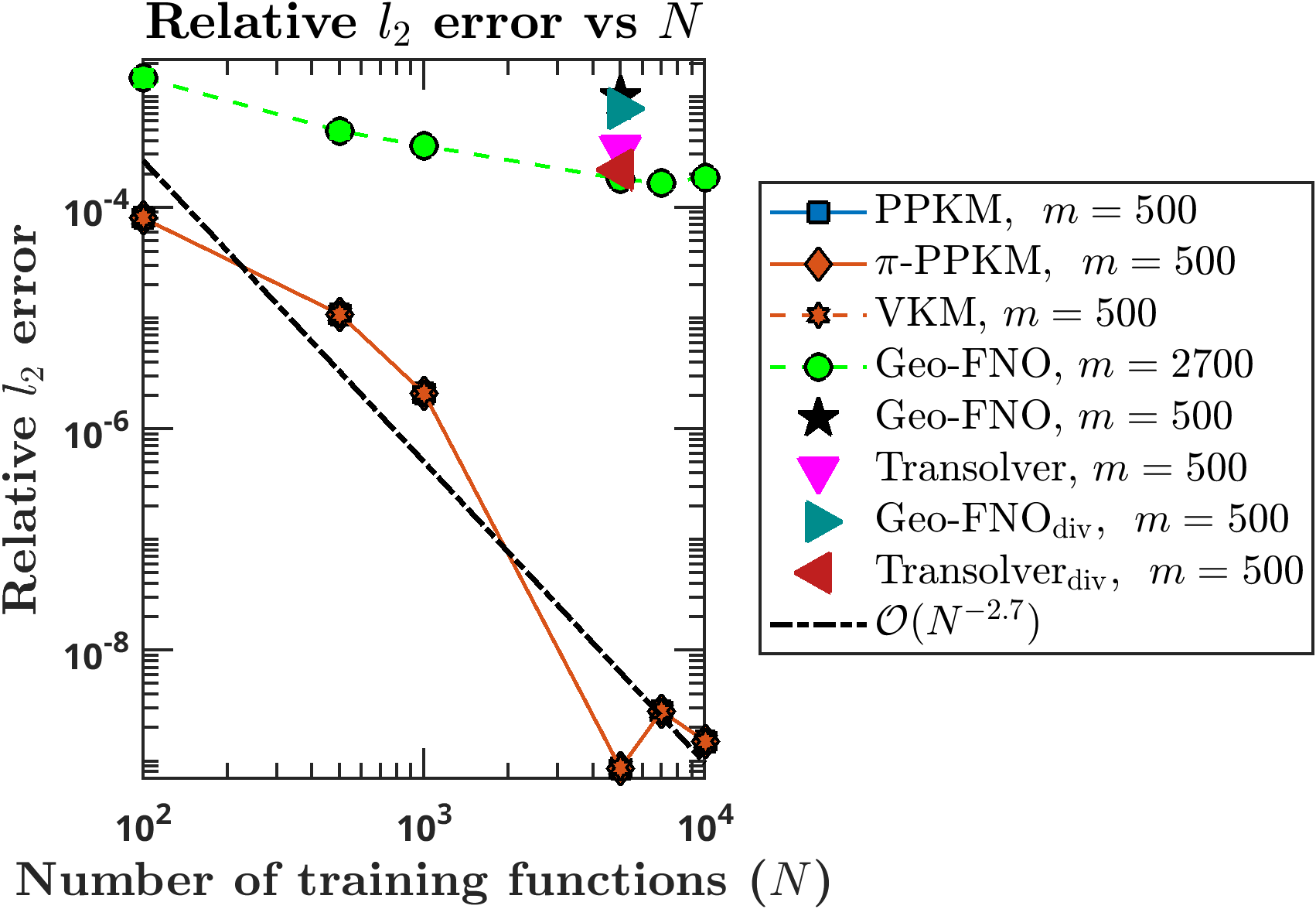}
        \caption{}
        \label{fig:taylor_green_time_errors}
    \end{subfigure} &
    \begin{subfigure}{0.48\linewidth}
        \centering
        \includegraphics[width=\linewidth]{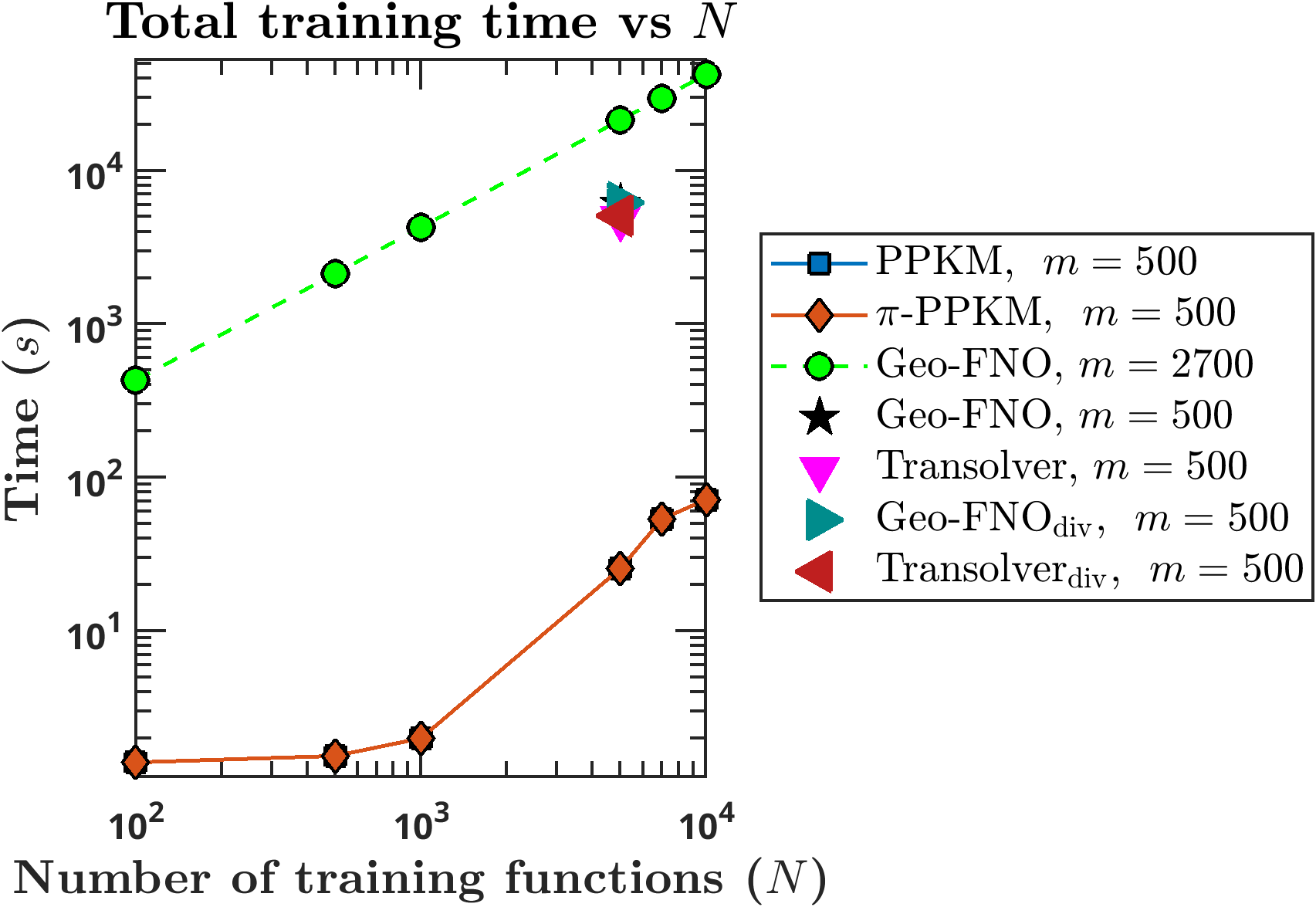}
        \caption{}
        \label{fig:taylor_green_time_times}
    \end{subfigure} \\
    \raisebox{3ex}{\rotatebox{90}{Flow parameters to final velocity}} &
    \begin{subfigure}{0.48\linewidth}
        \centering
        \includegraphics[width=\linewidth]{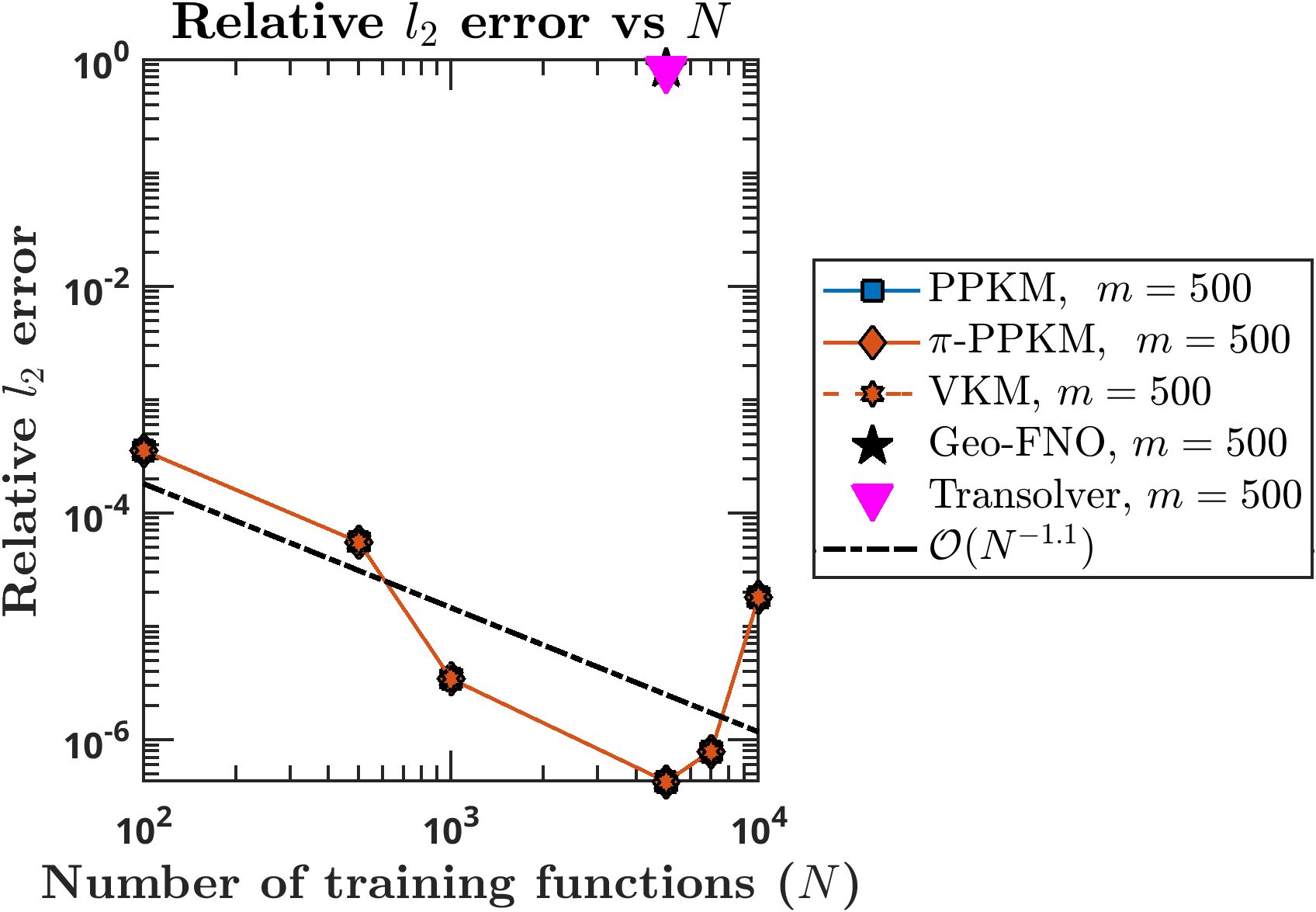}
        \caption{}
        \label{fig:taylor_green_time_coeffs_errors}
    \end{subfigure} &
    \begin{subfigure}{0.48\linewidth}
        \centering
        \includegraphics[width=\linewidth]{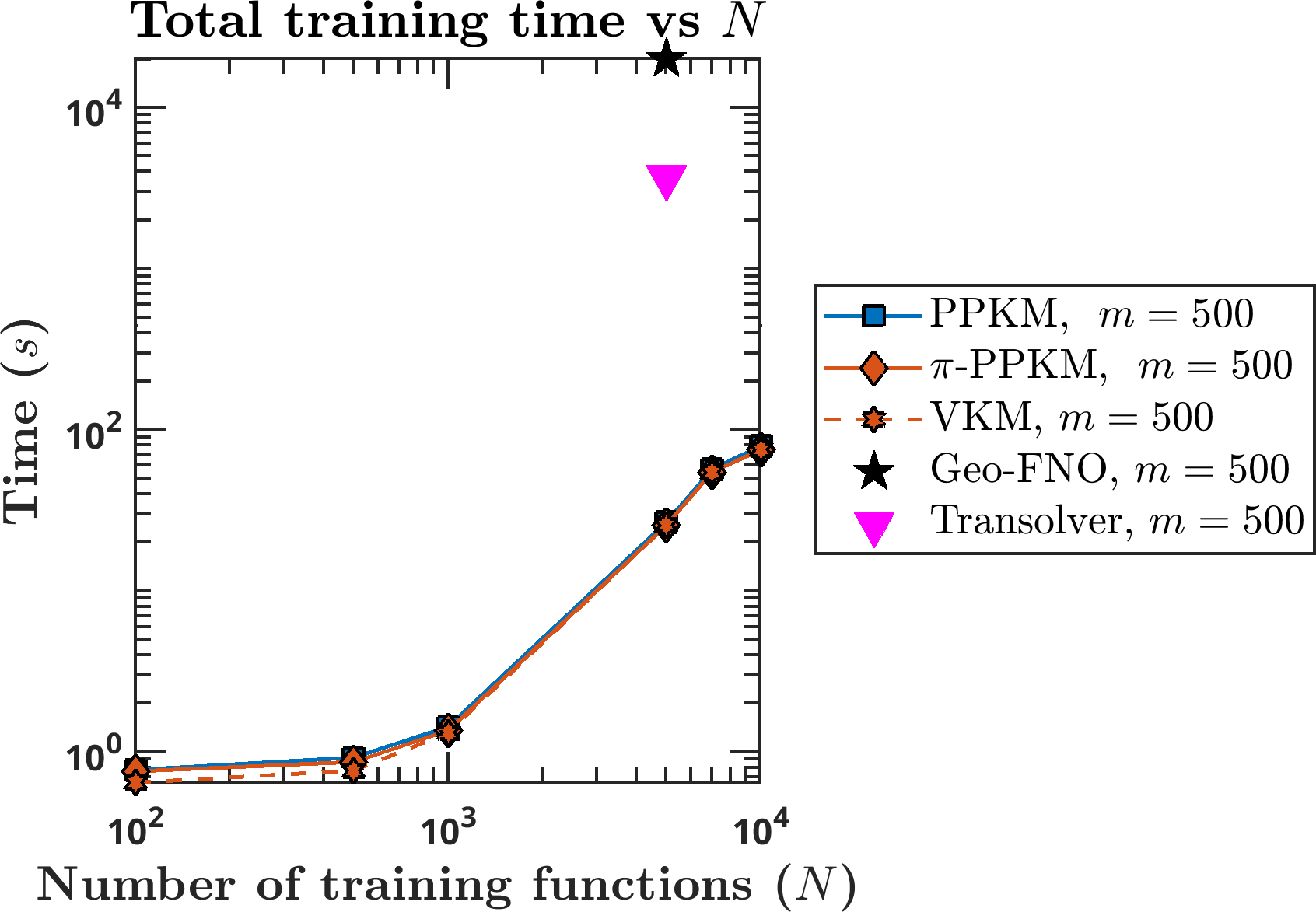}
        \caption{}
        \label{fig:taylor_green_time_coeffs_times}
    \end{subfigure}
    \end{tabular}
    \caption{The 2D laminar Taylor--Green vortices problem for the \textbf{spacetime} operator map. (\textbf{A}) and (\textbf{B}) show examples of an input function (the initial velocity) and an output function (the final velocity) at time $T=1$, respectively. Here, the output functions are snapshots of the velocity field at four timesteps (see Section~\ref{sec:taylor_green_spacetime} for details). (\textbf{C}) and (\textbf{D}) show the test relative $\ell_2$ errors and training runtimes as functions of $N$ for the operator map from the initial velocity to the final velocity at the four timesteps. (\textbf{E}) and (\textbf{F}) show the same results for the operator map from the flow parameters to the final velocity at the four timesteps.}
    \label{fig:taylor_green_spacetime}
\end{figure}
Taylor--Green vortex flow~\citep{taylor1937mechanism, kim1985application} describes the 2D laminar flow of a decaying periodic vortex in the torus $\Omega_a = \Omega_v=[0, 2\pi]^2$ (due to periodic boundary conditions). The flow admits an analytic solution $\bu = [u_1, u_2]$ where $u_1(x) = A \sin(x_1) \cos(x_2) e^{-2 \nu t}$ and $u_2(x) = A \cos(x_1) \sin(x_2) e^{-2 \nu t}$. We randomly sampled $A \sim \mathcal{N}_{[0.1, 80]}(40, 30)$ and $\nu \sim \mathcal{N}_{[0.0001, 1]}(0.006, 0.1)$; for the domain, we used a 2D triangular mesh with $7477$ points. Here, the PPKM analytically enforces both incompressibility and spatial periodicity; see Section \ref{sec:spatial_kernel} for details. We then sought to recover the operator map given by $\G_1: \bu(x, 0) \rightarrow \bu(y, t), \ t \in T$, which maps an initial condition to $\Omega_v \times T$ where $T = [0.7, 0.8, 0.9, 1]$. We also alternatively learned the parametric operator map $\G_2: \tau \rightarrow \bu(y, t), \ t \in T$ where $\Omega_a = \tau \subset \mathbb{R}^2$ is the parameter space of $A$ and $\nu$. We report the relative $\ell_2$ error and the maximum pointwise divergence averaged over the four timesteps. Example input and output (at $T=1$) functions are shown in Figures~\ref{fig:taylor_green_input} and~\ref{fig:taylor_green_output}. Results for the problem focusing on the spatial operator map are shown in Appendix~\ref{sec:taylor_green}.

Due to limitations in computational resources, we were unable to report neural operator errors for different $N$ for $\G_2$ in this problem and in Section~\ref{sec:taylor_green}. Our method significantly outperformed the neural operator baselines in this problem, achieving up to 4.5 orders of magnitude lower errors (Figures~\ref{fig:taylor_green_time_errors} and~\ref{fig:taylor_green_time_coeffs_errors}) and up to 2.5 orders of magnitude faster training times (Figures~\ref{fig:taylor_green_time_times} and~\ref{fig:taylor_green_time_coeffs_times}). Since this problem has no numerical noise in the data set, all variations of the PPKM and the VKM achieve similar errors as they are all interpolatory; however, only the PPKM analytically satisfies periodicity and incompressibility. 
\textcolor{edit}{The Geo-FNO's divergence loss variant substantially improved the accuracy, and both the divergence-loss variants substantially improved the numerical divergence in this problem.}
\subsection{3D species transport}
\label{sec:species_transport}
\begin{figure}[pos=!htpb]
    \centering
    \begin{subfigure}{0.3\linewidth}
        \centering
        \includegraphics[width=\linewidth]{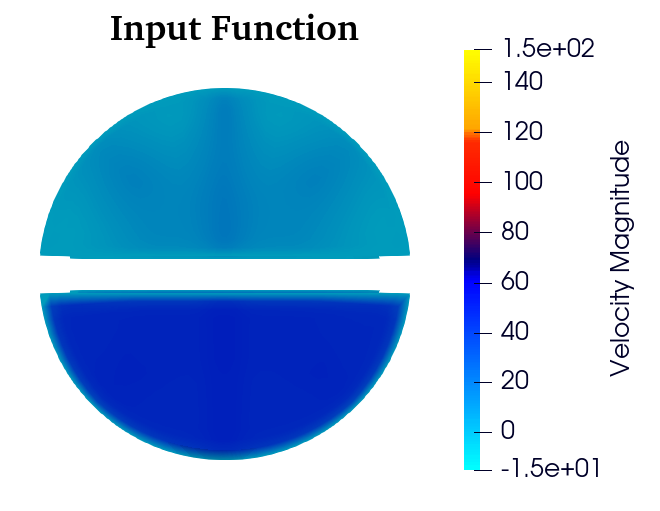}
        \caption{}
        \label{fig:species_transport_input}
    \end{subfigure} \\
    \begin{subfigure}{0.65\linewidth}
        \centering
        \includegraphics[width=\linewidth]{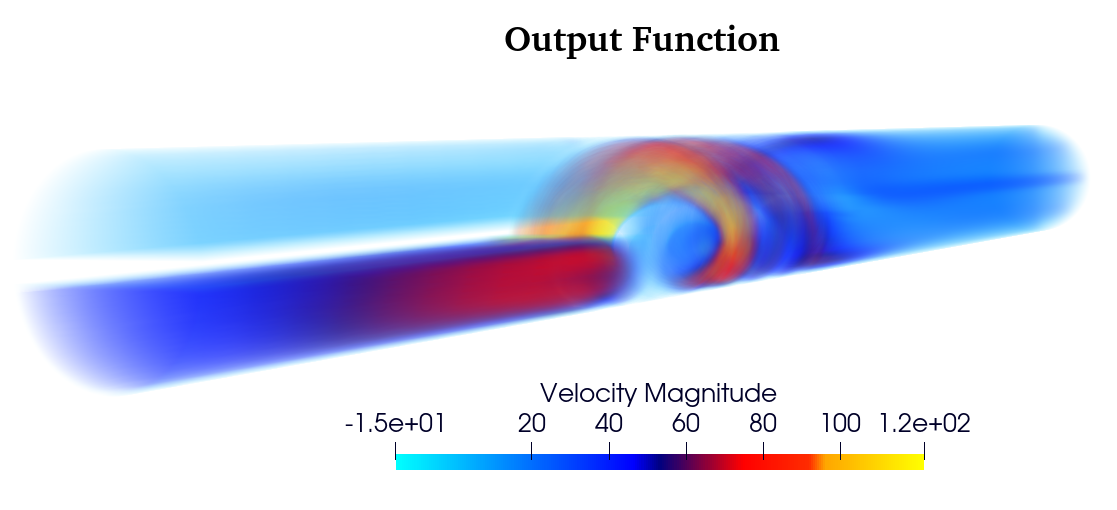}
        \caption{}
        \label{fig:species_transport_output}
    \end{subfigure} \\
    \begin{subfigure}{\linewidth}
        \centering
        \begin{tikzpicture}[>=Stealth, thick, scale=0.9, transform shape]
            \draw[black, very thick] (0,0) -- (10,0);  
            \draw[black, very thick] (10,-0.02) -- (10,2.02); 
            \draw[black, very thick] (10,2) -- (0,2);  
            \draw[black, very thick] (0,2.02) -- (0,1.18); 
            \draw[black, very thick] (0,0.82) -- (0,-0.02); 

            \draw[black, very thick] (0,1.2) -- (1.2,1.2);
            \draw[black, very thick] (0,0.8) -- (1.2,0.8);

            \draw[->, black, very thick] (-0.5,1.5) -- (0.5,1.5);
            \node[left, black] at (-0.5,1.5) {Inlet gas};

            \draw[->, black, very thick] (9.5,1.5) -- (10.4,1.5);
            \node[right, black] at (10.4,1.5) {Outlet};

            \draw[->, black, very thick] (-0.5,0.5) -- (0.5,0.5);
            \node[left, black] at (-0.5,0.5) {Inlet air};

            \draw[black, very thick] (3,2) -- (3,0);
            \node[above] at (3,2.1) {Blade 1};

            \draw[black, very thick] (5,2) -- (5,0);
            \node[above] at (5,2.1) {Blade 2};

            \draw[black, very thick] (7,2) -- (7,0);
            \node[above] at (7,2.1) {Blade 3};

            \draw[black, thick, dotted] (4,2) -- (4,0);
            \draw[black, thick, dotted] (6,2) -- (6,0);
            \draw[black, thick, dotted] (8,2) -- (8,0);

            \node[black] at (9,1.1) {Inner walls};
            \draw[black, ->] (8.5,0.8) -- (8.1,0.8);

            \node[black] at (5,-0.9) {Outer walls};
            \draw[->, black, very thick] (5,-0.7) -- (5,-0.01);
        \end{tikzpicture}
        \caption{}
        \label{fig:species_transport_domain}
    \end{subfigure} \\
    \begin{tabular}{cc}
    \begin{subfigure}{0.5\linewidth}
        \centering
        \includegraphics[width=\linewidth]{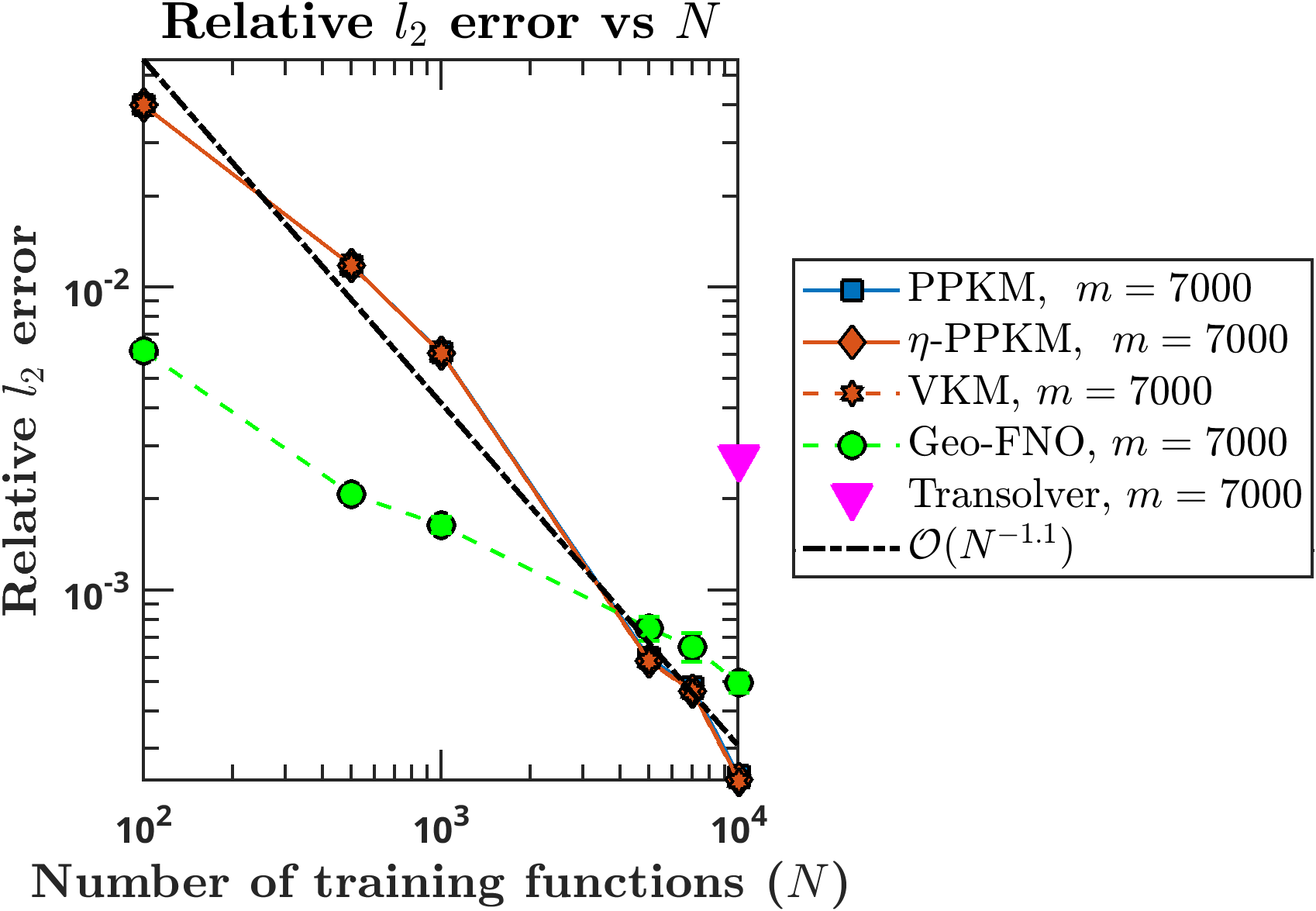}
        \caption{}
        \label{fig:species_transport_errors}
    \end{subfigure} &
    \begin{subfigure}{0.5\linewidth}
        \centering
        \includegraphics[width=\linewidth]{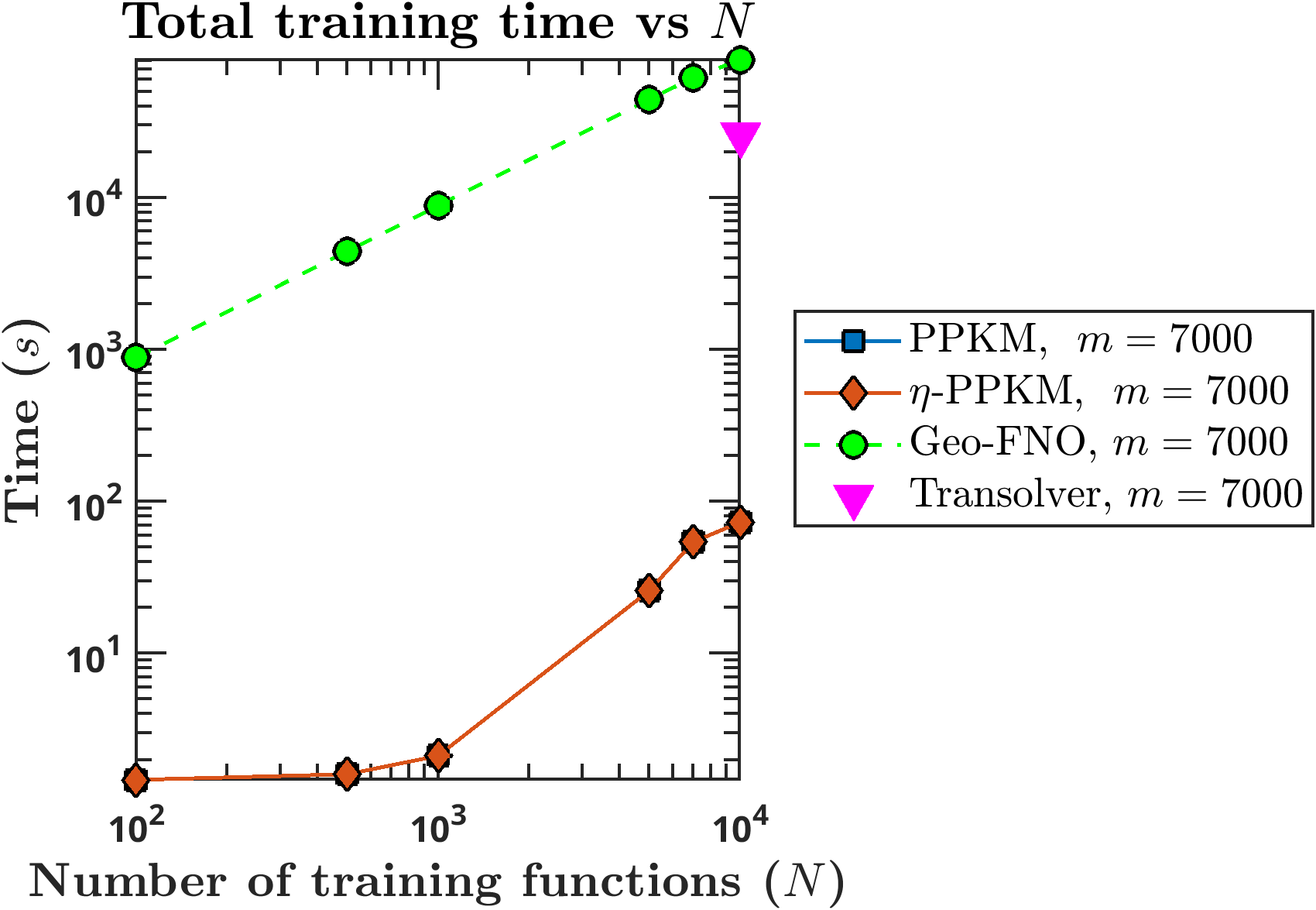}
        \caption{}
        \label{fig:species_transport_times}
    \end{subfigure}
    \end{tabular}
    \caption{The 3D turbulent species transport example. (\textbf{A}) and (\textbf{B}) show examples of an input function (the inlet velocity of the two gaseous species at the plane $z=0$) and an example output function (the final velocity), respectively. (\textbf{C}) shows a top-down \texttt{yz} view of the domain and its placement of the three blades and the inner walls. (\textbf{D}) and (\textbf{E}) show the test relative $\ell_2$ errors and training runtimes as functions of $N$.}
    \label{fig:species_transport}
\end{figure}
\begin{figure}[pos=!htpb]
    \centering
    \includegraphics[width=0.85\linewidth]{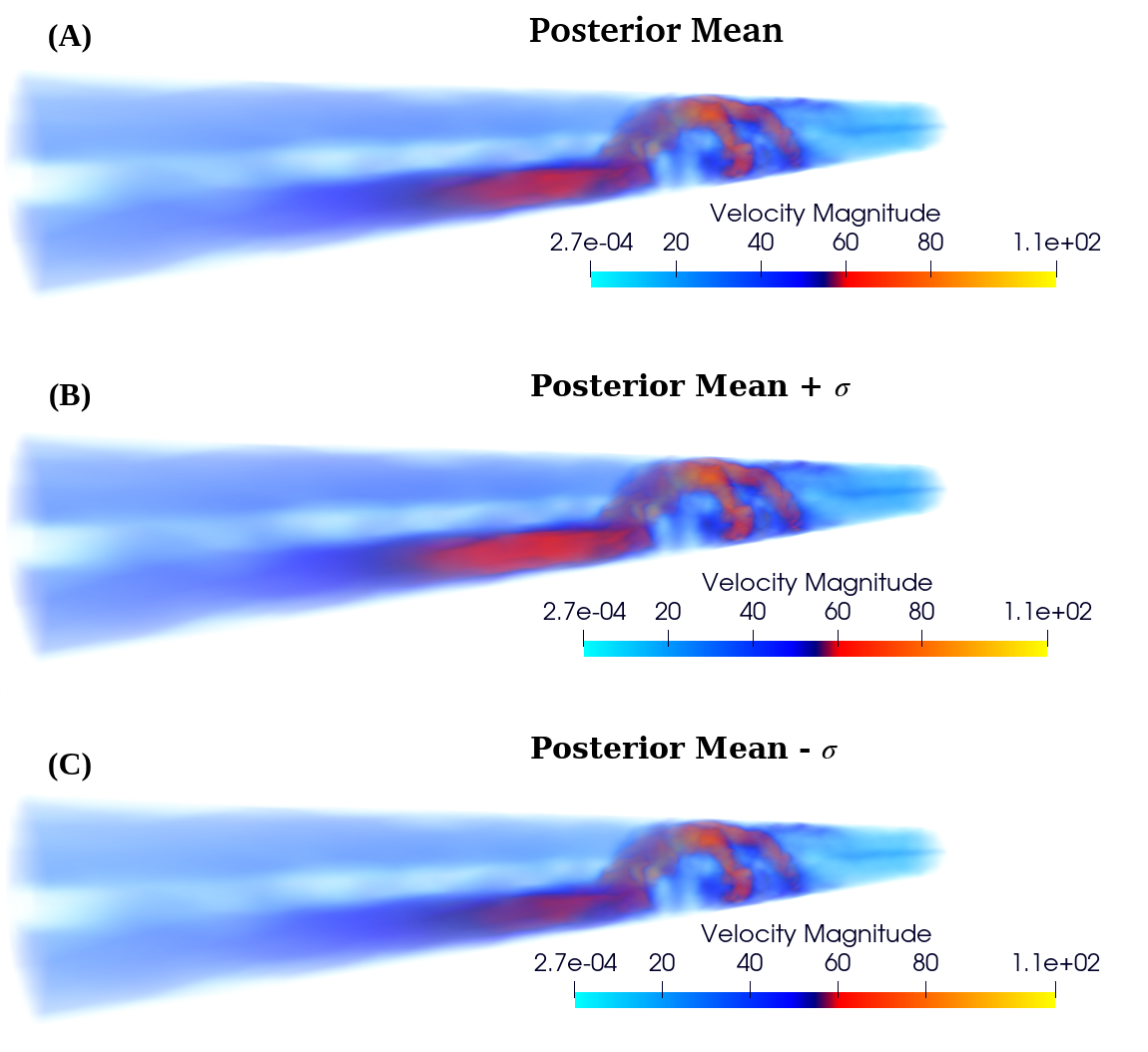}
    \caption{Uncertainty quantification for the 3D turbulent species transport example. (\textbf{A}) shows the posterior mean velocity field. (\textbf{B}) and (\textbf{C}) show the velocity fields corresponding to the one standard deviation of the posterior predictive distribution $\xi$. Section~\ref{sec:gp_results} provides a discussion of these results.}
    \label{fig:uqfigure}
\end{figure}
Next, we investigated operator learning on a 3D turbulent flow problem~\footnote{\url{https://su2code.github.io/tutorials/Inc_Species_Transport_Composition_Dependent_Model}} inspired by the work of~\citet{ubal2024adjoint}. This problem models the mixing of air and methane in a static Kenics mixer~\citep{ubal2024adjoint} with three perpendicular blades twisted along the \texttt{z} axis; see Figure~\ref{fig:species_transport_domain}. 
At the $z=0$ plane, we specified the air's velocity at $y<0$ and methane's velocity at $y>0$, sampled from $\mathcal{N}_{[1, 40]}(20, 20)$; the velocity everywhere else was initialized to zero. The three blades, inner wall, and outer wall used no-slip BCs; the outlet boundary at the plane $z=0.26$ used a zero pressure BC. We used a maximum of $7000$ approximate kernel Fekete points due to memory constraints. We ran the simulation until $T=0.5$ with $\Delta t=0.005$. We learned the operator $\G: \bu^{\mathrm{inlet}} \rightarrow \bu(y, 0.5)$.
We show examples of input and output functions in Figures~\ref{fig:species_transport_input} and~\ref{fig:species_transport_output}, respectively.

We report the errors in Figure~\ref{fig:species_transport_errors}. The Geo-FNO achieved lower errors at smaller values of $N$, but was overtaken by the PPKM for larger $N$. Despite the added computational cost due to the 3D problem domain, our method was orders of magnitude faster to train than the neural operators (as shown in Figure~\ref{fig:species_transport_times}). \textcolor{edit}{The divergence-loss variants' loss blew up in this problem across all tested hyperparameters.}
\subsubsection{Uncertainty quantification}
\label{sec:gp_results}
We also demonstrate the GP-based uncertainty quantification capabilities of our framework here, described in Section~\ref{sec:uq}. We solved~\eqref{eq:gp_posterior_mean} and~\eqref{eq:gp_posterior_covariance} to compute the posterior mean $\overline{\bb}$ and covariance $\Sigma^\star$, respectively, of $\xi$. Then, we applied a Cholesky factorization to the covariance; $\Sigma^\star = L L^T$. Finally, we used the $\Phi^\eta$ kernel to obtain velocity fields corresponding to $\pm 1$ standard deviation with the perturbed coefficients $\overline{\bb} + L z$ and $\overline{\bb} - L z$. For one of the 200 test functions, we show the posterior mean velocity field in Figure~\ref{fig:uqfigure}A and the corresponding velocity fields within $\pm 1$ standard deviation in Figures~\ref{fig:uqfigure}B and~\ref{fig:uqfigure}C. The maximum spatial variance for this test function within one standard deviation was $10^{-5}$. All three velocity fields are divergence-free and encoded with turbulence power laws by construction.
\subsection{3D forced homogeneous isotropic turbulence}
\label{sec:forced_turb}
\begin{figure}[pos=!htpb]
    \centering
    \begin{tabular}{cc}
    \begin{subfigure}{0.4\linewidth}
        \centering
        \includegraphics[width=\linewidth]{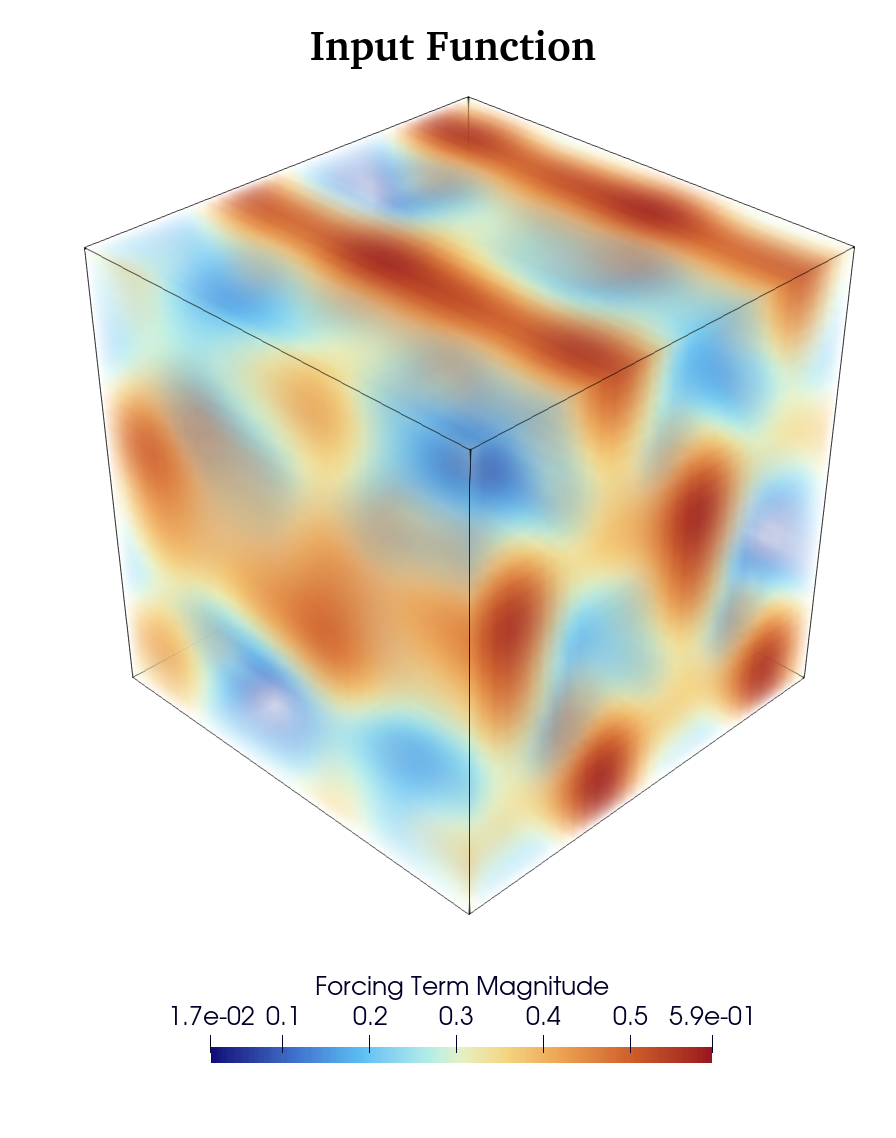}
        \caption{}
        \label{fig:forced_turb_input}
    \end{subfigure} &
    \begin{subfigure}{0.4\linewidth}
        \centering
        \includegraphics[width=\linewidth]{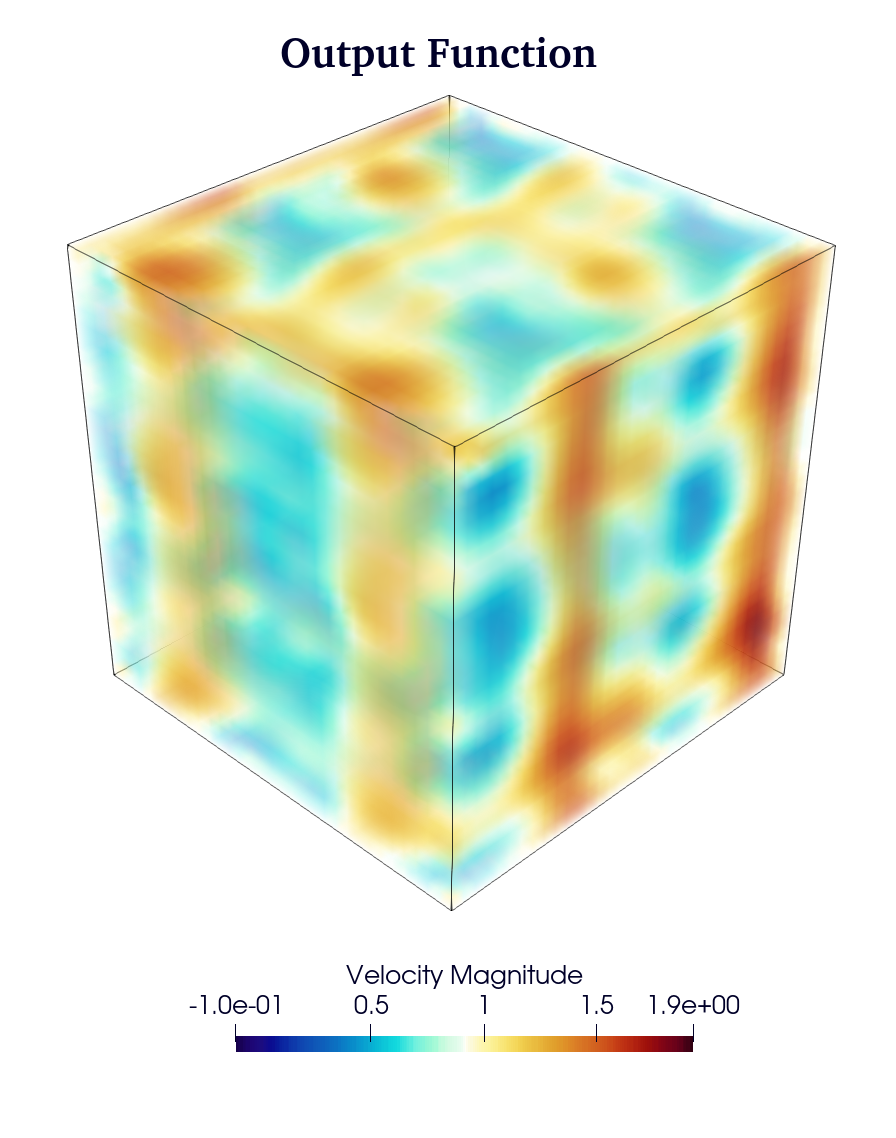}
        \caption{}
        \label{fig:forced_turb_output}
    \end{subfigure} \\
    \begin{subfigure}{0.5\linewidth}
        \centering
        \includegraphics[width=\linewidth]{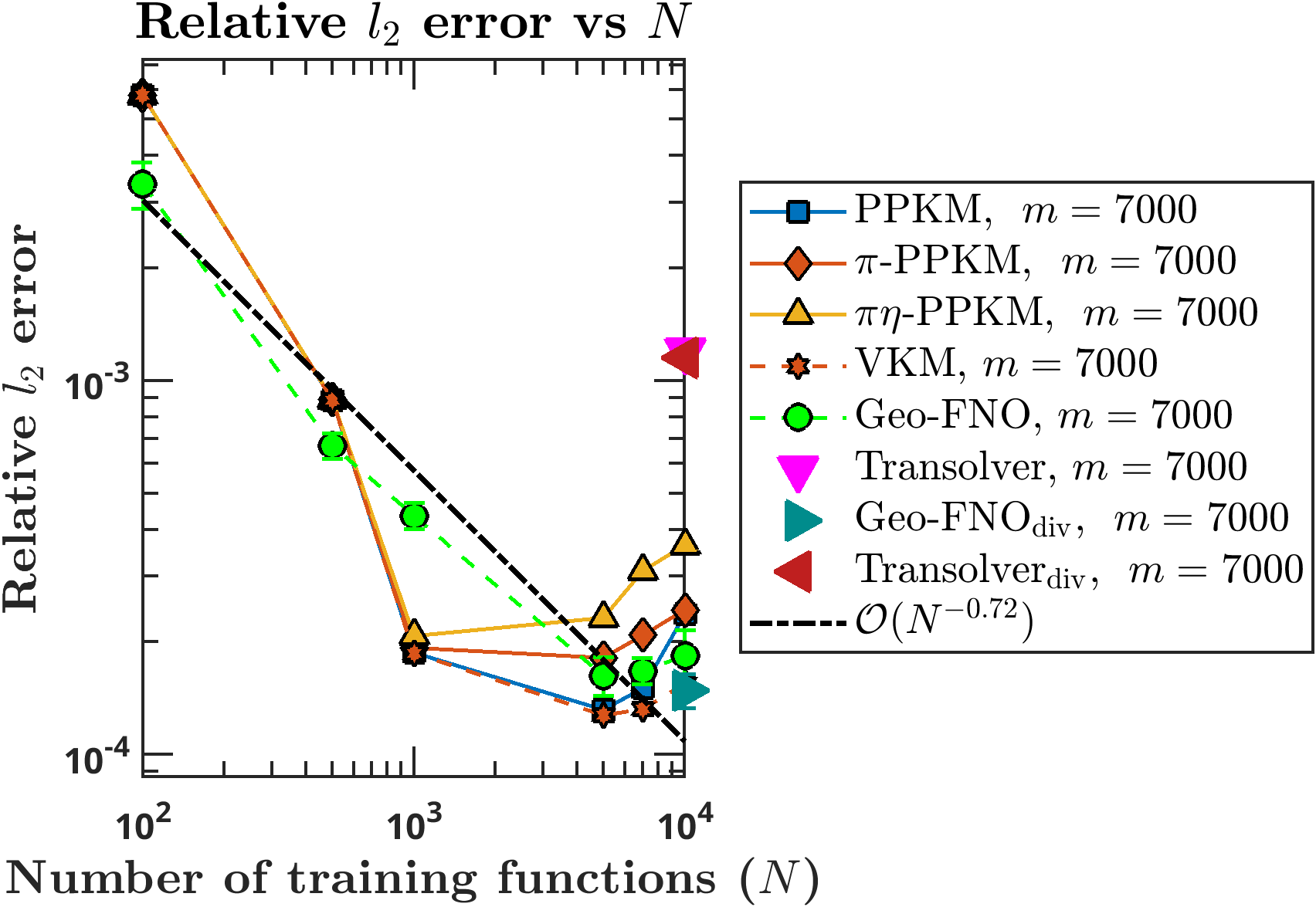}
        \caption{}
        \label{fig:forced_turb_errors}
    \end{subfigure} &
    \begin{subfigure}{0.5\linewidth}
        \centering
        \includegraphics[width=\linewidth]{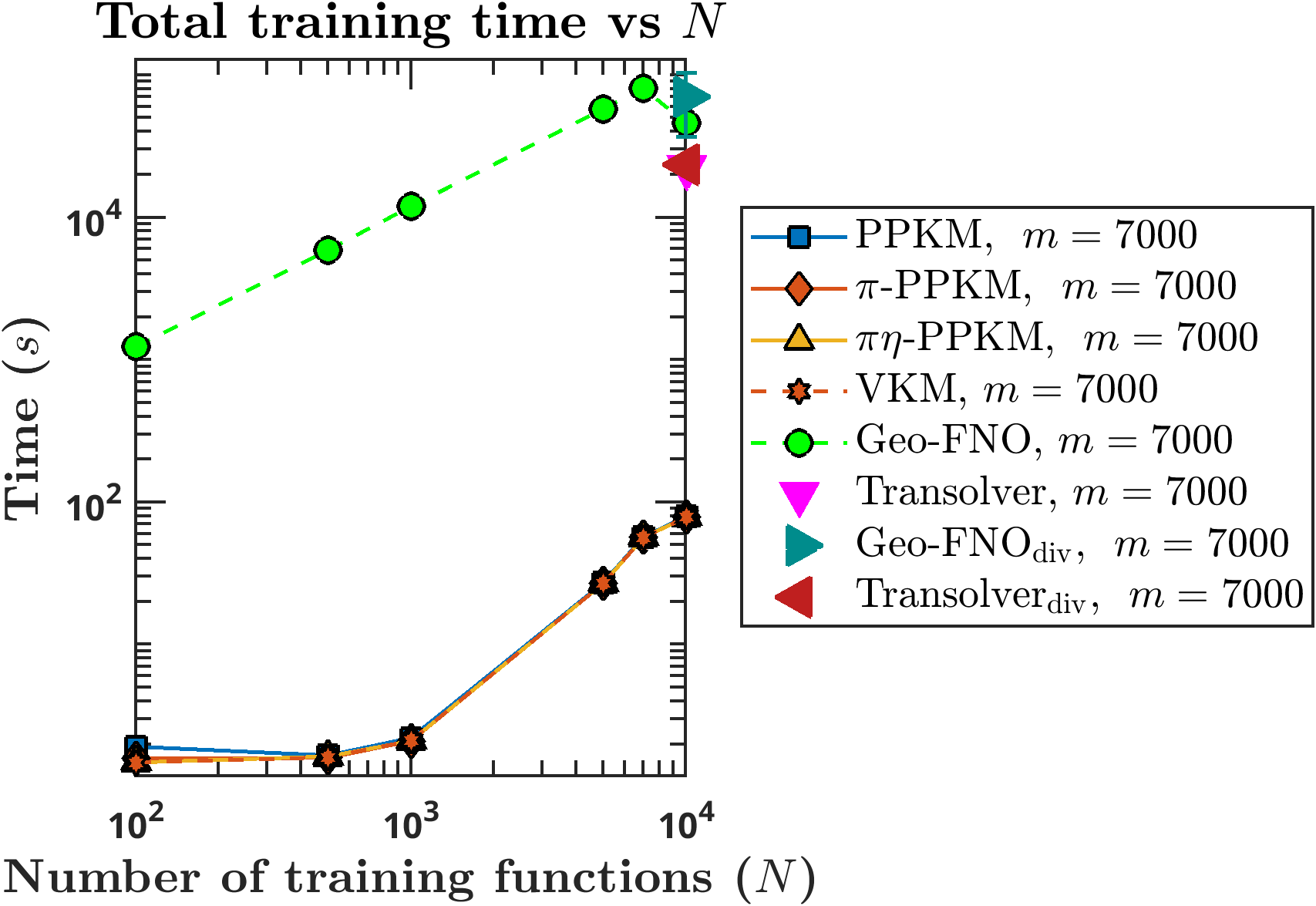}
        \caption{}
        \label{fig:forced_turb_times}
    \end{subfigure}
    \end{tabular}
    \caption{The 3D forced homogeneous isotropic turbulence example. (\textbf{A}) and (\textbf{B}) show examples of an input function (the spatial forcing term) and an example output function (the final velocity), respectively. (\textbf{C}) and (\textbf{D}) show the test relative $\ell_2$ errors and training runtimes as functions of $N$.}
    \label{fig:forced_turb}
\end{figure}
\textcolor{edit}{Finally, we considered the 3D forced homogeneous isotropic turbulence problem in the triply periodic domain $\Omega=[0,2\pi]^3$. Unlike the previous problem, which uses RANS, we employed a large-eddy simulation (LES) in OpenFOAM~\footnote{\url{https://openfoam.org}} with the Smagorinsky subgrid-scale (SGS) model~\citep{smagorinsky1963general}. LES resolves the larger turbulent structures while the Smagorinsky model introduces an eddy viscosity to represent the dissipative effects of unresolved small-scale motions~\cite{meneveau2000scale}. We discretized $\Omega$ using a uniform finite-volume grid with $32$ cells in each spatial direction and prescribed the fixed periodic, divergence-free, multiscale initial velocity field
\begin{align}
u_x &= \sin(p_1 x_3)+\cos(p_1 x_2)
+0.5\left(\sin(p_2 x_3)+\cos(p_2 x_2)\right)
+0.25\left(\sin(p_3 x_3)+\cos(p_3 x_2)\right), \\
u_y &= \sin(p_1 x_1)+\cos(p_1 x_3)
+0.5\left(\sin(p_2 x_1)+\cos(p_2 x_3)\right)
+0.25\left(\sin(p_3 x_1)+\cos(p_3 x_3)\right), \\
u_z &= \sin(p_1 x_2)+\cos(p_1 x_1)
+0.5\left(\sin(p_2 x_2)+\cos(p_2 x_1)\right)
+0.25\left(\sin(p_3 x_2)+\cos(p_3 x_1)\right),
\end{align}
where $p_1=2$, $p_2=3$, and $p_3=5$. Each simulation used a different spatially varying, time-independent, divergence-free forcing field
\begin{align}
\bf(x)=
\begin{bmatrix}
A\left(\sin(q_1 x_3)+\cos(q_2 x_2)\right) \\
B\left(\sin(q_1 x_1)+\cos(q_2 x_3)\right) \\
C\left(\sin(q_1 x_2)+\cos(q_2 x_1)\right)
\end{bmatrix},
\end{align}
where the amplitudes $A,B,C$ are sampled independently from $\mathcal{N}_{[-0.3,0.3]}(0,0.15)$, while $q_1$ and $q_2$ are sampled uniformly without replacement from $\{1,2,3\}$. The low wavenumber forcing injects energy primarily at large spatial scales, which is transferred toward smaller scales through nonlinear turbulent interactions. Periodic boundary conditions are imposed in all three directions, with time-step $\Delta t=2\times10^{-3}$ and final time $T=0.5$. We learned the operator map $\G:\bf(x) \rightarrow \bu(y, 0.5)$. The Reynolds number for this flow and the relevant configuration details are provided in row~9 in Table~\ref{tab:exp_config}. We show examples of input and output functions in Figures~\ref{fig:forced_turb_input} and~\ref{fig:forced_turb_output}, respectively.}

\textcolor{edit}{We report the errors in Figure~\ref{fig:forced_turb_errors}. Similar to the previous problem, the Geo-FNO achieved lower errors at smaller $N$, but was overtaken by the PPKM variants as $N$ increased. At the largest $N$, the Geo-FNO models were the most accurate, though the PPKM variants achieved errors in the same order of magnitude. The divergence-loss variants did not improve the accuracy or numerical divergence noticeably. Figure~\ref{fig:forced_turb_times} shows the training time across $N$; the PPKM and its variants train orders of magnitude faster than the neural operators.}
\subsection{Effect of the ridge parameter $\theta$ on accuracy}
\label{sec:nugget_effect}
\begin{figure}[pos=!htpb]
    \centering
    \begin{tabular}{cc}
    \begin{subfigure}{0.49\linewidth}
        \centering
        \includegraphics[width=\linewidth]{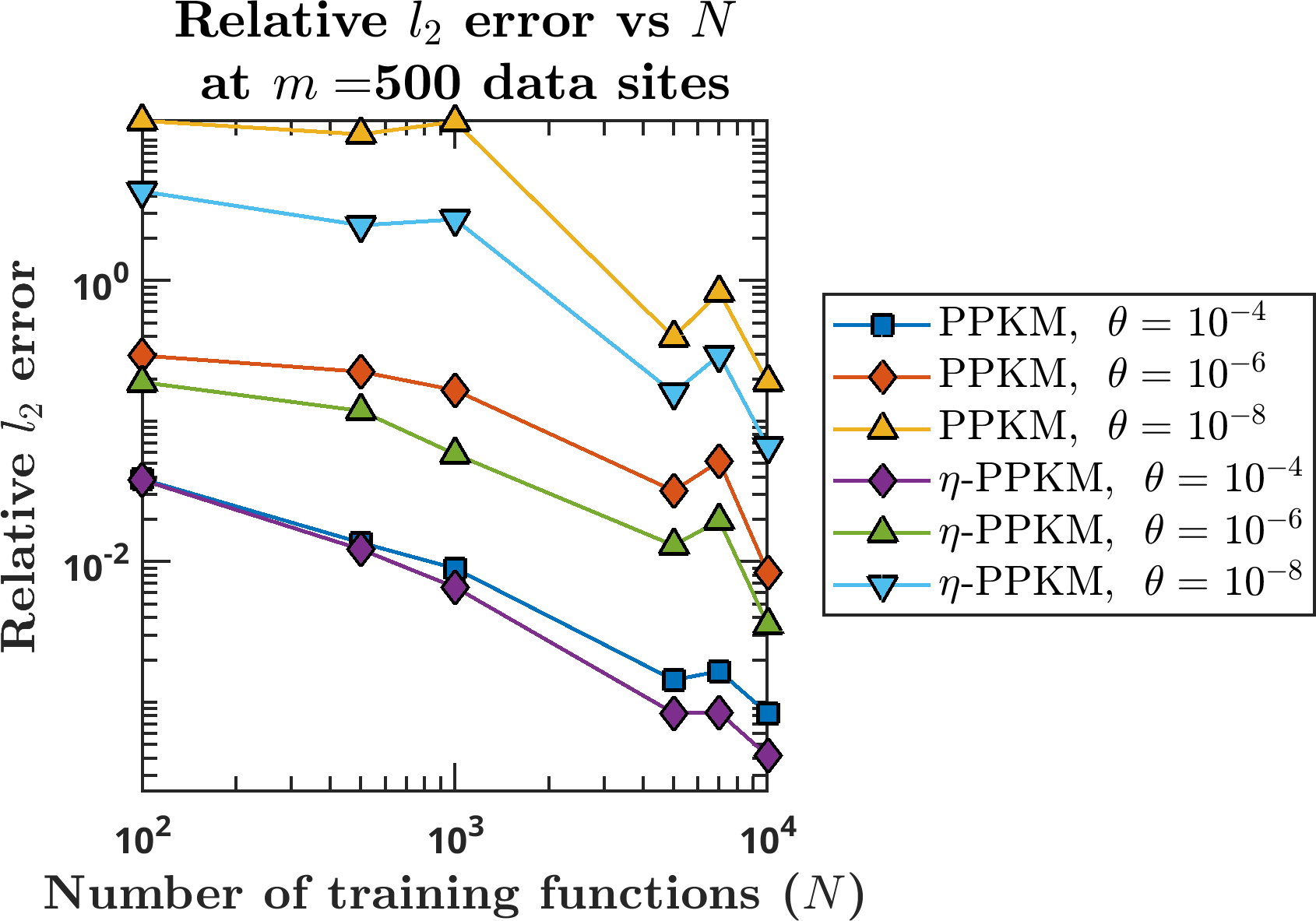}
        \caption{}
        \label{fig:species_transport_ridge_500}
    \end{subfigure} &
    \begin{subfigure}{0.49\linewidth}
        \centering
        \includegraphics[width=\linewidth]{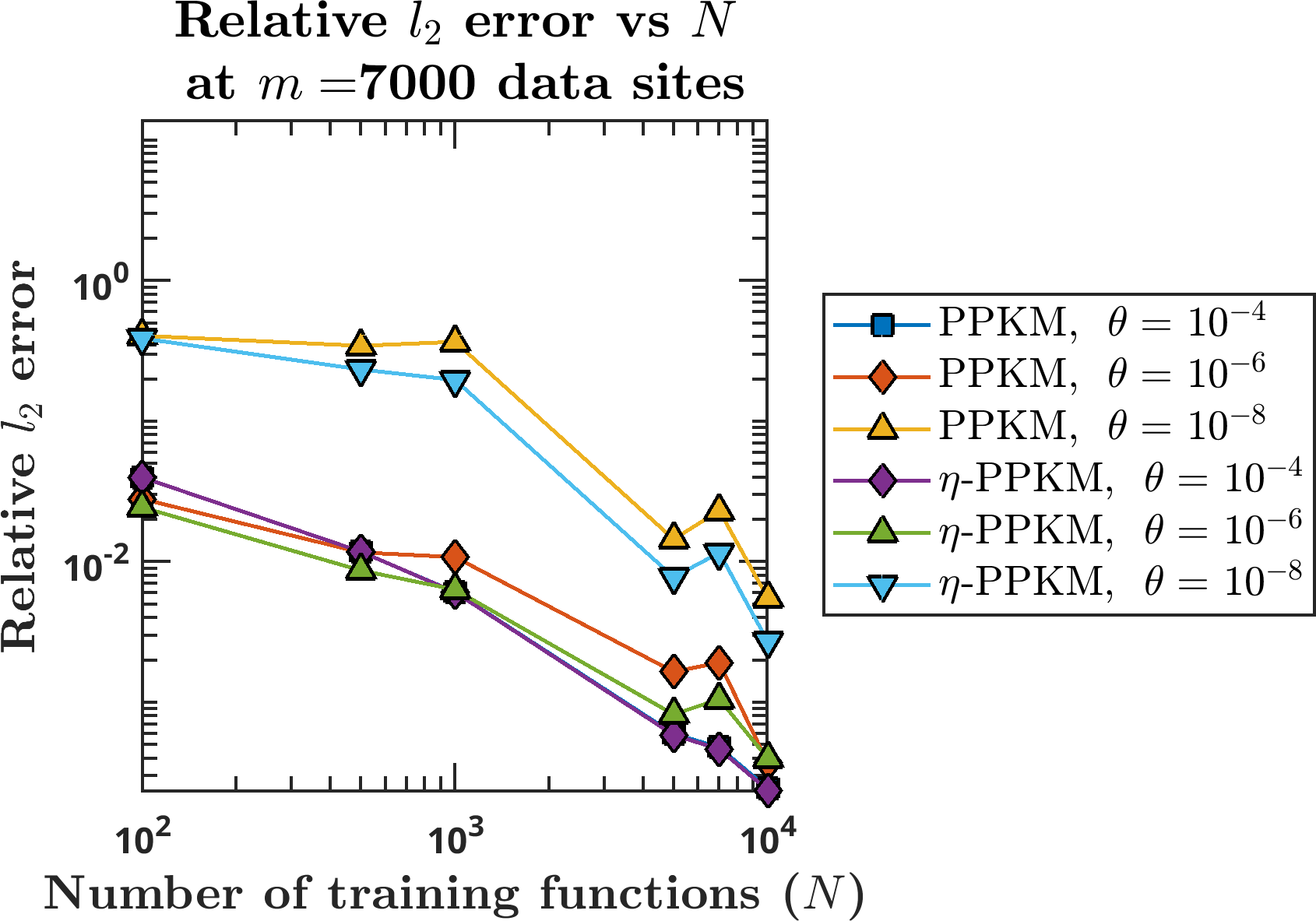}
        \caption{}
        \label{fig:species_transport_ridge_7000}
    \end{subfigure}
    \end{tabular}
    \caption{The effect of the magnitude of $\theta$ in the 3D species transport problem. We tested using $\theta=10^{-4}$, $\theta=10^{-6}$, and $\theta=10^{-8}$. (\textbf{A}) and (\textbf{B}) show the test relative $\ell_2$ errors at 500 and 7000 points, respectively, using different $\theta$.}
    \label{fig:species_transport_nugget_exp}
\end{figure}
\begin{figure}[pos=!htpb]
    \centering
    \begin{tabular}{cc}
    \begin{subfigure}{0.49\linewidth}
        \centering
        \includegraphics[width=\linewidth]{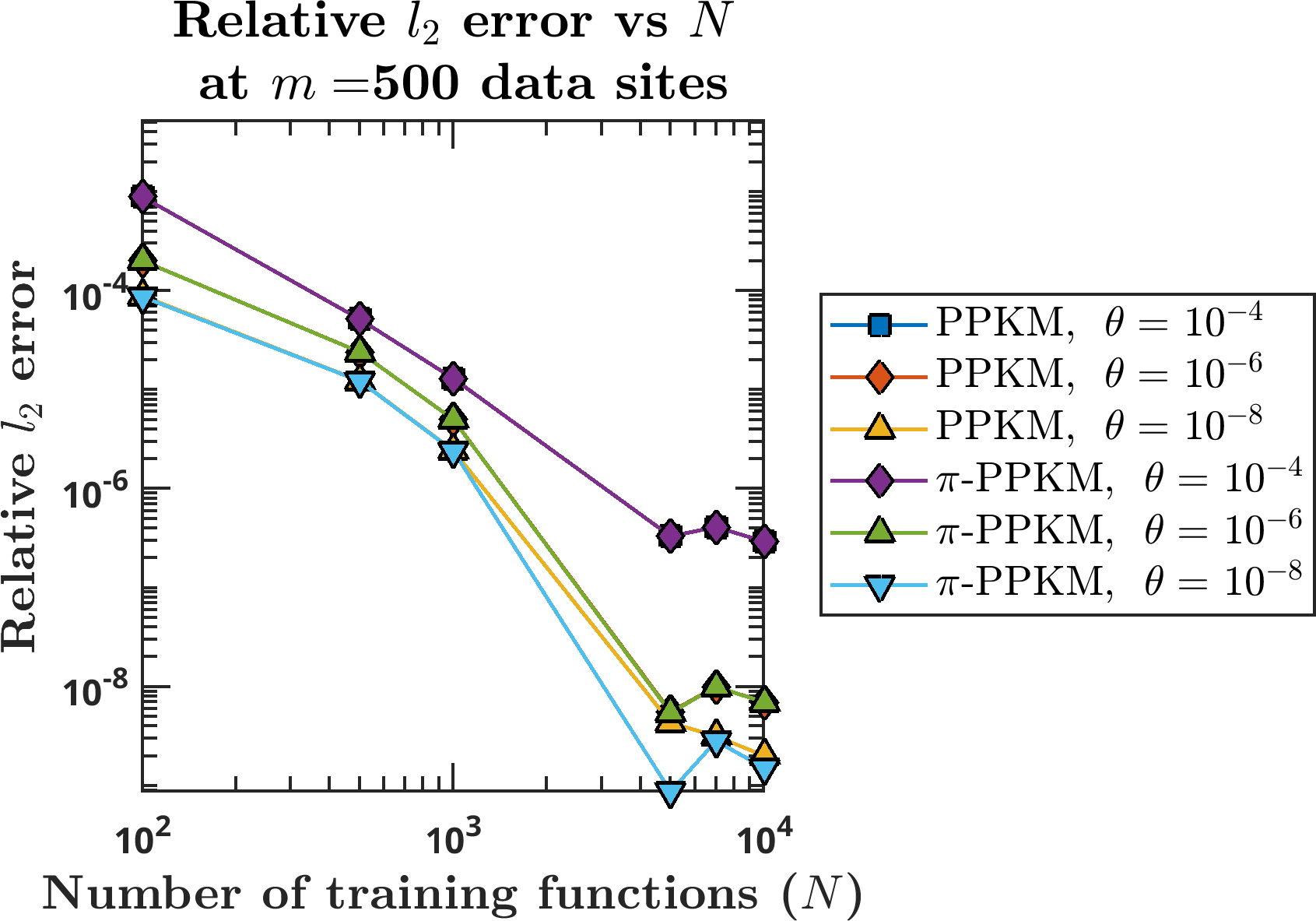}
        \caption{}
        \label{fig:taylor_green_ridge_500}
    \end{subfigure} &
    \begin{subfigure}{0.49\linewidth}
        \centering
        \includegraphics[width=\linewidth]{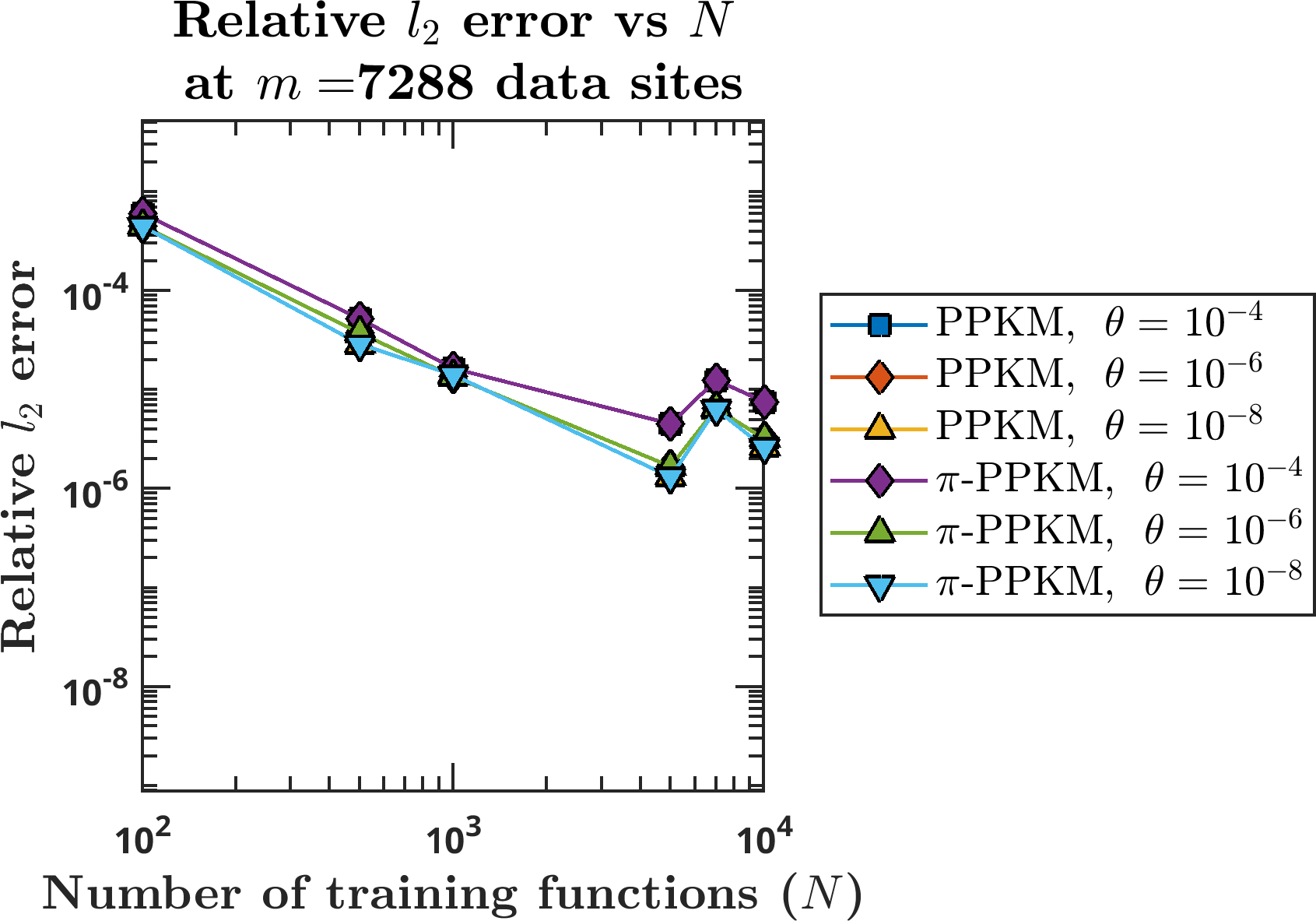}
        \caption{}
        \label{fig:taylor_green_ridge_7288}
    \end{subfigure}
    \end{tabular}
    \caption{The effect of the magnitude of $\theta$ in the 2D laminar Taylor--Green vortices problem for the purely spatial operator map. We tested using $\theta=10^{-4}$, $\theta=10^{-6}$, and $\theta=10^{-8}$. (\textbf{A}) and (\textbf{B}) show the test relative $\ell_2$ errors at 500 and 7288 points, respectively, using different $\theta$.}
    \label{fig:taylor_green_nugget_exp}
\end{figure}
As described in Section~\ref{sec:common_details}, we added a small regularization parameter (``nugget'') $\theta$ to the diagonal of the operator kernel. Across our experiments, we tested three different values of $\theta$: $10^{-4}$, $10^{-6}$, and $10^{-8}$, and reported the one with the best errors. We formed two insights from this experiment. First, for problems where the data set was generated using SU2, the operator kernel benefited greatly from a larger $\theta$ as opposed to problems where the data set was generated from analytical solutions. This can be seen in Figures~\ref{fig:species_transport_nugget_exp} and~\ref{fig:taylor_green_nugget_exp}, which show the relative $\ell_2$ errors at the coarsest and finest nodesets using all three magnitudes of $\theta$ for the 3D species transport and 2D Taylor--Green vortices problems, respectively. The effect of the magnitude of $\theta$ was more pronounced for the results on the coarser nodesets, shown in Figures~\ref{fig:species_transport_ridge_500} and~\ref{fig:taylor_green_ridge_500}, than on the finer nodesets, shown in Figures~\ref{fig:species_transport_ridge_7000} and~\ref{fig:taylor_green_ridge_7288}. This is likely due to the numerical truncation errors inherent in the solution fields, manifesting as both numerical dissipation and dispersion. These errors likely propagate to the spatial interpolation coefficients which are then interpolated by the operator kernel. 
In contrast, the Taylor--Green problems do not use SU2 as the solutions are known analytically. Consequently, adding a regularization parameter hurts the performance there. For most problems using SU2, $\theta=10^{-4}$ or $\theta=10^{-6}$ yielded the best errors. We also experimented with adding a similar $\theta$ to the property-preserving kernel to address the noise in the velocity fields; however, it exaggerated the regularization of the spatial interpolation coefficients which resulted in very poor relative $\ell_2$ errors on generalization.
\subsection{\textcolor{edit}{Second-order velocity structure in turbulence problems}}
\label{sec:velocity_turb_analysis}
\begin{figure}[pos=!htpb]
    \centering
    \begin{tabular}{c}
    \begin{subfigure}{0.8\linewidth}
        \centering
        \includegraphics[width=0.85\linewidth]{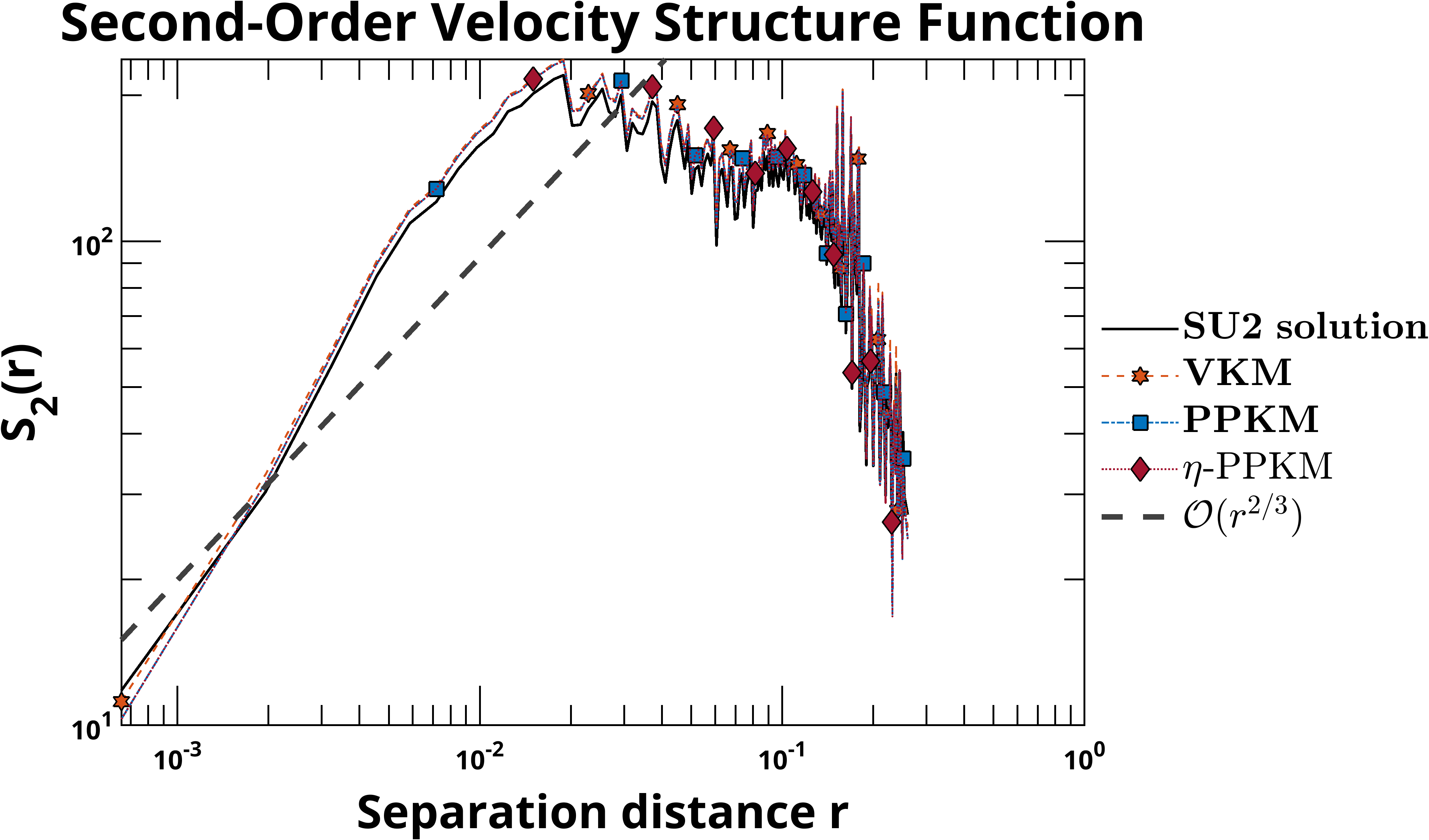}
        \caption{}
    \end{subfigure} \\
    \begin{subfigure}{0.9\linewidth}
        \centering
        \includegraphics[width=0.85\linewidth]{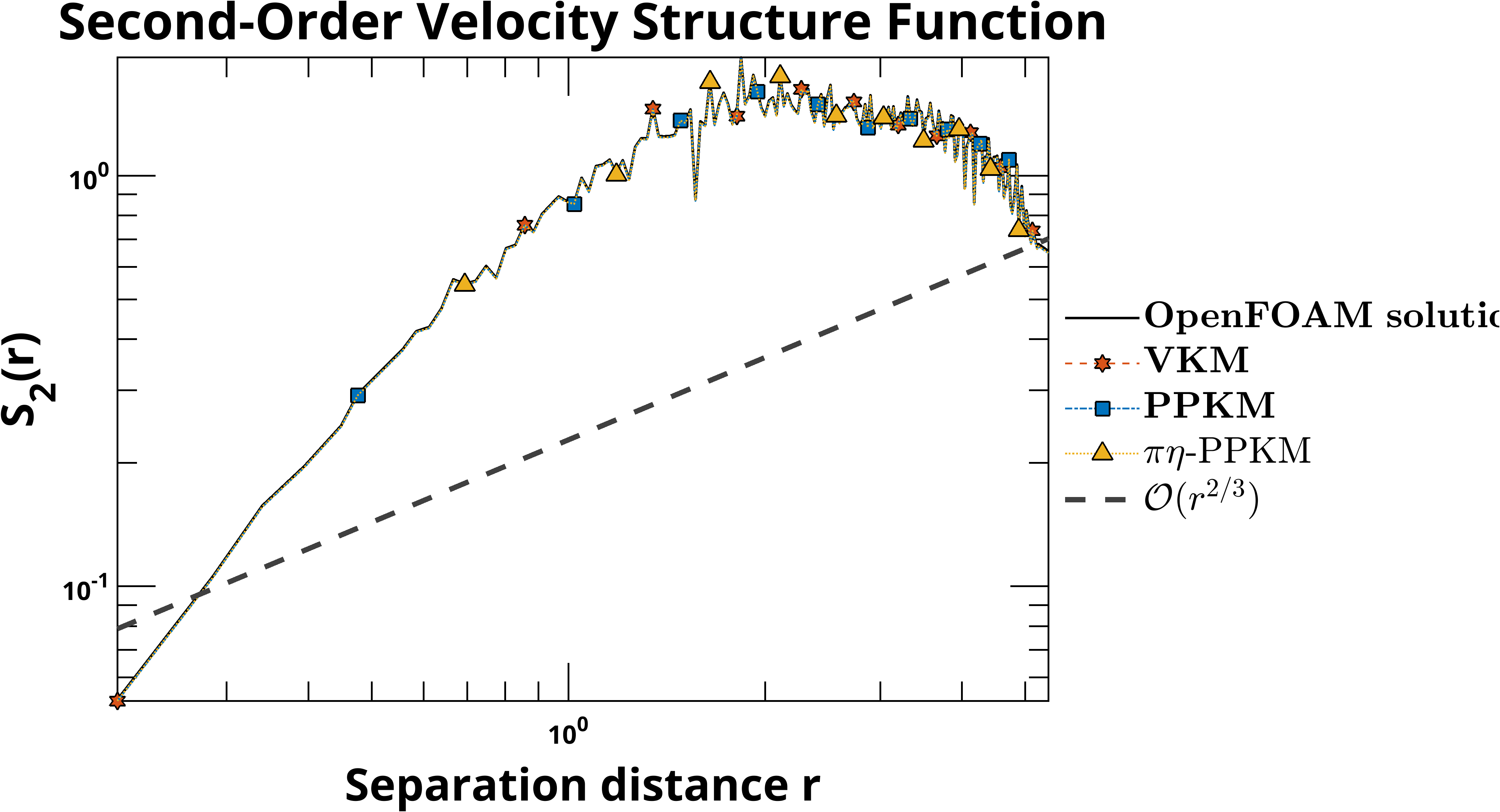}
        \caption{}
    \end{subfigure}
    \end{tabular}
    \caption{\textcolor{edit}{Second-order velocity structure functions of a test function for the (\textbf{A}) 3D turbulent species transport and the (\textbf{B}) 3D forced homogeneous isotropic turbulence problems. The dashed $\mathcal{O}(r^{2/3})$ line shows the Kolmogorov inertial range scaling as a reference.}}
    \label{fig:velocity_turb_analysis}
\end{figure}
\textcolor{edit}{In turbulent flows, global relative $\ell_2$ errors measure the overall predictive accuracy but do not reveal whether the multiscale spatial structure of turbulent velocity fields is preserved. To compare the scale-dependent behavior of the reference and predicted velocity fields, we compute the radially averaged second-order velocity structure function~\citep{Pope2000, arenas2006existence}
\begin{align}
S_2(r) &= \left\langle\left\|\bu(y)-\bu(y_j)\right\|_2^2\right\rangle_{\|y-y_j\|_2=r},
\end{align}
where $y, y_j$ are the spatial points, $r$ is the standard Euclidean distance between them, and $\langle\cdot\rangle$ denotes an average over pairs of points separated by distance $r$ (within a radial bin). Under homogeneous and isotropic turbulence, the Kolmogorov energy cascade scaling $E(k)\sim k^{-5/3}$ corresponds to $S_2(r)\sim r^{2/3}$~\cite{arenas2006existence}; the dashed line in Figure~\ref{fig:velocity_turb_analysis} is therefore included only as a reference. Neither simulation is expected to exhibit a fully resolved inertial range: the 3D species transport solution is obtained from incompressible RANS, while the 3D forced homogeneous isotropic turbulence solution is obtained from LES and has limited scale separation at its spatial resolution. Accordingly, we do not interpret agreement with the $r^{2/3}$ line as evidence of a \emph{correctly} resolved energy cascade. Instead, the principal observation is that the kernel methods' predictions closely reproduce the reference structure functions of their respective reference solutions across the resolved separation distances. Thus, the kernel methods \emph{empirically} recover the second-order scale-dependent structure of the velocity fields.}
\subsection{\textcolor{edit}{Empirical convergence rates}}
\label{sec:error_analysis}
To assess how the proposed method’s accuracy improves with the number of training functions $N$, we examine its empirical convergence rates and compare them with the theoretical estimates in Section~\ref{sec:approx-thms}. First, we fit a straight line to $\log E(N)$ versus $\log N$ over the training set sizes used in the ablation studies, where $E(N)$ denotes the mean test relative $\ell_2$ error. The fitted slope $s$ gives the empirical power law $E(N)\approx C N^s$. We then express each observed slope in the intrinsic-dimensional form predicted by Theorem~\ref{th:dimension1}. For the $C^4$ Mat\'ern operator kernel used here, $\nu=5/2$; thus, when $s<-1/2$, we define the rate-compatible effective dimension
\begin{align}
r_{\mathrm{eff}}=\frac{5/2}{-s-1/2},\label{eq:r_eff}
\end{align}
so that the observed rate can be written as $E(N)\approx C N^{-(5/(2r_{\mathrm{eff}})+1/2)}$. The worst-case estimate in Theorem~\ref{th:dimension1} corresponds to $r=d_X$, where $d_X$ is the input function dimensionality. Hence, $r_{\mathrm{eff}}<d_X$ indicates faster empirical convergence than the worst-case estimate.
\begin{table}[pos=!htbp]
\color{edit}
\centering
\setlength{\tabcolsep}{2pt}
\renewcommand{\arraystretch}{0.3}
\resizebox{0.6\linewidth}{!}{%
\begin{tabular}{lrrr}
\toprule
Problem & $d_X$ & Fitted $s$ & $r_{\mathrm{eff}}$ \\
\midrule
2D flow past a cylinder (no shedding)     & 1    & $ 1.148$ & --    \\
2D flow past a cylinder (vortex shedding) & 1    & $-0.913$ & 6.05  \\
2D Taylor--Green vortices                           & 500  & $-2.608$ & 1.19  \\
2D spacetime Taylor--Green vortices                 & 500  & $-2.721$ & 1.13  \\
2D backward-facing step                    & 11   & $-0.281$ & --    \\
2D buoyancy-driven cavity flow                  & 3    & $-0.785$ & 8.78  \\
2D lid-driven cavity flow                       & 33   & $-1.401$ & 2.78  \\
2D merging vortices                        & 500  & $ 0.903$ & --    \\
3D species transport                       & 1388 & $-1.133$ & 3.95  \\
3D forced homogeneous isotropic turbulence & 7000 & $-0.724$ & 11.18 \\
\bottomrule
\end{tabular}}
\caption{Comparison of the PPKM's empirical convergence rates and the rate-compatible effective dimensions $r_{\mathrm{eff}}$. $s$ is obtained by fitting $E(N)\approx CN^s$, and $d_X$ is the number of nonredundant values in the discretized input functions. A negative $s$ indicates that the fitted error decreases with $N$; more negative values indicate faster decay. A dash indicates that $s\geq-1/2$, for which no positive $r_{\mathrm{eff}}$ is compatible with the formula in Theorem~\ref{th:dimension1}.}
\label{tab:empirical_convergence_rates}
\end{table}

The fitted $r_{\mathrm{eff}}$ in~\eqref{eq:r_eff} is the value of $r$ whose theoretical exponent matches the observed slope. In Table~\ref{tab:empirical_convergence_rates}, $d_X$ counts the nonredundant values in each discretized input function. We omit values that never change and count exact repetitions only once. For example, only the $33$ top-boundary values vary in the 2D lid-driven cavity flow problem, so $d_X=33$.

The spatial and spacetime 2D Taylor--Green vortices, 2D lid-driven cavity flow, 3D species transport, and 3D forced homogeneous isotropic turbulence problems have $r_{\mathrm{eff}}\ll d_X$. Their fitted errors decay substantially faster than the worst-case rate obtained by setting $r=d_X$ in Theorem~\ref{th:dimension1}. The other problems show slower or increasing error trends. On 2D flow past a cylinder (vortex shedding) and 2D buoyancy-driven cavity flow problems, the errors decrease with $N$, but $r_{\mathrm{eff}}>d_X$. The 2D backward-facing step problem also shows slow decay ($-1/2<s<0$), so~\eqref{eq:r_eff} does not give a positive $r_{\mathrm{eff}}$. On 2D flow past a cylinder (no vortex shedding) and 2D merging vortices, the fitted errors increase with $N$; consequently, no finite positive $r_{\mathrm{eff}}$ is found. However, these fits describe the observed trends over the tested $N$ and need not reflect the theorem's asymptotic rates. One possible explanation is that these problems sample a narrowly constrained flow regime, such that adding new training functions adds little new information and produces limited or no improvement in predictive accuracy.
\subsection{\textcolor{edit}{Out-of-distribution (OOD) tests}}
\label{sec:ood_tests}
\begin{table}[pos=!htbp]
\color{edit}
\centering
\setlength{\tabcolsep}{4pt}
\renewcommand{\arraystretch}{1.15}
\resizebox{\linewidth}{!}{%
\begin{tabular}{llll}
\toprule
Problem & Varied quantity & Training distribution & OOD distribution \\
\midrule
2D flow past a cylinder (no shedding)
& Initial velocity magnitude & $\mathcal{N}_{[0.1,0.26]}(0.18,0.026)$
& $\mathcal{N}_{[0.3,0.6]}(0.45,0.15)$ \\
\hhline{----}
2D flow past a cylinder (vortex shedding)
& Initial velocity magnitude & $\mathcal{N}_{[0.45,0.8]}(0.625,0.1)$
& $\mathcal{N}_{[0.9,1.2]}(1.1,0.1)$ \\
\hhline{----}
2D lid-driven cavity flow
& Top-wall velocity $u_1^{\mathrm{top}}$ & $\mathcal{N}_{[0.8,2.5]}(1.5,1)$
& $\mathcal{N}_{[2.6,5]}(3.5,1.5)$ \\
\hhline{----}
2D backward-facing step
& Kinematic viscosity $\nu$ & $\mathcal{N}_{[1,36]}(18,18)$
& $\mathcal{N}_{[38,70]}(54,30)$ \\
\hhline{----}
2D buoyancy-driven cavity flow
& Left-wall temperature $\texttt{T}^{\mathrm{left}}$ & $\mathcal{N}_{[300,600]}(450,100)$
& $\mathcal{N}_{[700,900]}(800,100)$ \\
\hhline{----}
\multirow{2}{*}{2D Taylor--Green vortices}
& Amplitude $A$ & $\mathcal{N}_{[0.1,80]}(40,30)$
& $\mathcal{N}_{[90,200]}(110,50)$ \\
& Kinematic viscosity $\nu$ & $\mathcal{N}_{[0.0001,1]}(0.006,0.1)$
& $\mathcal{N}_{[1,10]}(5,2)$ \\
\hhline{----}
2D merging vortices
& Exponent $\alpha$ & $\mathcal{U}([1,14])$
& $\mathcal{U}([15,30])$ \\
\hhline{----}
3D species transport
& Inlet velocity magnitude & $\mathcal{N}_{[1,40]}(20,20)$
& $\mathcal{N}_{[110,130]}(120,20)$ \\
\hhline{----}
3D forced homogeneous isotropic turbulence
& Forcing amplitudes $A,B,C$ & $\mathcal{N}_{[-0.3,0.3]}(0,0.15)$
& $\mathcal{N}_{[0.4,0.8]}(0.5,0.15)$ \\
\bottomrule
\end{tabular}}
\caption{Training and OOD distributions used in the OOD tests. The truncated normal distribution $\mathcal{N}_{[a,b]}(\mu,\sigma)$ is defined in Section~\ref{sec:results} and $\mathcal{U}([a,b])$ denotes the uniform distribution on $[a,b]$.}
\label{tab:ood_distributions}
\end{table}
\begin{table}[pos=!htbp]
\color{edit}
\centering
\setlength{\tabcolsep}{4pt}
\renewcommand{\arraystretch}{1.15}
\resizebox{0.7\linewidth}{!}{%
\begin{tabular}{lrrr}
\toprule
& & \multicolumn{2}{c}{Maximum posterior variance} \\
Problem & $N$ & ID & OOD \\
\midrule
2D flow past a cylinder (no shedding) & 100 & $2.01\times10^{-3}$ & $3.00$ \\
2D flow past a cylinder (vortex shedding) & 10,000 & $1.08\times10^{-4}$ & $3.00$ \\
2D Taylor--Green vortices & 5,000 & $3.58\times10^{-6}$ & $3.00$ \\
2D spacetime Taylor--Green vortices & 5,000 & $3.58\times10^{-6}$ & $3.00$ \\
2D backward-facing step & 500 & $3.02\times10^{-5}$ & $3.00$ \\
2D merging vortices & 500 & $9.06\times10^{-2}$ & $3.00$ \\
2D lid-driven cavity flow & 10,000 & $2.95\times10^{-4}$ & $3.00$ \\
2D buoyancy-driven cavity flow & 10,000 & $2.44\times10^{-4}$ & $3.00$ \\
3D species transport & 10,000 & $1.77\times10^{-4}$ & $3.00$ \\
3D forced homogeneous isotropic turbulence & 10,000 & $2.03\times10^{-4}$ & $0.126$ \\
\bottomrule
\end{tabular}}
\caption{Maximum operator-kernel GP posterior variance for the PPKM over $200$ in-distribution (ID) and $500$ out-of-distribution (OOD) test functions.}
\label{tab:ood_gp_uncertainty}
\end{table}
\textcolor{edit}{Finally, we evaluated the OOD generalization of all methods across our experiments by training on the original in-distribution data and evaluating on $500$ unseen inputs sampled outside the corresponding training distributions. The OOD inputs were evaluated using the same procedure as the in-distribution test inputs. Table~\ref{tab:ood_distributions} summarizes the varied quantities and their training and OOD distributions; all other problem configurations were held fixed. Table~\ref{tab:ood_gp_uncertainty} reports the maximum operator-kernel GP posterior uncertainty (see Section~\ref{sec:uq}) for the PPKM on the in-distribution (ID) and OOD test sets.}

\textcolor{edit}{The Geo-FNO achieved OOD relative $\ell_2$ errors around $0.6$ on 2D flow past a cylinder (both flow regimes), 2D lid-driven cavity flow, and 3D species transport, and around $1$ on the backward-facing step. The Transolver's errors were well below $1$ on the 2D backward-facing step, 2D flow past a cylinder (both flow regimes), 2D lid-driven cavity flow, and 2D merging vortices. On the 3D homogeneous isotropic forced turbulence problem, both achieved errors around $0.01$. On all other problems, the OOD errors of both neural operators were around $1$ or reached $10^2$--$10^3$.}

\textcolor{edit}{Across all problems, the VKM and the PPKM variants produced relative $\ell_2$ errors around $1$, indicating that the present kernel methods do not extrapolate beyond the support of their training distributions. The neural operators had mixed performance. This occurs because the kernel methods are interpolatory and thus are more sensitive to the range of the training data upon generalization~\citep{wilson2013gaussian, canatar2021out}, while neural operators are essentially trained with least-squares techniques, have more trainable parameters, and therefore slightly better OOD behavior sometimes. We also emphasize that the results shown in Table~\ref{tab:ood_gp_uncertainty} clearly demonstrate a key advantage of the PPKM: \emph{a priori} detection of the source of a test example. The posterior variance of the PPKM indicates whether an unseen example is in the training distribution. The ID test functions for each problem result in posterior variances several orders of magnitude smaller than the ones resulting from the OOD test functions. In contrast, neural operators can only know of an unseen example's distribution according to the pointwise errors of the predicted velocity fields.}

\textcolor{edit}{This advantage of the PPKM can be used for active learning extensions by helping decide whether including a new unseen example in the training set will improve the model's generalizability. We strongly suspect that building a parameter-aware kernel (w.r.t. temperature, viscosity, amplitude, etc.) will change this easily. It is also likely that making the kernels have more trainable parameters will help close the gap.}

\section{Discussion}
\label{sec:discussion}
The underlying motivation for this work is the observation that a trustworthy surrogate must actually satisfy incompressibility and other properties analytically, much like the PDE solvers that these surrogates aim to replace; while pointwise generalization errors are important, lower errors alone are insufficient.

Our property-preserving operator learning framework achieves this through a key idea: learning maps from input function samples to expansion coefficients of output functions in a property-preserving basis. The core of the method involves (1) a property-preserving kernel interpolant for each training output function that analytically preserves incompressibility, periodicity (when applicable), and turbulence (when applicable); and (2) an operator kernel interpolant that maps from input function samples to the interpolation/expansion coefficients obtained from (1). \textcolor{edit}{We chose kernels because they (1) allow for the \emph{simultaneous} and \emph{analytical} enforcement of multiple properties without any feature engineering, (2) allow for rapid training at scale, and (3) have an RKHS structure that provides rigorous \emph{a priori} convergence analysis}. Our method also allows for natural uncertainty quantification via Gaussian processes (GPs), though this is not the primary focus of this work. Our method also admits pessimistic and optimistic worst-case \emph{a priori} error estimates, as discussed in Section \ref{sec:approx-thms}.

Remarkably, our results from Section \ref{sec:results} show that the PPKM almost always outperforms state-of-the-art neural operators, \textcolor{edit}{their divergence-loss variants, and the vanilla kernel method} even on the metric of pointwise generalization error across a suite of benchmark problems (see Table \ref{tab:exp_config} for experimental configuration details). Our method was anywhere from one to six orders of magnitude more accurate on generalization. In addition, our method produces \emph{analytically} incompressible flow fields, in contrast with neural operators; Table \ref{tab:best_results} shows the latter typically exhibited absolute divergences of $O(1) - O(10^4)$ across this benchmark suite. Surprisingly, our method achieves these superior results with \textbf{only two trainable parameters}---one for the property-preserving interpolant, and one for the operator interpolant; in contrast, neural operators typically required $O(10^5)$ to $O(10^6)$ trainable parameters.

The benefits of our method extend beyond superior accuracy. The PPKM consistently trains up to \textcolor{edit}{three orders of magnitude faster than the neural operators}, despite the use of double precision floating point (fp64) for our method and single precision floating point (fp32) for neural operators; note that it is in general inadvisable to use fp64 with neural operators due to the difficulties in storing trainable parameters.
This computational advantage extends to inference as well, where the PPKM is up to $36 \times$ (approximately $1.5$ orders of magnitude) faster than the neural operators; see Table~\ref{tab:best_results}.

Neural operators are able to tackle very large numbers of training functions ($N$) primarily through batching techniques, but standard kernel methods do not directly allow for batching in $N$. \textcolor{edit}{We address this limitation with a simple but powerful streaming technique that allowed us to apply the kernel method on the GPU efficiently even for $N=10,000$ training functions sampled at up to $10,000$ spatial points. Further, a central computational contribution of this work is a recursive Schur-complement technique tailored to the block linear systems induced by the property-preserving kernels. Together with the careful interpolation node selection, this approach allows us to reduce the preprocessing times to $O(d m^3)$ (over a naive $O(d^3 m^3)$), storage costs to $O(d m^2)$ (over a naive $O(d^2 m^2)$), and inference times to $O(N^2 + Nd m^2)$ (over a naive $O(N^2 + Nd^2 m^2)$).} These \textcolor{edit}{computational techniques} allowed us to actually tackle large, real-world problems in 3D.
The out-of-distribution (OOD) tests show that the operator-kernel GP provides a reliable uncertainty signal for unseen input functions, which could potentially guide the selection of new training examples in an active learning formulation.
Another contribution of this work, though not major, is the meticulous and careful documentation of domain geometries, solver details from the SU2 solver, initial and boundary conditions, and flow regimes for all of our problems in Table \ref{tab:exp_config}.

We conclude with a discussion on the limitations of our method and avenues for future work. The primary limitation of kernel methods---which we partially worked around by streaming and efficient linear algebra---is the need to store and compute with large dense interpolation and evaluation matrices. Further, kernel methods typically perform best when at the edge of ill-conditioning (at least when using kernel translates as the basis), which necessitates the use of fp64 to ensure high accuracy; however, this in turn leads to slower GPU computation. Further speedups may be made possible by the use of multipole expansions~\citep{rokhlin1985rapid}, treecodes~\citep{yokota2011treecode}, compactly-supported kernels~\citep{Wendland2004}, and low-rank approximation~\citep{drineas2005approximating, gittens2013revisiting}. Finally, our approach leverages \emph{global} property-preserving kernel interpolation for the output functions, which inflates memory requirements; local interpolation techniques offer a way around this while potentially allowing for greater accuracy. We will also explore generalizations to compressible flow, electromagnetism, and magnetohydrodynamics (MHD), each of which come with their own unique properties and constraints.

\section{Acknowledgements}
\label{sec:ack}
RS and VS were supported by the National Science Foundation under DMS 2505986, Cooperative Agreement 2421782, and the Simons Foundation grant MPS-AI-00010515 awarded to the NSF-Simons AI Institute for Cosmic Origins — CosmicAI, \url{https://www.cosmicai.org}. VS was additionally supported by the Air Force Office of Scientific Research (AFOSR) under LRIR award FA9550-25-1-0042. HO acknowledges support from the Air Force Office of Scientific Research under MURI award number FOA-AFRL-AFOSR-2023-0004 (Mathematics of Digital Twins), the Department of Energy under award number DE-SC0023163 (SEA-CROGS: Scalable, Efficient, and Accelerated Causal Reasoning Operators, Graphs and Spikes for Earth and Embedded Systems) and the DoD Vannevar Bush Faculty Fellowship Program under ONR award number N00014-18-1-2363. This work used the Delta system at the National Center for Supercomputing Applications [award OAC 2005572] through allocation [CIS260011] from the Advanced Cyberinfrastructure Coordination Ecosystem: Services \& Support (ACCESS) program, which is supported by National Science Foundation grants \#2138259, \#2138286, \#2138307, \#2137603, and \#2138296.

\newpage
\appendix
\counterwithin{figure}{section}
\counterwithin{table}{section}
\numberwithin{equation}{section}
\section{Additional results}
\label{app:additional_results}
We present additional experimental results in this section on four classic incompressible fluid flow problems (including a temperature-driven flow) and a challenging one involving merging vortices.
\subsection{2D Taylor--Green vortices}
\label{sec:taylor_green}
\begin{figure}[pos=!htpb]
    \centering
    \begin{tabular}{ccc}
    \raisebox{8ex}{\rotatebox{90}{Initial velocity to final velocity}} &
    \begin{subfigure}{0.48\linewidth}
        \centering
        \includegraphics[width=\linewidth]{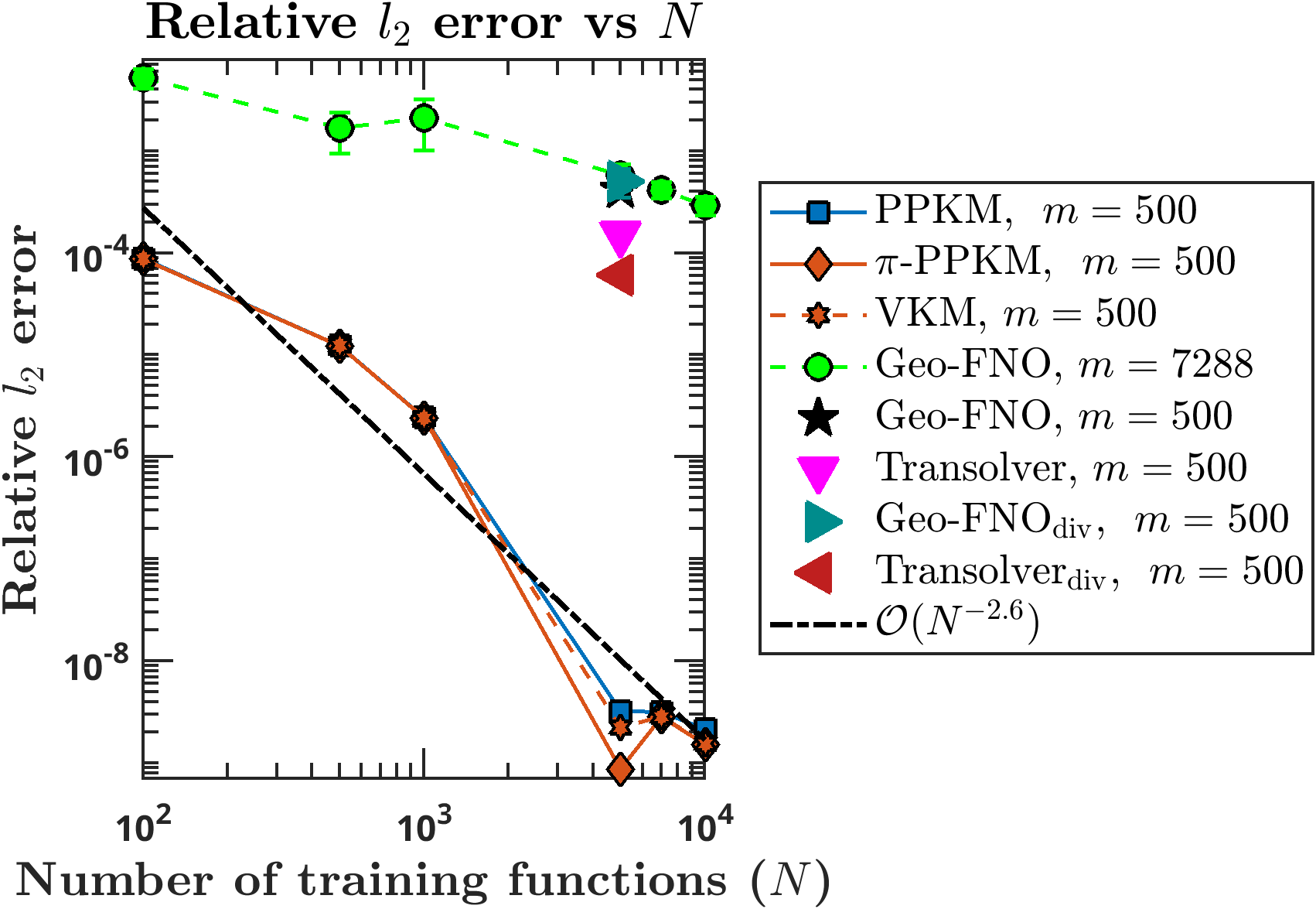}
        \caption{}
        \label{fig:taylor_green_errors}
    \end{subfigure} &
    \begin{subfigure}{0.48\linewidth}
        \centering
        \includegraphics[width=\linewidth]{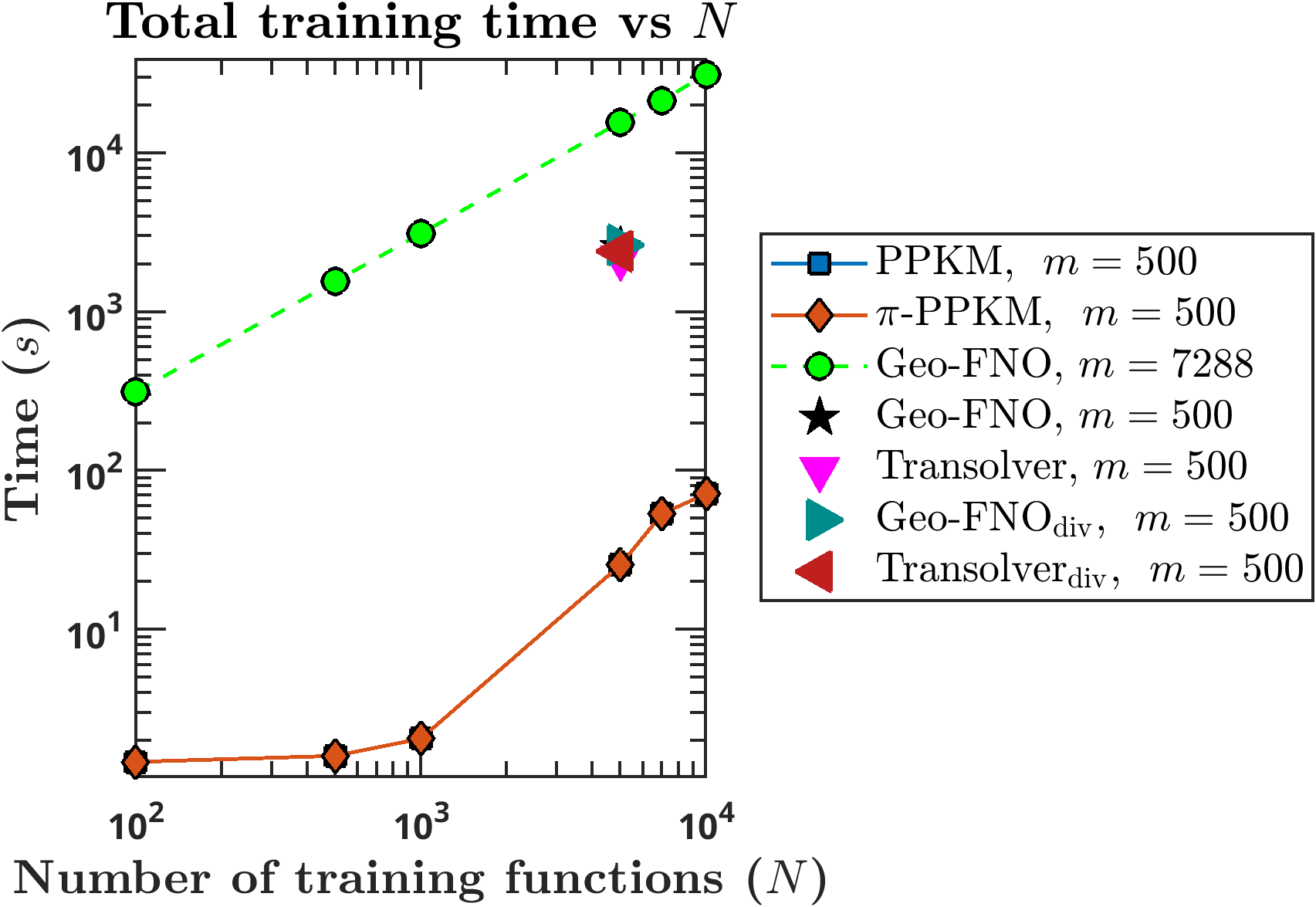}
        \caption{}
        \label{fig:taylor_green_times}
    \end{subfigure} \\
    \raisebox{8ex}{\rotatebox{90}{Flow parameters to final velocity}} &
    \begin{subfigure}{0.48\linewidth}
        \centering
        \includegraphics[width=\linewidth]{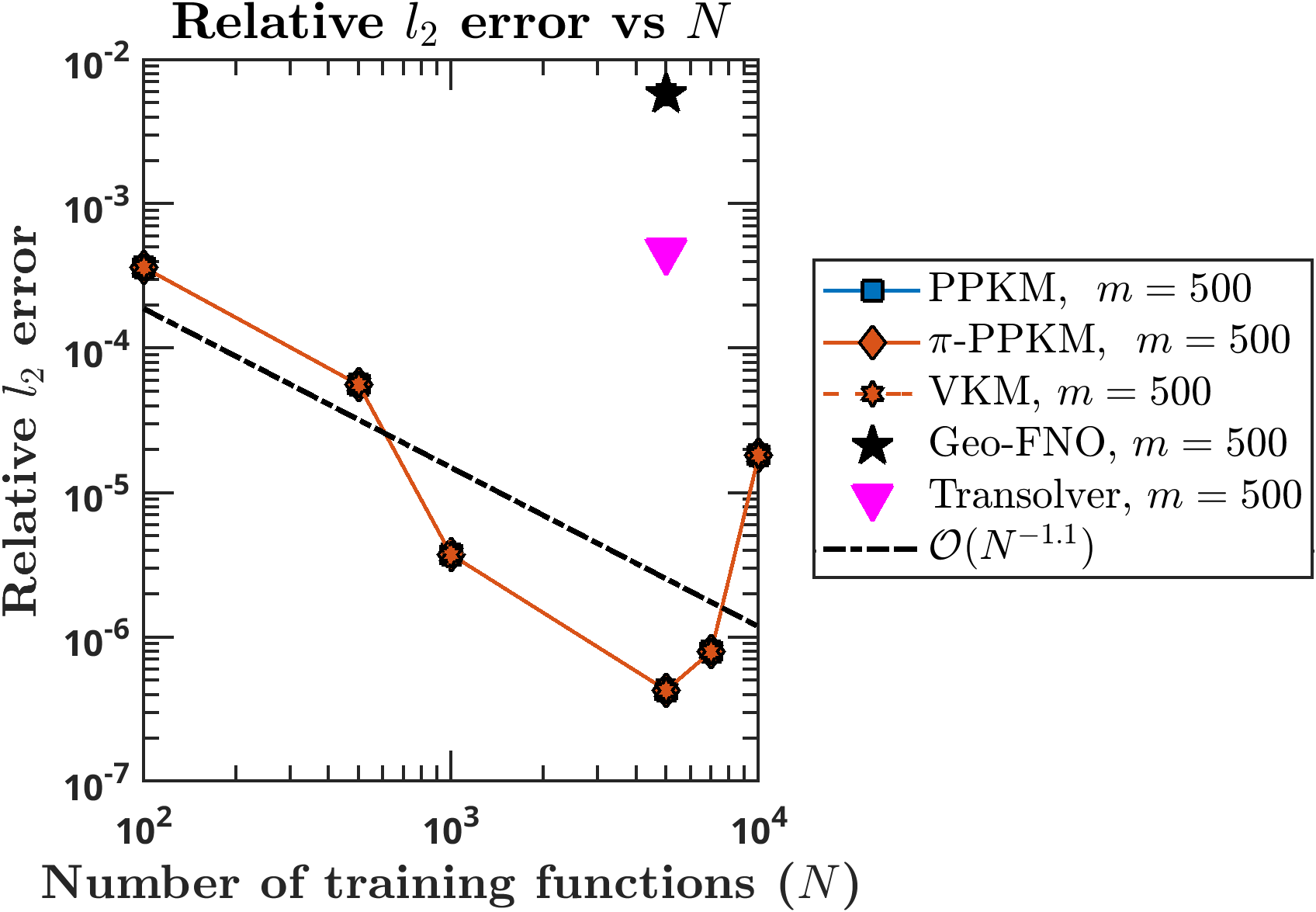}
        \caption{}
        \label{fig:taylor_green_coeffs_errors}
    \end{subfigure} &
    \begin{subfigure}{0.48\linewidth}
        \centering
        \includegraphics[width=\linewidth]{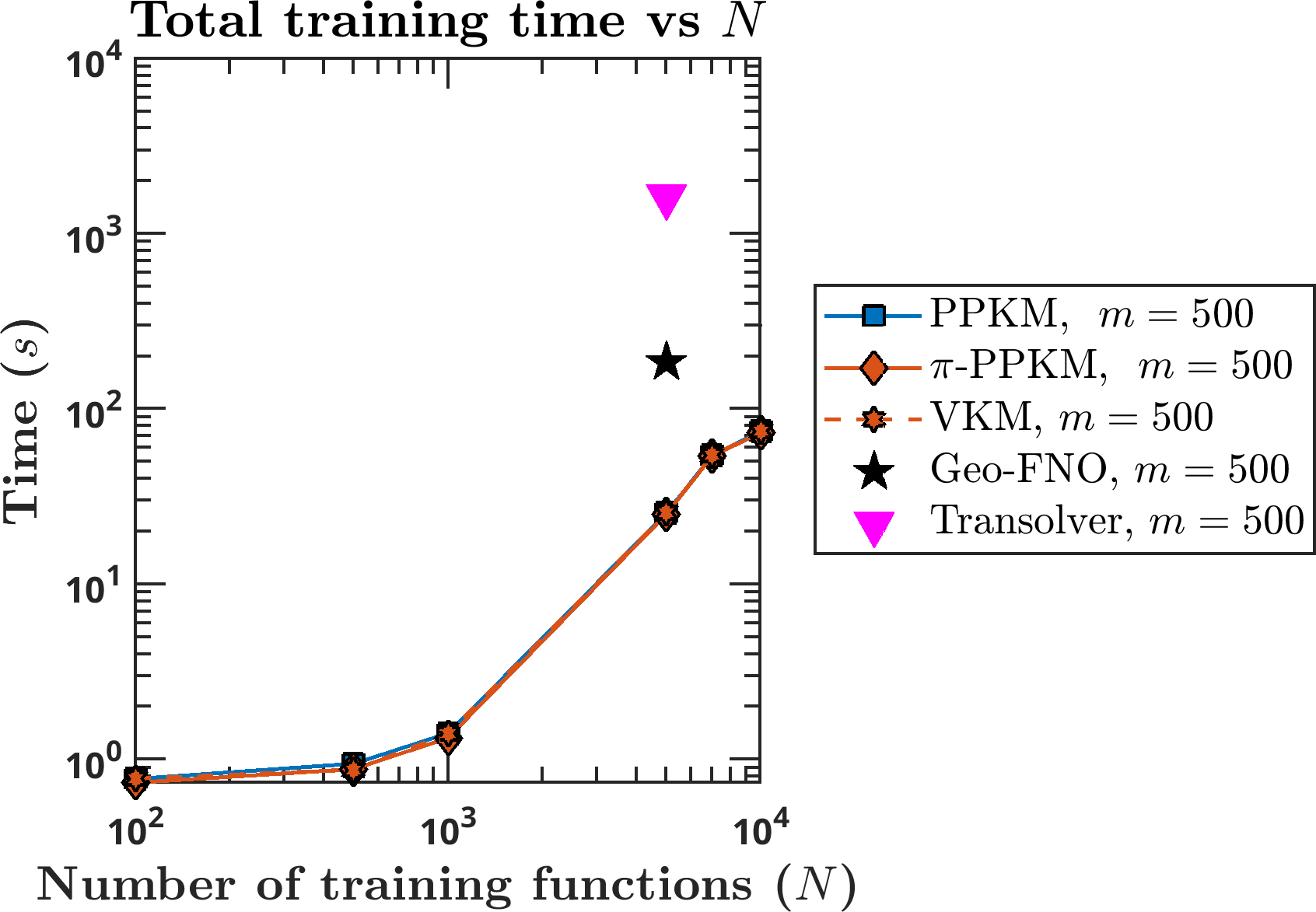}
        \caption{}
        \label{fig:taylor_green_coeffs_times}
    \end{subfigure}
    \end{tabular}
    \caption{The 2D laminar Taylor--Green vortices problem for the purely \textbf{spatial} operator map. The output functions are snapshots of the velocity at time $T=1$. (\textbf{A}) and (\textbf{B}) show the test relative $\ell_2$ errors and training runtimes as functions of $N$ for the operator map from the initial velocity to the final velocity. (\textbf{C}) and (\textbf{D}) show the same results for the operator map from the flow parameters to the final velocity.}
    \label{fig:taylor_green}
\end{figure}
We considered a purely spatial variant of the classical Taylor-Green vortex problem described in Section~\ref{sec:taylor_green_spacetime}, which we remind the reader has an exact solution. We are interested in two different operator maps being learned, (i) the spatial map $\G: \bu(x, 0) \rightarrow \bu(y, 1)$, and (ii) the parametric map $\G: \tau \rightarrow \bu(y, 1)$ where $\Omega_a = \tau \subset \mathbb{R}^2$ is the parameter space of $A$ and $\nu$. All other flow details were identical to the setup in Section~\ref{sec:taylor_green_spacetime} and are provided in row~5 in Table~\ref{tab:exp_config}. We show the results for both operator maps in Figure~\ref{fig:taylor_green}. We saw similar trends in this problem as the ones in the spacetime version of this problem in Section~\ref{sec:taylor_green_spacetime}. However, here we saw the $\pi$-PPKM gain a mild accuracy advantage over the PPKM method for $N=5000$, demonstrating the benefit of incorporating spatial periodicity in this periodic problem. Additionally, our method significantly outperforms the neural operator baselines in terms of both accuracy and training times. \textcolor{edit}{The Transolver divergence-loss variant improved the accuracy and numerical divergence, though it remained far less accurate than the PPKM and exhibited nonzero divergence.}
\subsection{2D lid-driven cavity flow}
\begin{figure}[pos=!htpb]
    \centering
    \begin{tabular}{cc}
    \begin{subfigure}{0.39\linewidth}
        \centering
        \includegraphics[width=\linewidth]{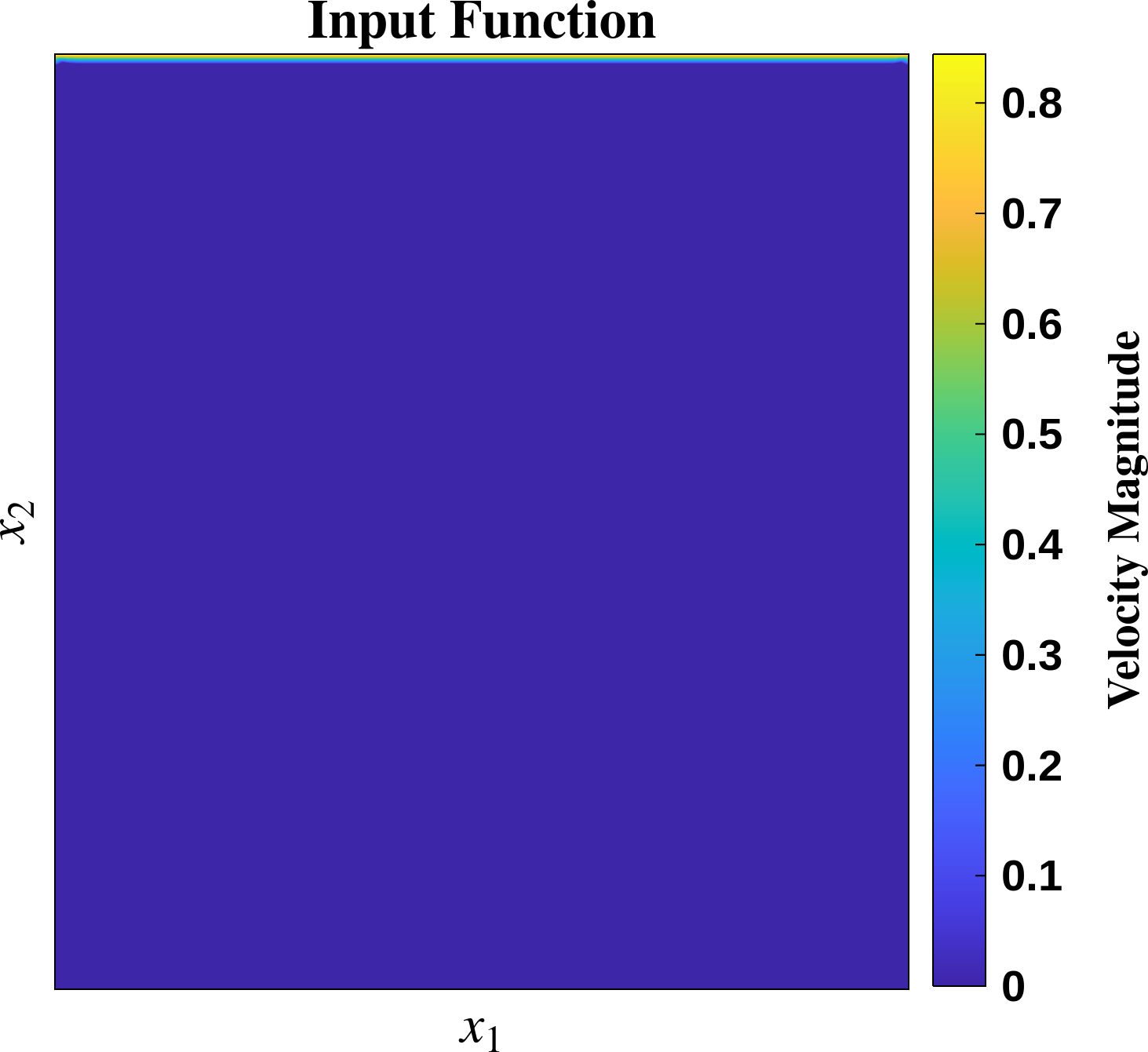}
        \caption{}
        \label{fig:lid_cavity_flow_input}
    \end{subfigure} &
    \begin{subfigure}{0.39\linewidth}
        \centering
        \includegraphics[width=\linewidth]{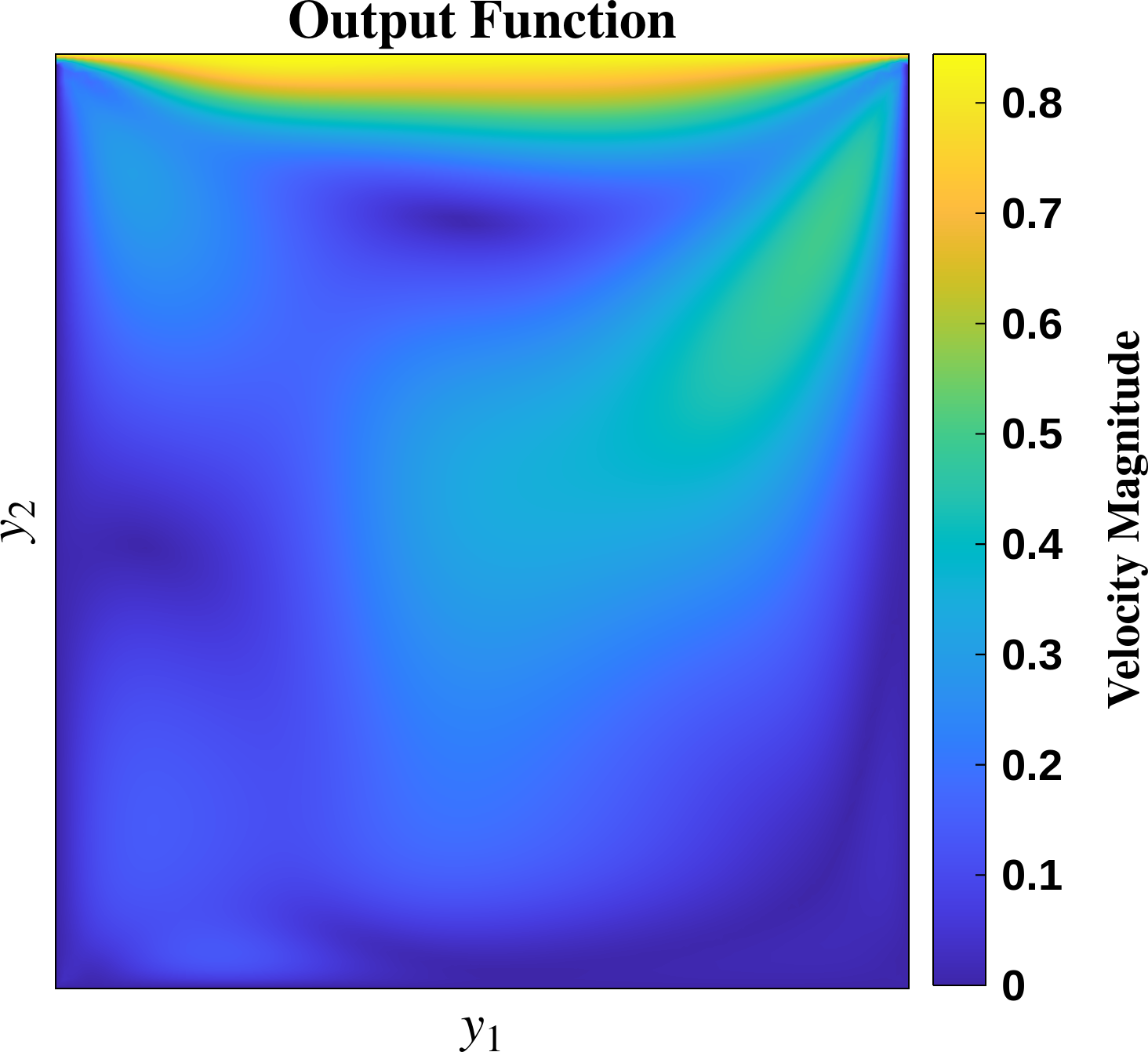}
        \caption{}
        \label{fig:lid_cavity_flow_output}
    \end{subfigure} \\
    \begin{subfigure}{0.5\linewidth}
        \centering
        \includegraphics[width=\linewidth]{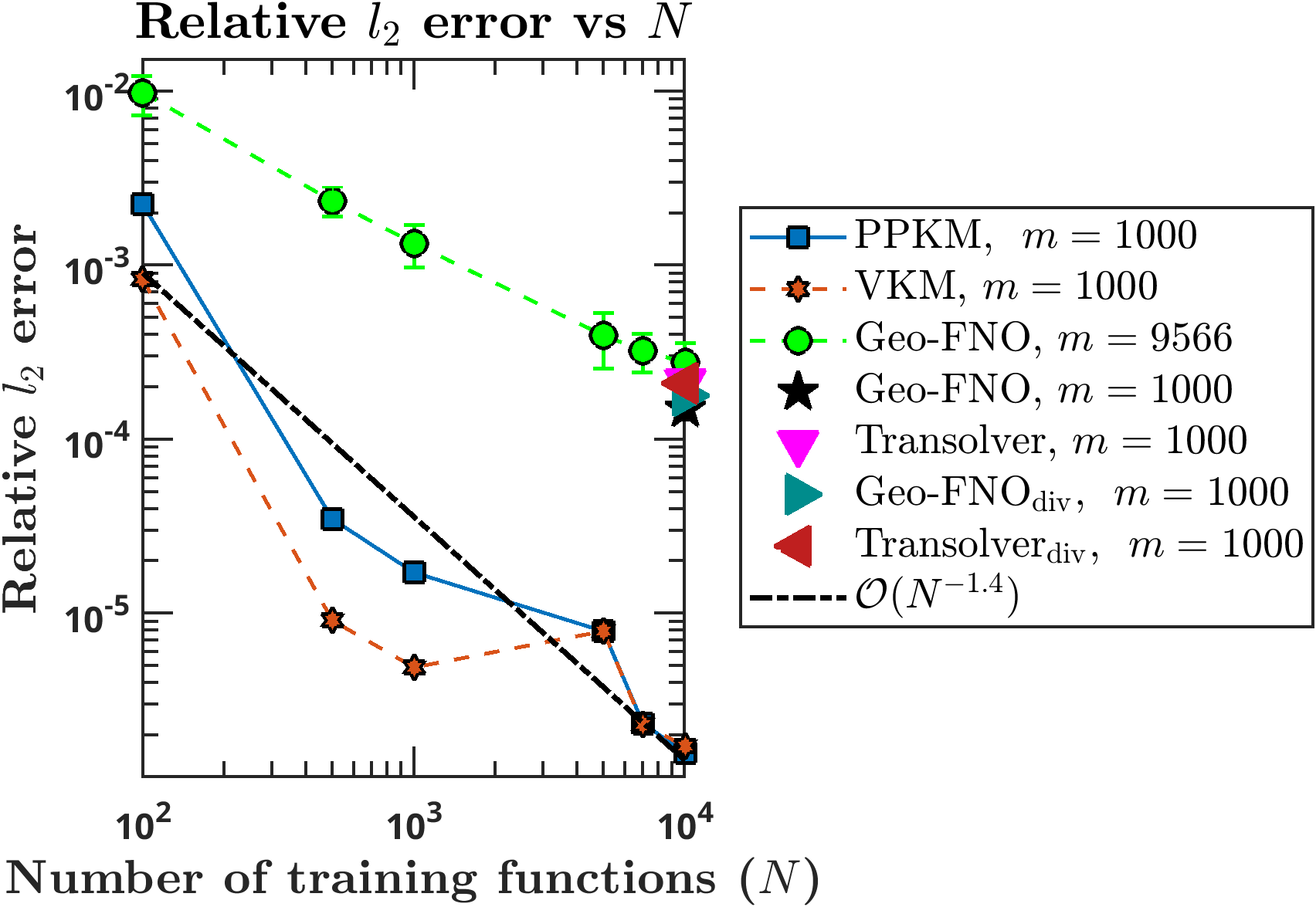}
        \caption{}
        \label{fig:lid_cavity_flow_errors}
    \end{subfigure} &
    \begin{subfigure}{0.5\linewidth}
        \centering
        \includegraphics[width=\linewidth]{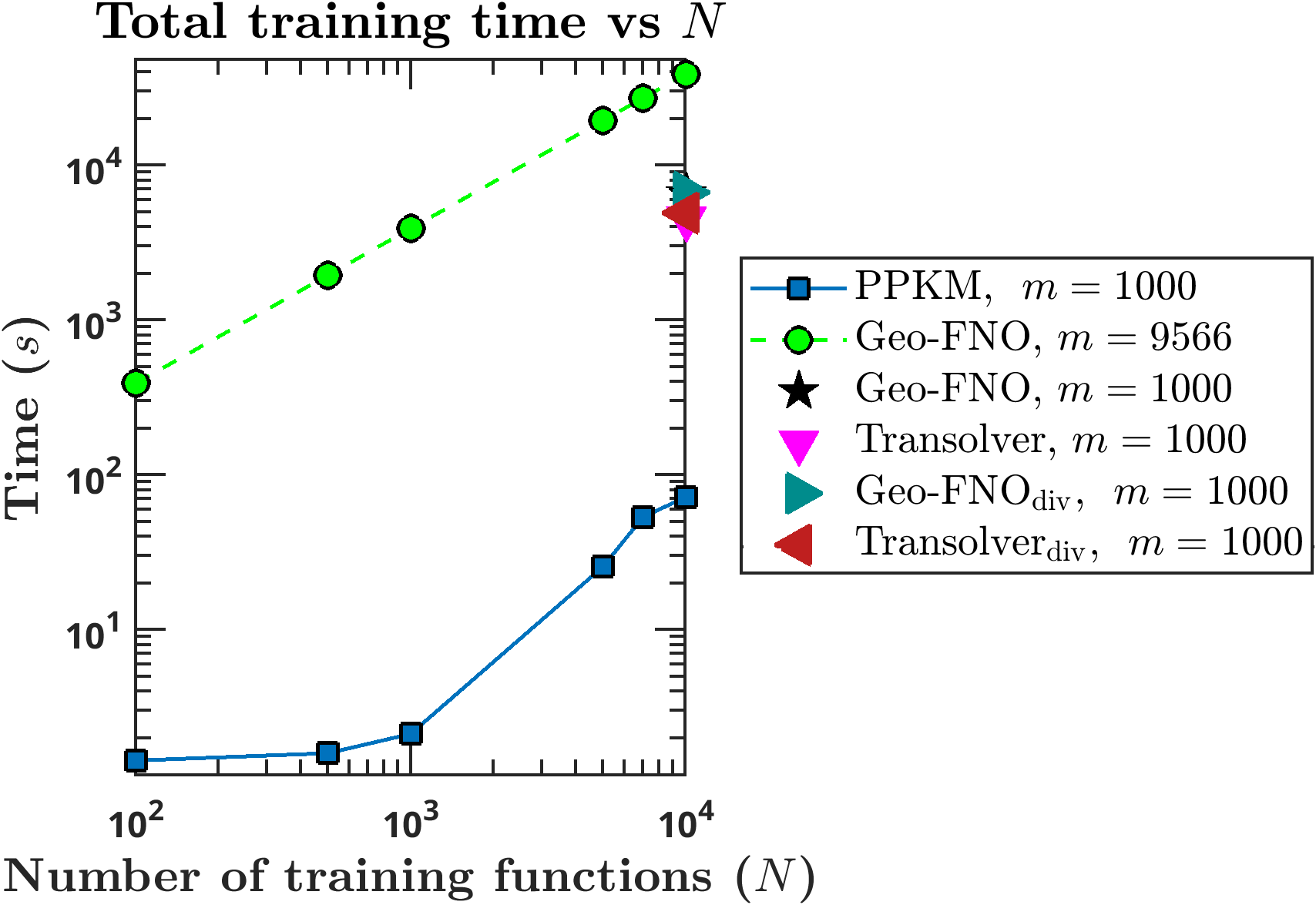}
        \caption{}
        \label{fig:lid_cavity_flow_times}
    \end{subfigure}
    \end{tabular}
    \caption{The 2D laminar lid-driven cavity flow problem. (\textbf{A}) and (\textbf{B}) show examples of an input function (the initial velocity) and an output function (the final velocity), respectively. (\textbf{C}) and (\textbf{D}) show the test relative $\ell_2$ errors and training runtimes as functions of $N$.}
    \label{fig:lid_cavity_flow}
\end{figure}
The classic 2D lid-driven cavity flow problem~\citep{ramanan1994linear, peng2003transition, sahin2003novel, kuhlmann2019} describes laminar fluid flow in a square cavity $\Omega_a = \Omega_v =[0, 1]^2$ whose lid moves tangentially along the top boundary. We initialized the horizontal velocity $u_1^{\mathrm{top}}$ to be constant along the top boundary (hence an impulsive start), sampled from the random distribution $\mathcal{N}_{[0.8, 2.5]}(1.5, 1)$, and zero everywhere else. $u_2$ was initialized to zero everywhere. We used the moving wall option provided in SU2 to prescribe the BC on the top boundary. The other three boundaries had no-slip BCs. We used a triangular mesh with $9566$ points. The simulation ran until $T=5$ with $\Delta t=10^{-3}$ and the operator map of interest was $\G: \bu(x, 0) \rightarrow \bu(y, 5)$. We also provide these details in row~2 in Table~\ref{tab:exp_config}. We show example input and output functions in Figures~\ref{fig:lid_cavity_flow_input} and~\ref{fig:lid_cavity_flow_output} respectively.

The VKM method outperformed the PPKM method for small $N$ as shown in Figure~\ref{fig:lid_cavity_flow_errors}, but both methods achieved the same errors as $N$ increased. Our method achieved around 2.5 orders of magnitude lower errors and around 1.4 orders of magnitude faster training times (Figure~\ref{fig:lid_cavity_flow_times}) than the neural operators.
\subsection{2D backward-facing step}
\begin{figure}[pos=!htpb]
    \centering
    \begin{subfigure}{0.78\linewidth}
        \centering
        \includegraphics[width=\linewidth]{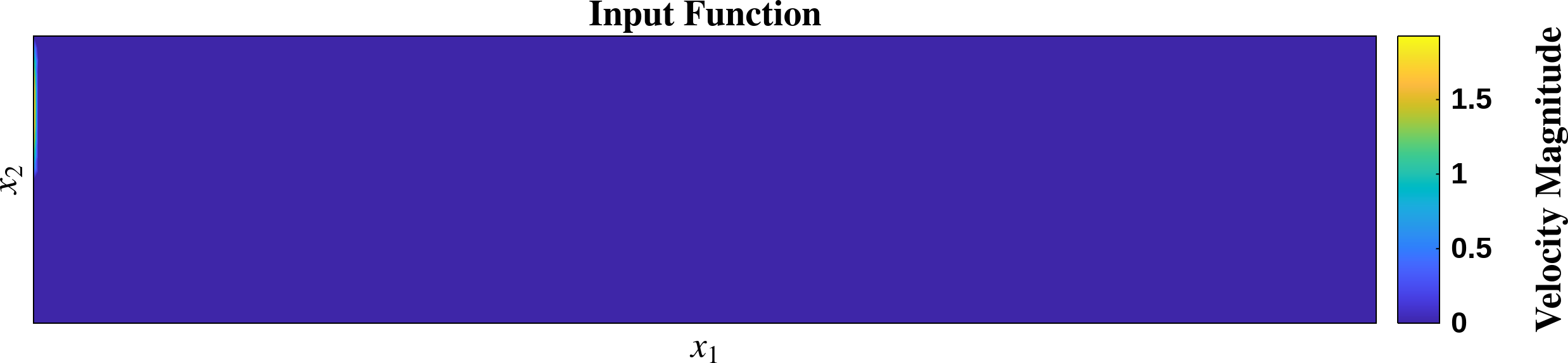}
        \caption{}
        \label{fig:backward_facing_step_input}
    \end{subfigure} \\
    \begin{subfigure}{0.78\linewidth}
        \centering
        \includegraphics[width=\linewidth]{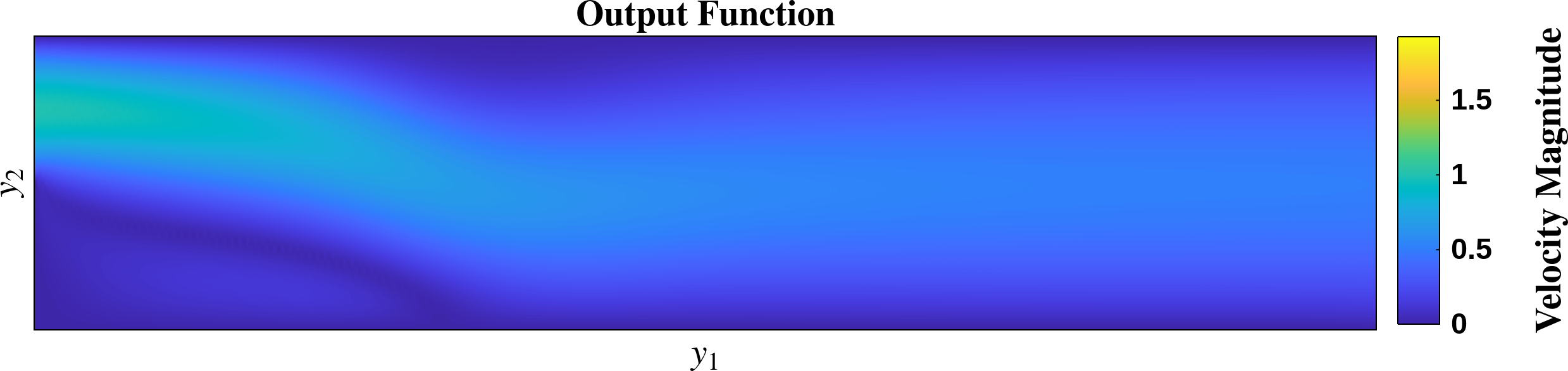}
        \caption{}
        \label{fig:backward_facing_step_output}
    \end{subfigure} \\
    \begin{subfigure}{0.36\linewidth}
        \centering
        \includegraphics[width=\linewidth]{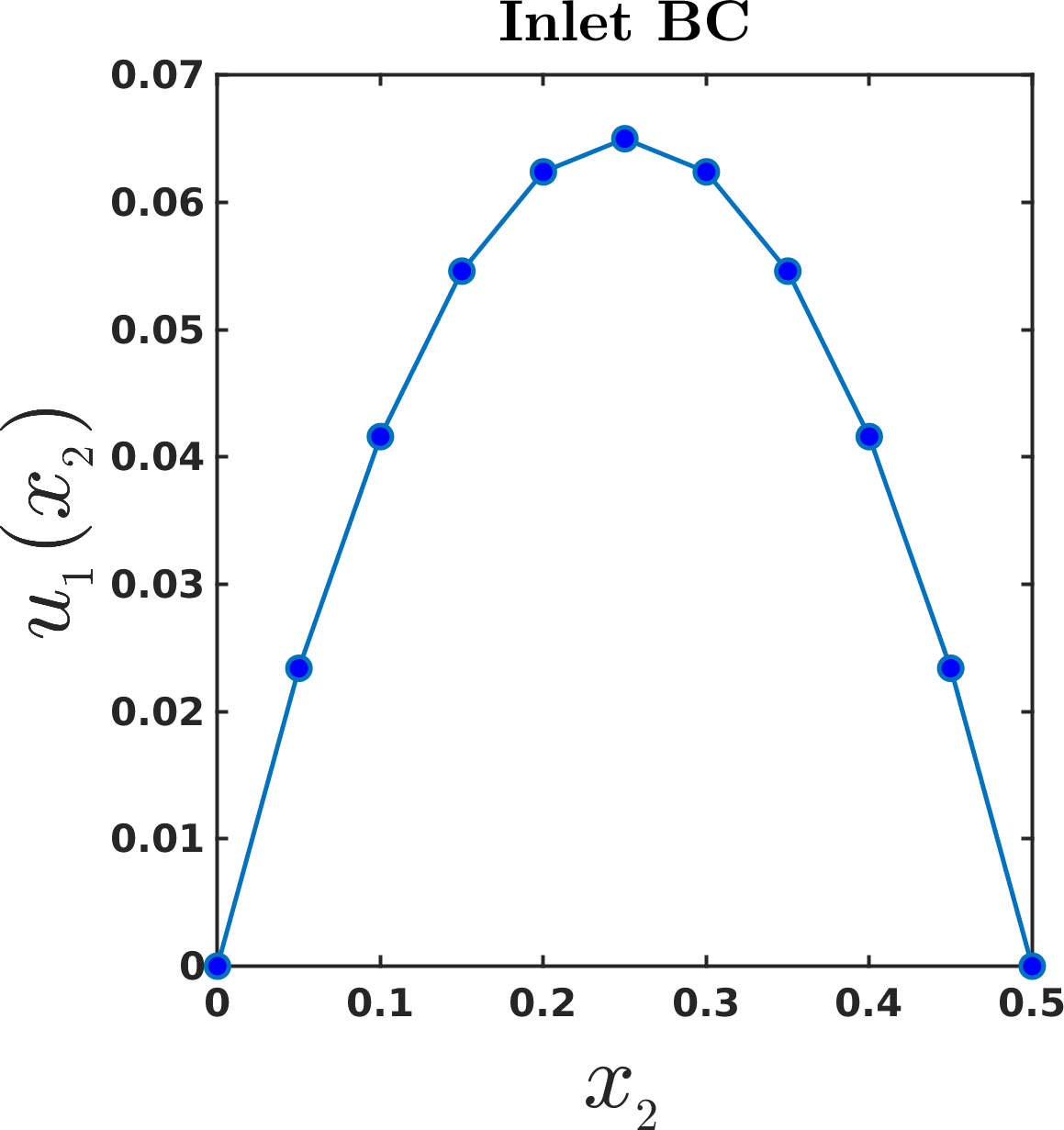}
        \caption{}
        \label{fig:backward_facing_step_input_bc}
    \end{subfigure}
    \begin{tabular}{cc}
    \begin{subfigure}{0.5\linewidth}
        \centering
        \includegraphics[width=\linewidth]{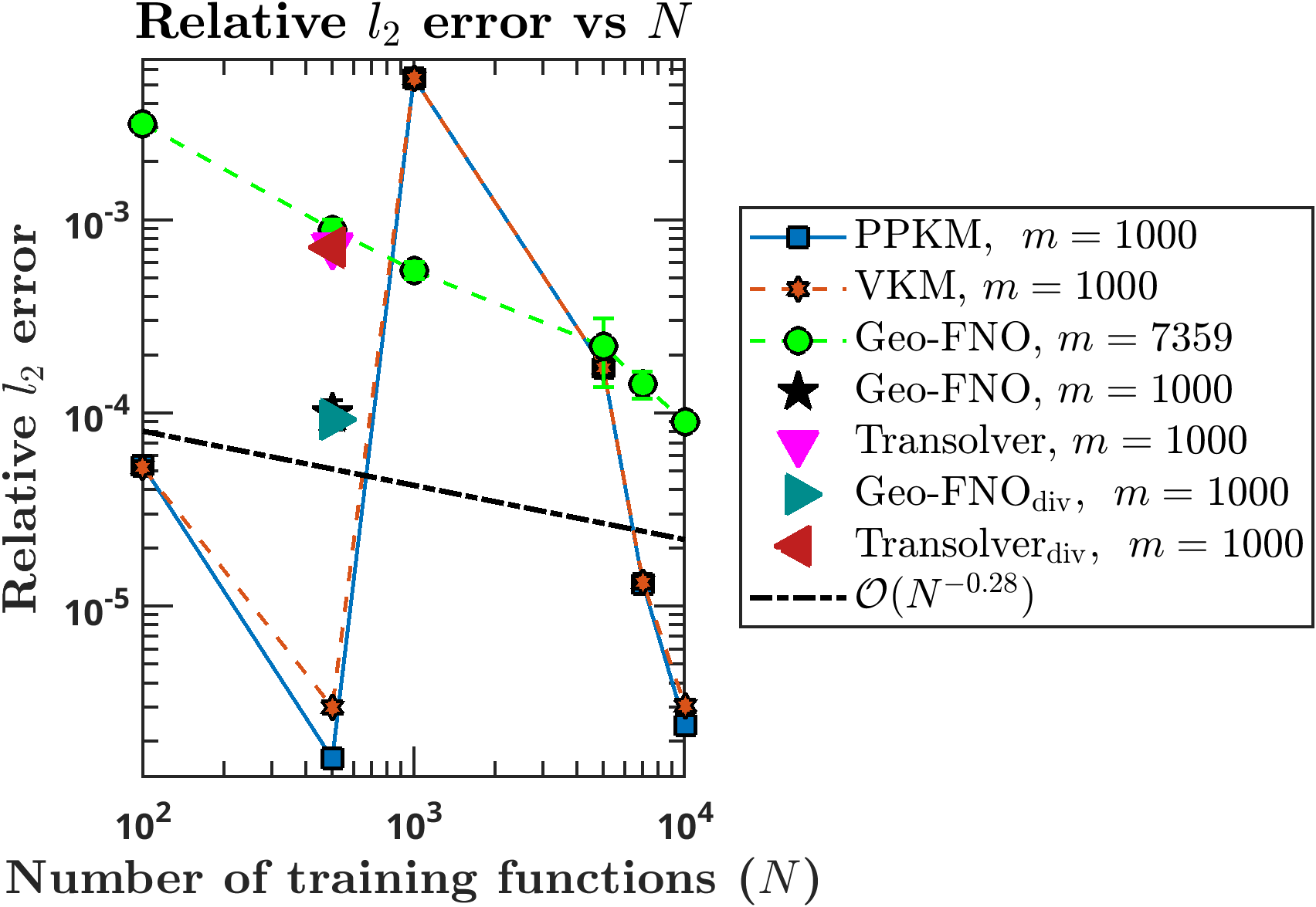}
        \caption{}
        \label{fig:backward_facing_step_errors}
    \end{subfigure} &
    \begin{subfigure}{0.5\linewidth}
        \centering
        \includegraphics[width=\linewidth]{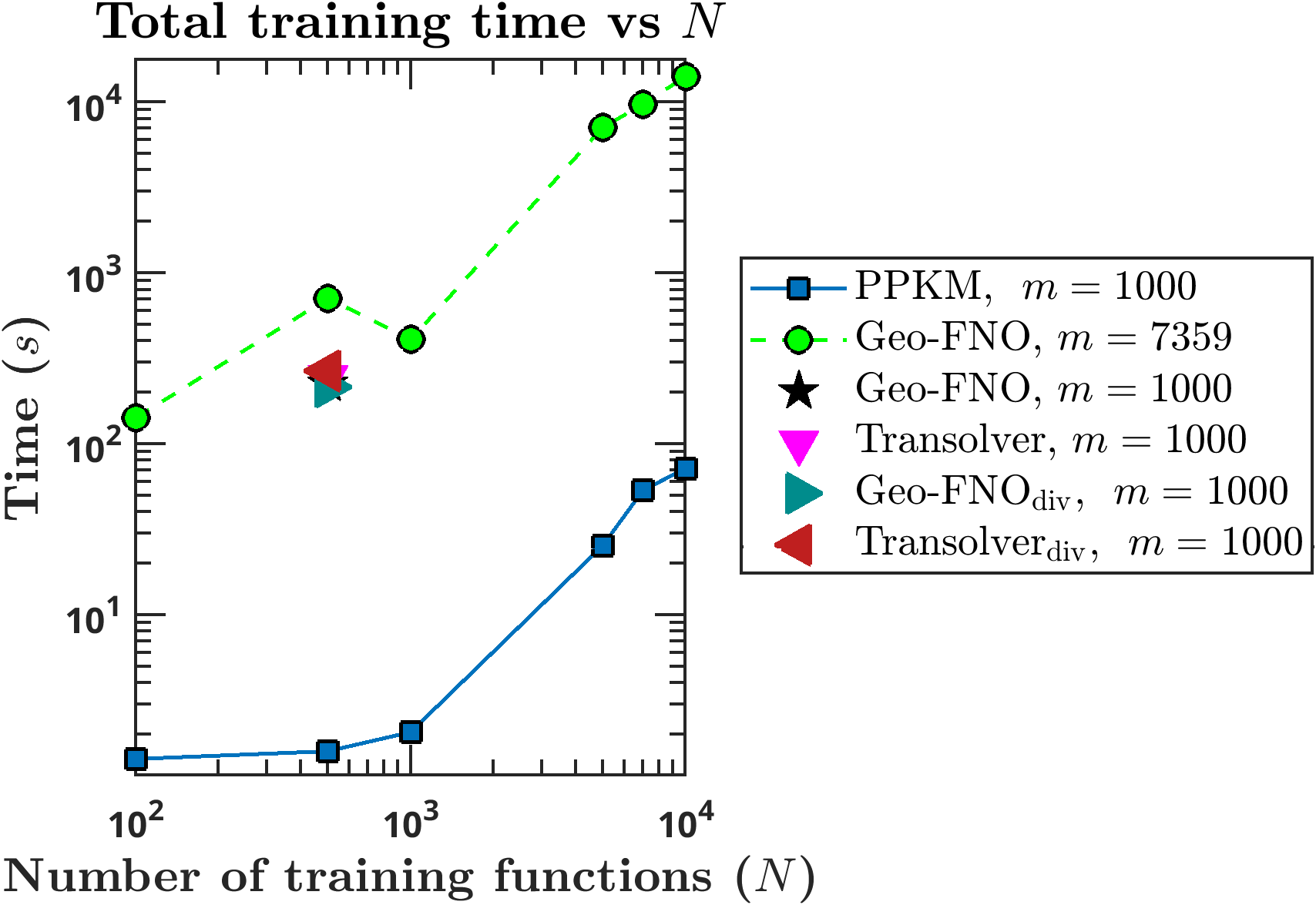}
        \caption{}
        \label{fig:backward_facing_step_times}
    \end{subfigure}
    \end{tabular}
    \caption{The 2D laminar backward-facing step problem. (\textbf{A}) and (\textbf{B}) show examples of an input function (the initial velocity) and an output function (the final velocity), respectively. (\textbf{C}) shows the parabolic profile of the inlet velocity. (\textbf{D}) and (\textbf{E}) show the test relative $\ell_2$ errors and training runtimes as functions of $N$.}
    \label{fig:backward_facing_step}
\end{figure}
This problem focuses on 2D laminar fluid flow over an isothermal backward-facing step at $Re=800$~\citep{chiang1999numerical}~\footnote{\url{https://su2code.github.io/tutorials/Inc_Laminar_Step}}. The domain of the problem is $\Omega_a = \Omega_v = [0, 15] \times [-0.5, 0.5]$. Typically, flow simulations for this problem involve fluid flow in both the upstream and downstream channels. We simplified the geometry by only simulating the downstream channel of the step. Let the channel height and width be $H$ and $15H$, respectively. The upper half of the left boundary was specified with the following initial condition,
\begin{align}
u_1(x_2) &= \nu x_2 * (0.5 - x_2), 0 \le x_2 \le 0.5,
\end{align}
where $\nu$ is the kinematic viscosity, $u_1$ is the horizontal velocity component, and $x_2$ is the vertical spatial coordinate. The velocity was initialized to be zero everywhere else. See Figure~\ref{fig:backward_facing_step_input_bc} for an example of the parabolic inlet BC. The top, bottom, and bottom half of the left (corresponding to the step face) boundaries were prescribed no-slip BCs. A zero pressure BC was imposed on the right boundary. We ran the SU2 simulation until $T=5$ with $\Delta t=10^{-3}$. $\nu$ was randomly sampled from the random distribution $\mathcal{N}_{[1, 36]}(18, 18)$ which resulted in Reynolds number range $[28, 900]$ across the simulations. A triangular mesh with $7359$ points was used. The operator of interest was $\G: \bu(x, 0) \rightarrow \bu(y, 5)$. These configuration details are provided in row~3 in Table~\ref{tab:exp_config}. We show examples of input and output functions in Figures~\ref{fig:backward_facing_step_input} and~\ref{fig:backward_facing_step_output} respectively.

Here, our method rapidly reached the lowest errors around $N=1000$, possibly due to hidden low-rank structure in the input functions. A subsequent spike was then followed by another steady decrease in error for increasing $N$ as shown in Figure~\ref{fig:backward_facing_step_errors}. Regardless, our method achieved close to three orders of magnitude lower errors and two orders of magnitude faster training times (Figure~\ref{fig:backward_facing_step_times}) than the neural operators.
\subsection{2D buoyancy-driven cavity flow}
\begin{figure}[pos=!htpb]
    \centering
    \begin{tabular}{cc}
    \begin{subfigure}{0.39\linewidth}
        \centering
        \includegraphics[width=\linewidth]{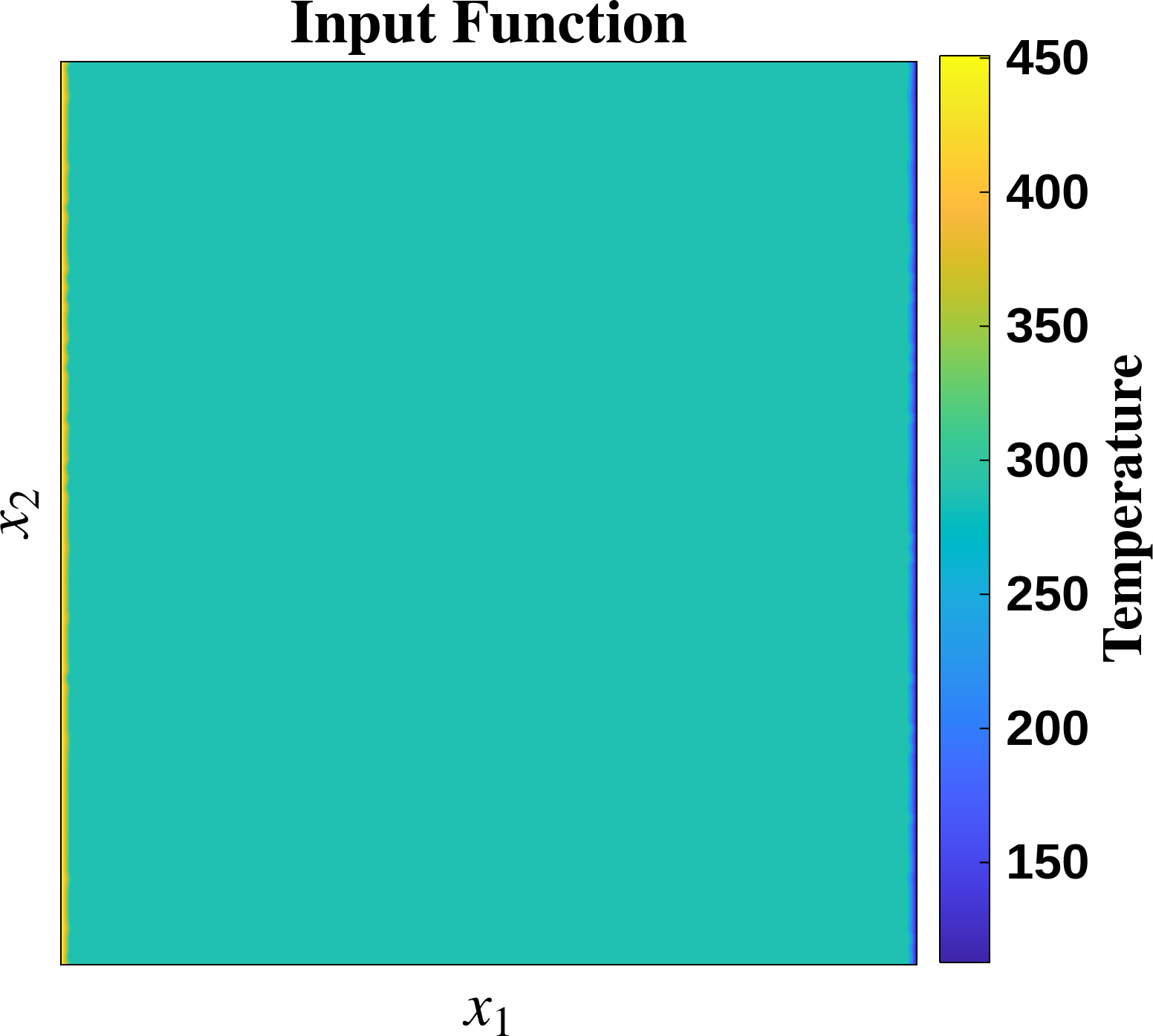}
        \caption{}
        \label{fig:buoyancy_cavity_flow_input}
    \end{subfigure} &
    \begin{subfigure}{0.39\linewidth}
        \centering
        \includegraphics[width=\linewidth]{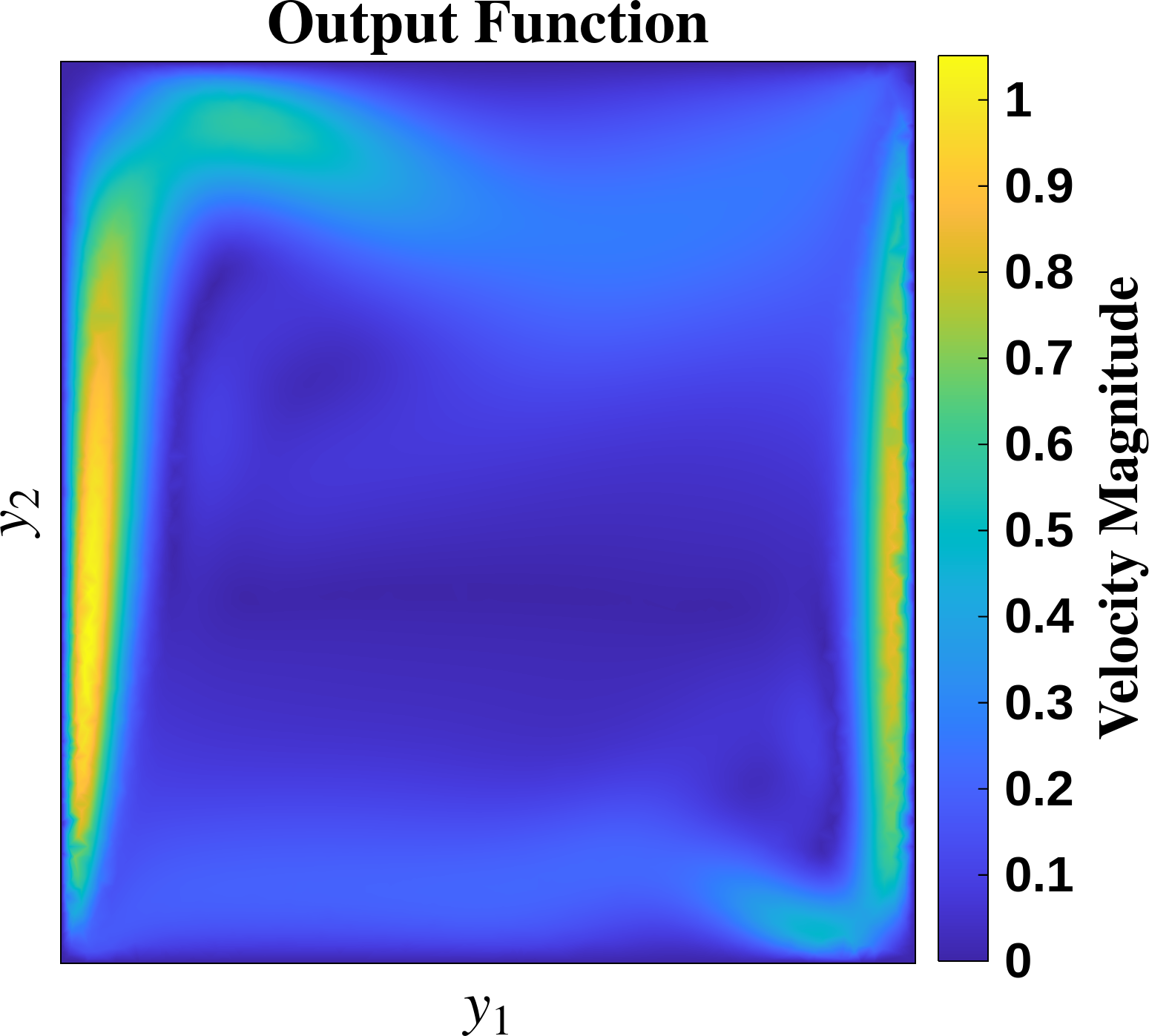}
        \caption{}
        \label{fig:buoyancy_cavity_flow_output}
    \end{subfigure} \\
    \begin{subfigure}{0.5\linewidth}
        \centering
        \includegraphics[width=\linewidth]{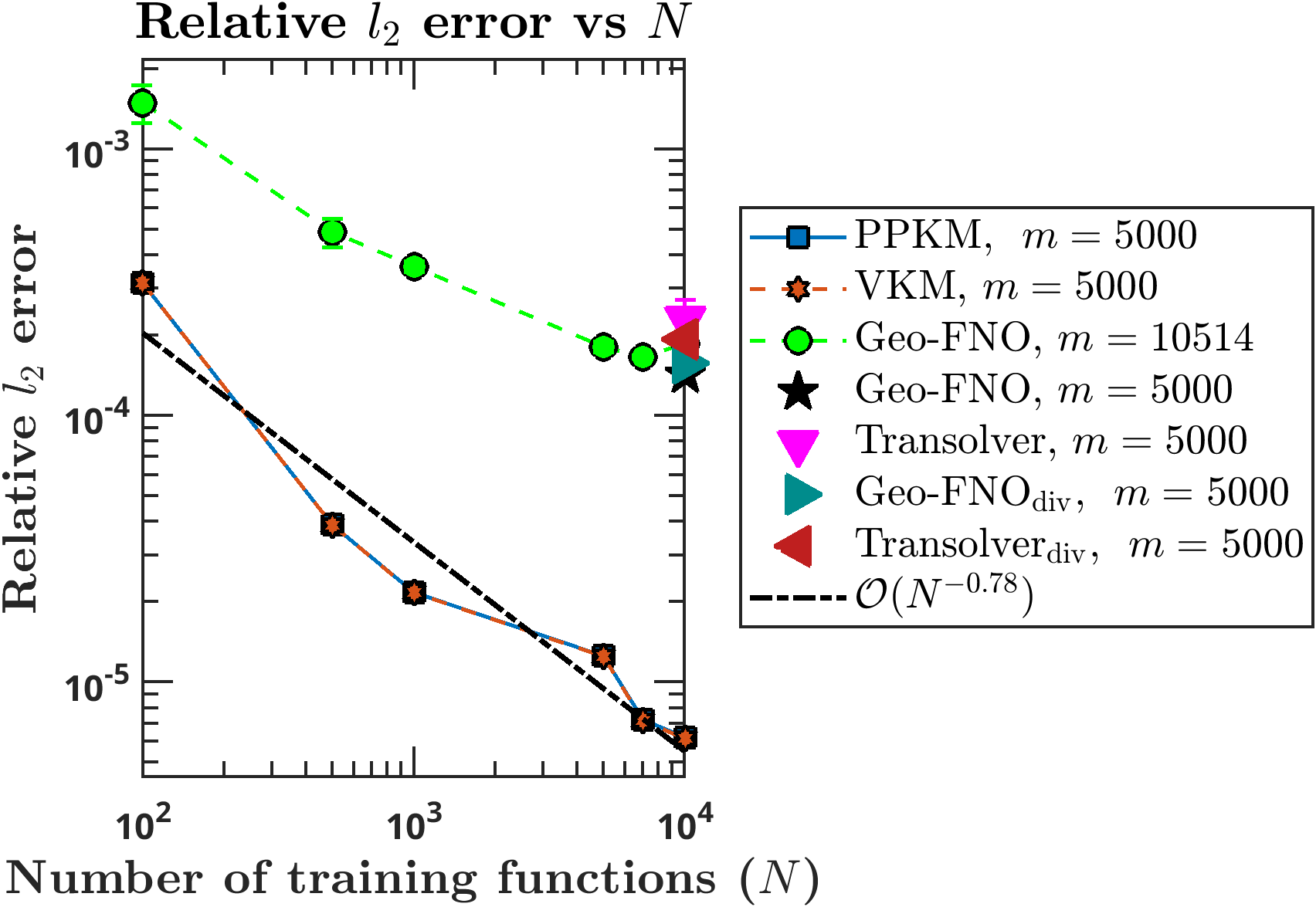}
        \caption{}
        \label{fig:buoyancy_cavity_flow_errors}
    \end{subfigure} &
    \begin{subfigure}{0.5\linewidth}
        \centering
        \includegraphics[width=\linewidth]{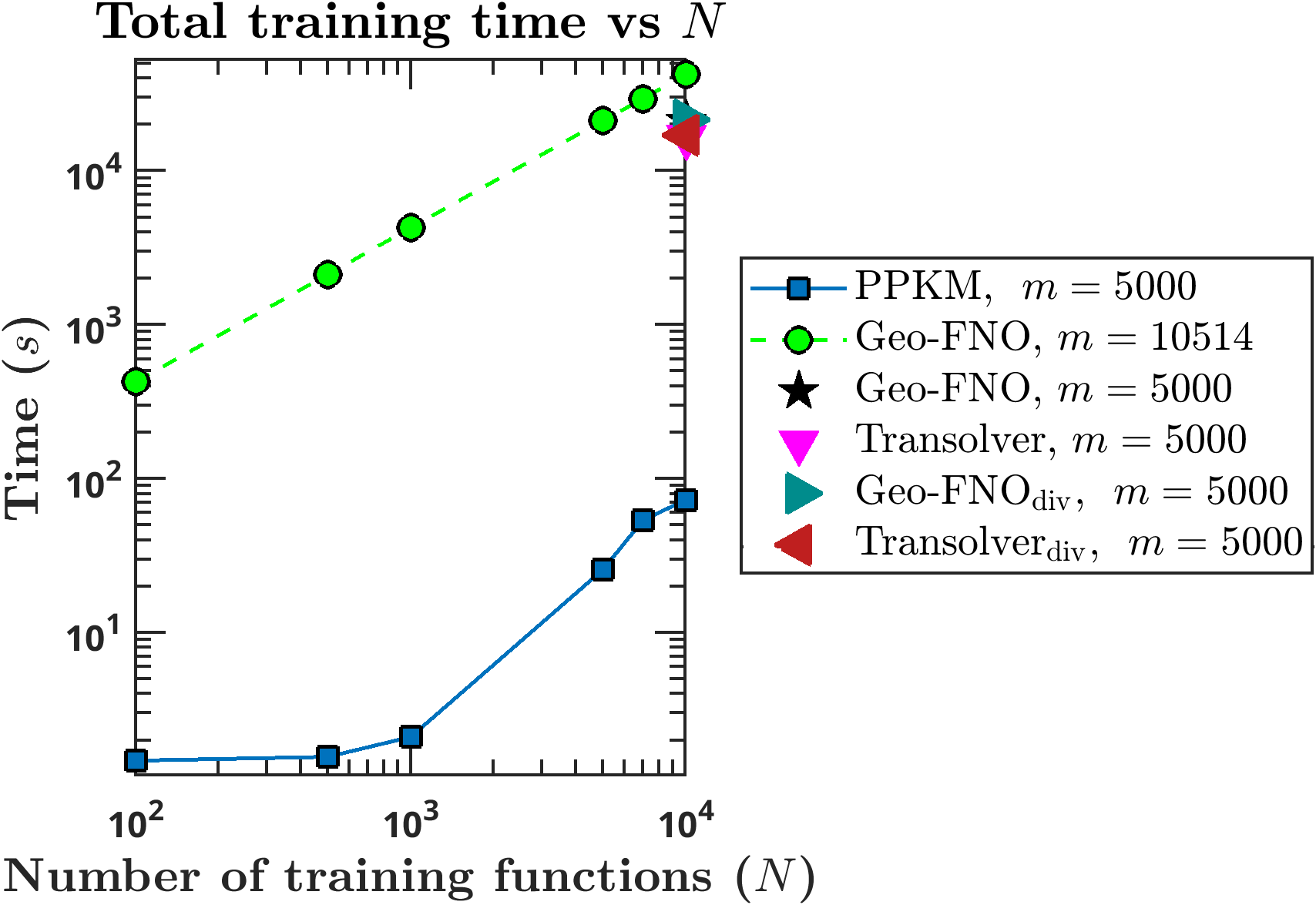}
        \caption{}
        \label{fig:buoyancy_cavity_flow_times}
    \end{subfigure}
    \end{tabular}
    \caption{The 2D laminar buoyancy-driven cavity flow problem. (\textbf{A}) and (\textbf{B}) show examples of an input function (the initial temperature) and an output function (the final velocity), respectively. (\textbf{C}) and (\textbf{D}) show the test relative $\ell_2$ errors and training runtimes as functions of $N$.}
    \label{fig:buoyancy_cavity_flow}
\end{figure}
This (modern) classic problem describes the 2D laminar buoyancy-driven flow in a square cavity $\Omega_a = \Omega_v=[0, 1]^2$ with adiabatic top and bottom walls and constant temperature on the left and right walls~\citep{sockol2003multigrid}~\footnote{\url{https://su2code.github.io/tutorials/Inc_Laminar_Cavity}}. Additionally, we also prescribed the gravitational body force $f=-g$. All four walls had no-slip BCs prescribed. We randomly initialized the temperature (in Kelvin) at the left wall $\texttt{T}^{\mathrm{left}}$ by sampling from the distribution $\mathcal{N}_{[300, 600]}(450, 100)$. $\texttt{T}^{\mathrm{right}}$ was set to $\texttt{T}^{\mathrm{left}} / 4$. The flow density was initialized to $5.97 \times 10^{-3} \ \mathrm{kg/m}^3$ which corresponds to a Rayleigh number of $10^6$. The temperature everywhere else in the domain was initialized to 288.15 K. The domain was discretized with a quadrilateral mesh containing $10,514$ points. We ran time dependent simulations in SU2 until $T=5$ with $\Delta t=10^{-3}$. The operator map being learned was $\G: \texttt{T}(x, 0) \rightarrow \bu(y, 5)$. These configuration details are provided in row~4 in Table~\ref{tab:exp_config}. We show examples of input and output functions in Figures~\ref{fig:buoyancy_cavity_flow_input} and~\ref{fig:buoyancy_cavity_flow_output}. Our method achieved orders of magnitude smaller errors (Figure~\ref{fig:buoyancy_cavity_flow_errors}) and training times (Figure~\ref{fig:buoyancy_cavity_flow_times}), and a faster convergence rate w.r.t $N$ than the neural operator baselines.
\subsection{2D merging vortices}
\begin{figure}[pos=!htpb]
    \centering
    \begin{tabular}{cc}
    \begin{subfigure}{0.39\linewidth}
        \centering
        \includegraphics[width=\linewidth]{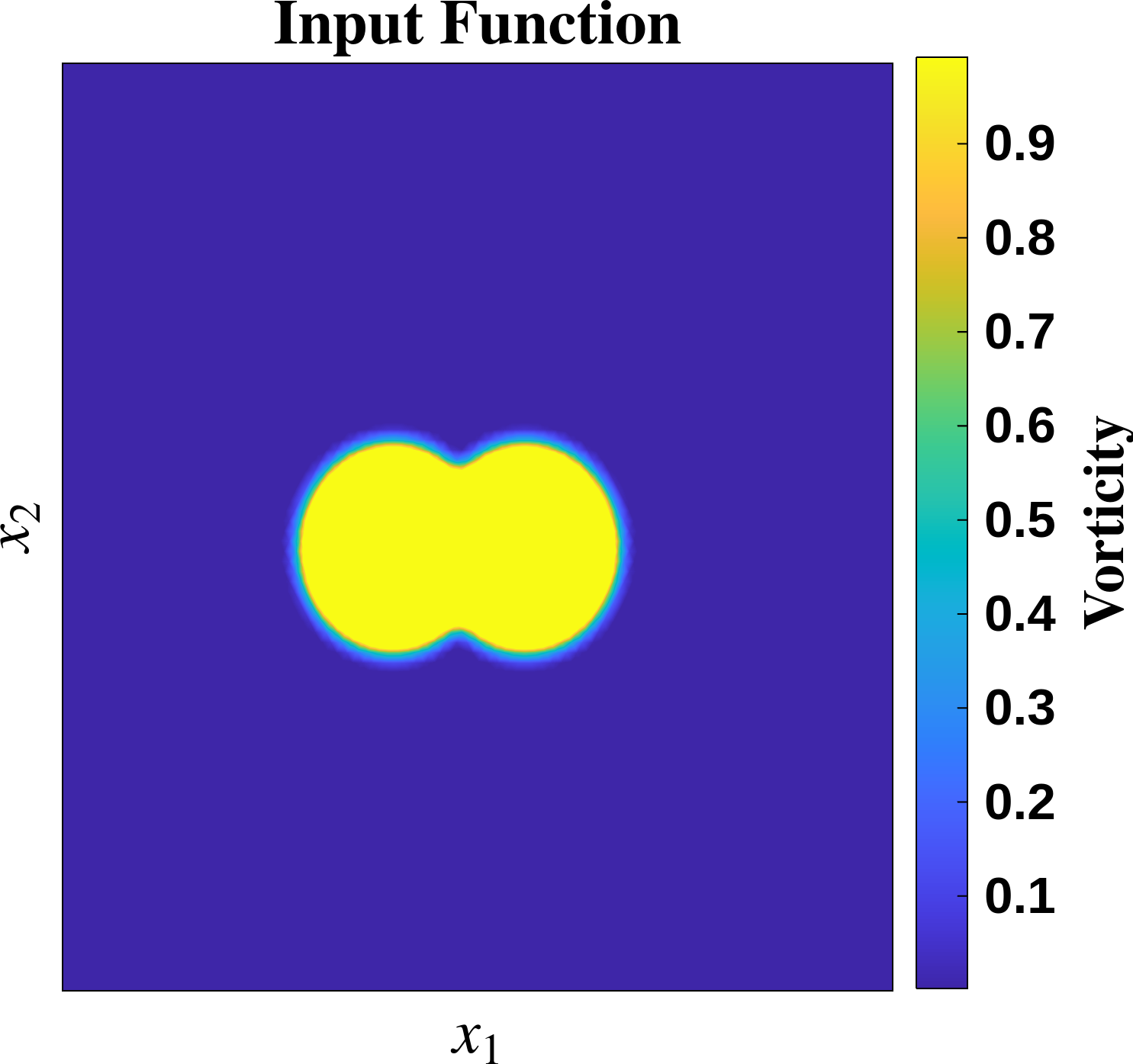}
        \caption{}
        \label{fig:merge_vortices_easier_input}
    \end{subfigure} &
    \begin{subfigure}{0.39\linewidth}
        \centering
        \includegraphics[width=\linewidth]{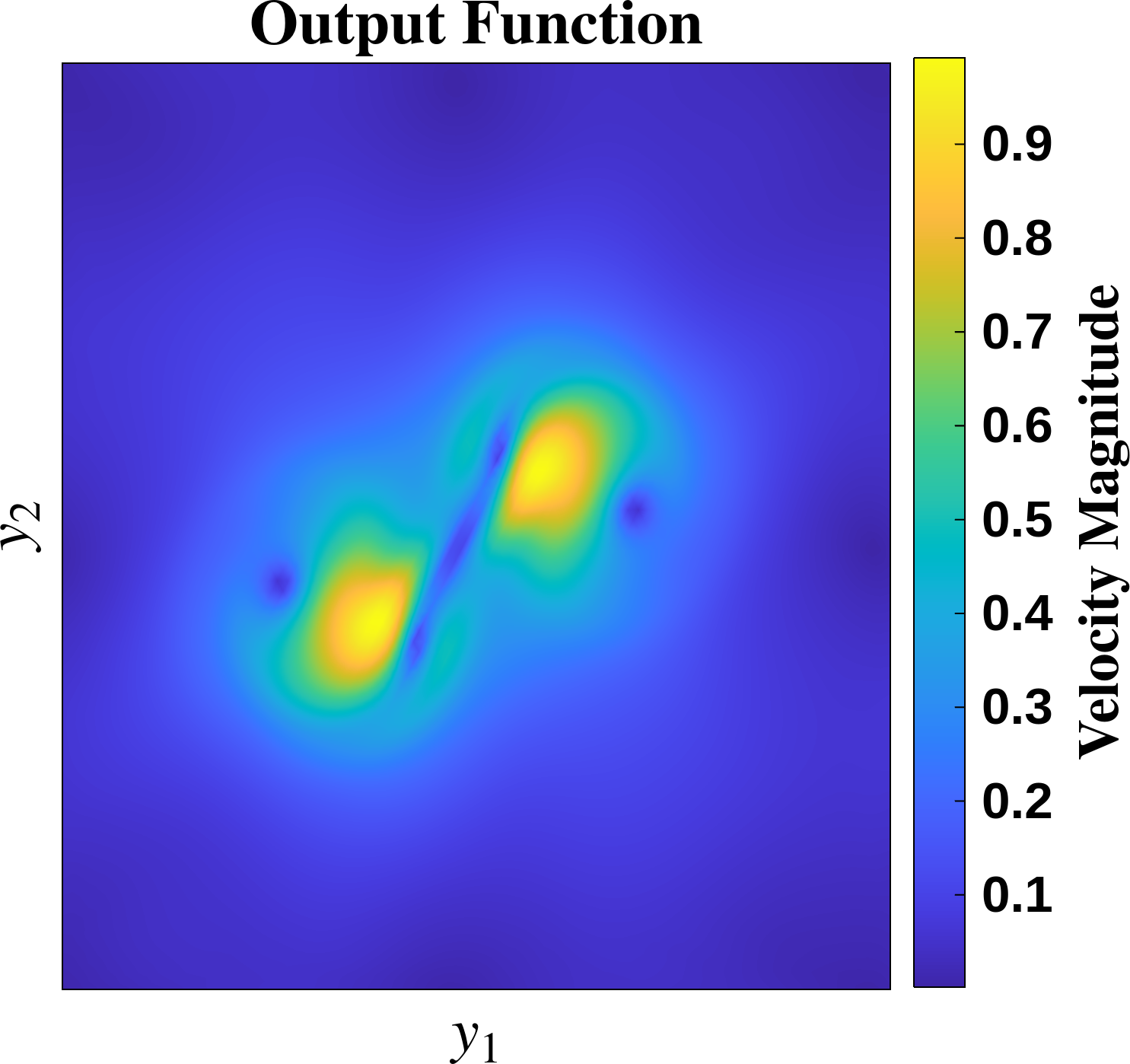}
        \caption{}
        \label{fig:merge_vortices_easier_output}
    \end{subfigure} \\
    \begin{subfigure}{0.5\linewidth}
        \centering
        \includegraphics[width=\linewidth]{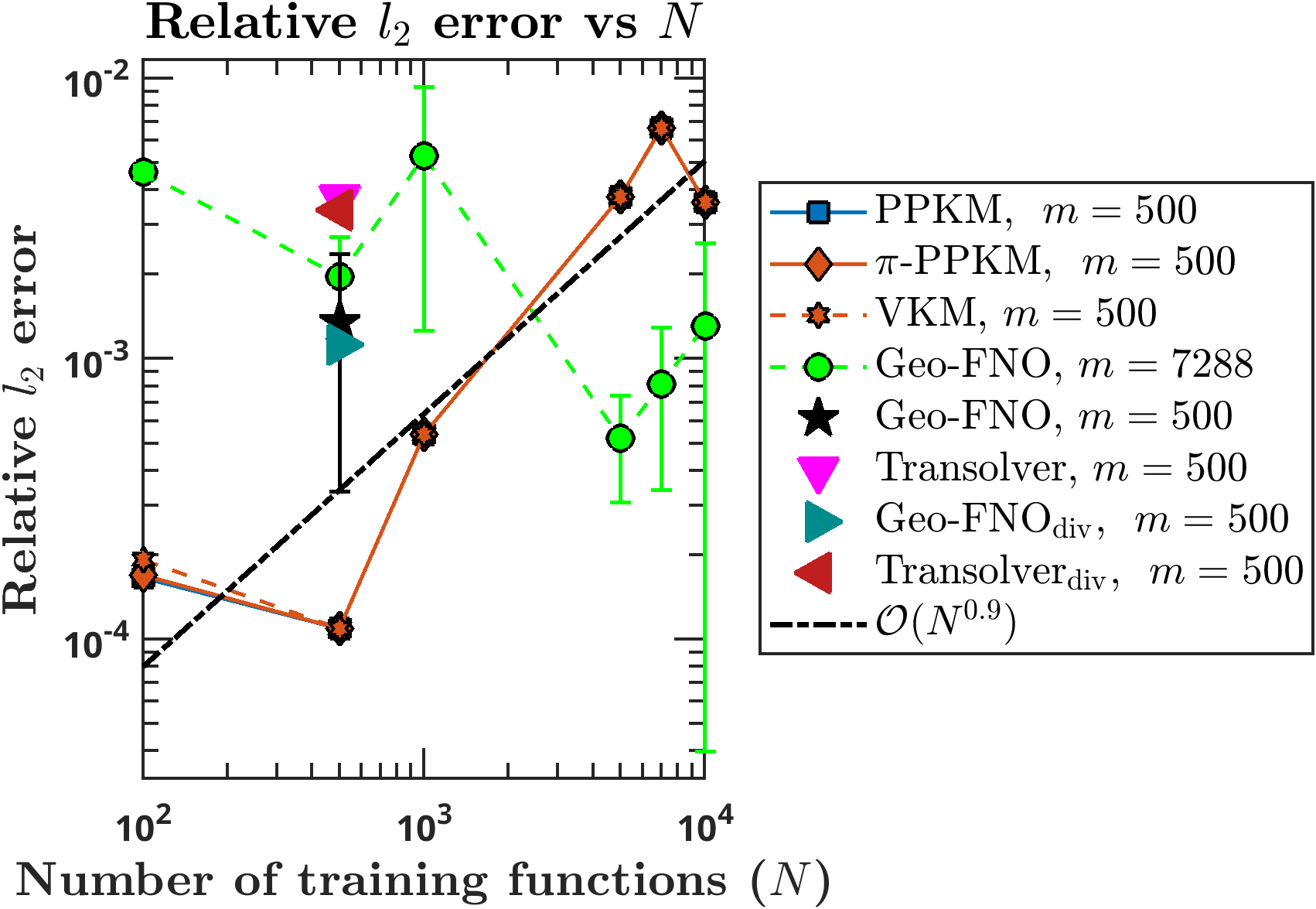}
        \caption{}
        \label{fig:merge_vortices_easier_errors}
    \end{subfigure} &
    \begin{subfigure}{0.5\linewidth}
        \centering
        \includegraphics[width=\linewidth]{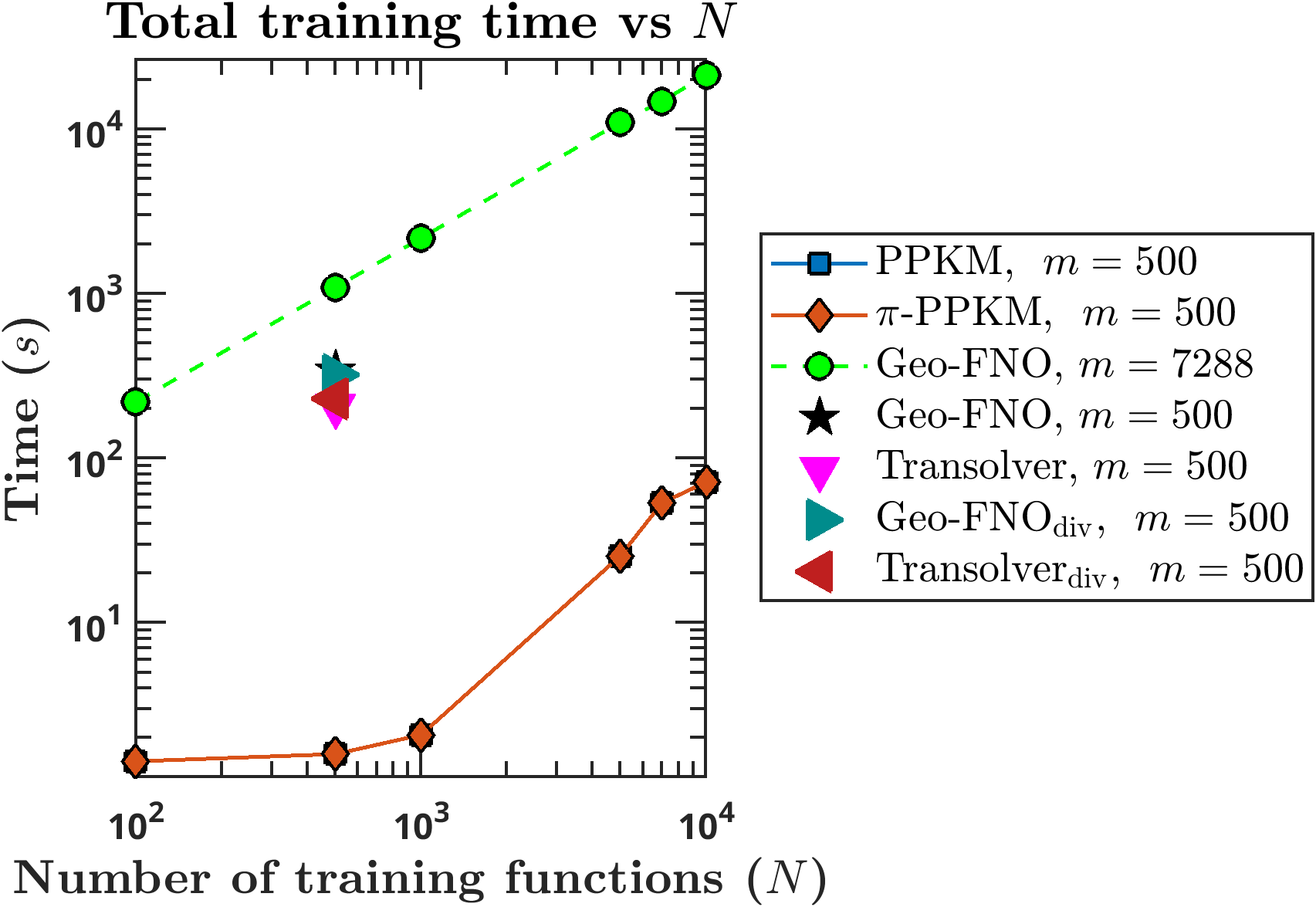}
        \caption{}
        \label{fig:merge_vortices_easier_times}
    \end{subfigure}
    \end{tabular}
    \caption{The 2D laminar merging vortices problem. (\textbf{A}) and (\textbf{B}) show examples of an input function (the initial vorticity containing two adjacent vortices) and an output function (the final velocity), respectively. (\textbf{C}) and (\textbf{D}) show the test relative $\ell_2$ errors and training runtimes as functions of $N$.}
    \label{fig:merge_vortices}
\end{figure}
This challenging problem describes the 2D laminar flow resulting from the merging of two vortices~\citep{dritschel1995general, josserand2007merging, trieling2010dynamics} in a square domain $\Omega_a = \Omega_v = [0, 2\pi]^2$. We started by prescribing an incompressible velocity field in the polar coordinates centered at $x = (x_1^c, x_2^c)$. Let $r = \sqrt{(x_1-x_1^c)^2 + (x_2-x_2^c)^2}$ where $x=(x_1, x_2) \in \Omega$. Then we prescribed the following Cartesian velocity field components:
\begin{align}
u_1(x) &= -\frac{x_2-x_2^c}{r}\, q(r), \\
u_2(x) &= \phantom{-}\frac{x_1-x_1^c}{r}\, q(r),\label{eq:merging_vortices_velocity}
\end{align}
with $q(r)=4r e^{-r \alpha}$, ensuring $q(0)=0$ for regularity at the vortex center. This construction yielded a divergence-free velocity field by construction. The associated scalar vorticity field $\omega$ is given by
\begin{align}
\omega(r) = (8 - 4\alpha r^{\alpha}) e^{-r^{\alpha}},
\end{align}
where $\alpha>0$ controls the sharpness of the vortex core. To the best of our knowledge as of the time of writing, SU2 does not allow the direct prescription of an initial vorticity field for time dependent simulations; we therefore initialized the flow through the velocity field in~\eqref{eq:merging_vortices_velocity}. We prescribed periodic BCs on all four walls. We sampled $\alpha \in [1, 14]$ from a uniform distribution. We used the same triangular mesh as the one in Section~\ref{sec:taylor_green}. We ran the SU2 simulation until $T=0.4$ with $\Delta t=10^{-3}$. The operator of interest here was $\G: \omega(x, t=0) \rightarrow \bu(y, 0.4)$. We provide these details in row~7 in Table~\ref{tab:exp_config}. We show examples of the input and output functions in Figures~\ref{fig:merge_vortices_easier_input} and~\ref{fig:merge_vortices_easier_output}.

The errors reported in Figure~\ref{fig:merge_vortices_easier_errors} show that our method rapidly achieves the lowest of all the errors just shy of $N=1000$. The errors then began to climb as a function of $N$. The difficulty of this problem is further shown by the unusually large error bars on the neural operators and the fact that there are no discernible differences in the accuracy between $\pi$-PPKM and PPKM despite this being a periodic problem. It is likely that our method's accuracy deteriorates for large $N$ due to the same reason as the no-vortex-shedding flow regime in Section~\ref{sec:flow_cylinder}; the output functions vary minimally with changes in the input functions. For a small training data set, $N=100$ or $N=500$, our method outperforms the neural operator baselines. These results where errors climb with $N$ also indicate the need for importance sampling and/or adaptive sampling of the input functions.
\section{Neural operator implementation details}
\label{app:neural_operator}
We outline relevant implementation details for the Geo-FNO~\citep{li2023fourier} and Transolver~\citep{wu2024transolver} models below.

\textbf{Geo-FNO}: The resolution is denoted by $s$ in the publicly available code by the authors of~\citep{li2023fourier}. We tested with $s = \{20, 30, 40, 50, 60\}$ for all the 2D problems and $s = \{10, 15, 20, 25\}$ for both 3D problems and picked the best performing $s$. We also tuned the number of modes used, searching over three distinct values for each problem including $\left\lfloor \frac{s}{2} \right\rfloor$. The channel dimension was varied over $\{32, 64, 128\}$. Following~\citep{li2023fourier}, we fixed the number of Fourier layers to four. Due to the Geo-FNO requiring the input and output functions to have the same $d$-dimensional domain, we made the following adjustments to the data sets:
\begin{itemize}
\item \textbf{2D Taylor--Green spacetime}: We treated this equivalently to a 3D spatial problem.
\item \textbf{2D Taylor--Green spacetime parametric map}: This was also treated as a 3D spatial problem except since the input functions are the $A$ and $\nu$ coefficients, we repeated them as vector-valued constant functions sampled on the same spacetime grid as the output functions.
\item \textbf{3D Species transport}: We padded the input functions in the \texttt{z} direction with zeros to make them functions of three variables.
\end{itemize}

\textbf{Transolver}: The hyperparameters in this model included the number of attention layers (varied over $\{3, 4, 5\}$), dimensions of the embeddings (varied over $\{32, 64, 128\}$), number of heads (varied over $\{4, 6, 8\}$), and number of slices (varied over $\{16, 32, 64\}$). In addition to matching the dimensionality of the input and output functions like the Geo-FNO, the Transolver model additionally requires that the input and output functions share the same underlying discretization. To achieve this for the 2D Taylor--Green vortices spacetime (both the spatiotemporal and parametric maps) and 3D species transport problems, we constructed a new grid as the union of the input and output discretizations. Both the input and output functions were evaluated on this new grid with zeros where their respective data was missing. The loss was computed only on the original output function grid.

\textbf{Training}: For both models, we used the Adam optimizer and the default learning rate schedule from the publicly available code. A cosine annealing learning rate schedule was used for the Geo-FNO and a OneCycle~\footnote{\url{https://docs.pytorch.org/docs/stable/generated/torch.optim.lr_scheduler.OneCycleLR.html}} learning rate schedule for the Transolver. For both, we set the maximum learning rate to $10^{-3}$ and the activation function to GELU. In all problems and for both models, we normalized the input spatial coordinates to $[0,1]^d$; we observed that omitting this step degraded the performance. We used the default normalization scheme for the input function samples themselves; specifically, the Geo-FNO rescales the input functions pointwise across all training functions to zero mean and unit variance, and the Transolver rescales input functions with a global normalization across all points \emph{and} training functions to zero mean and unit variance. Both were trained in single precision for $500$ epochs with a batch size of $20$ using the NCSA Delta GPU cluster~\footnote{\url{https://www.ncsa.illinois.edu/research/project-highlights/delta}} with NVIDIA A100 GPUs. Due to the excessive costs of training the Geo-FNO model on the 3D forced homogeneous isotropic turbulence problem, we trained it for $200$ epochs.

\textbf{Divergence computation}: We use high-order accurate local (stencil-wise) interpolation with polyharmonic spline (PHS) RBFs augmented with high-degree polynomials~\citep{flyer2016role} to interpolate the neural operators' output functions. Then, the interpolant's bases were differentiated analytically to compute the divergence. The order of accuracy of the interpolation was varied over $\{5, 7, 9\}$ and the lowest resulting divergence (pointwise maximum divergence averaged across all test functions) was reported; we excluded the boundary divergences for safety as the stencils at the boundary were one-sided and hence not as reliable.

\textbf{Divergence-loss training}: For the divergence-loss variants, we add a mean squared divergence loss term computed at the interior points of the predicted velocity fields, scaled by a hyperparameter $\lambda$, to the relative $\ell_2$ training loss. We compute the divergence using the approach described above and test $\lambda\in\{10^{-3},10^{-2},10^{-1},1\}$. We report the value of $\lambda$ that yields the lowest mean test relative $\ell_2$ error across the three seeds.
For the Geo-FNO divergence-loss results in Table~\ref{tab:best_results}, we selected $\lambda=10^{-2}$ for 2D flow past a cylinder (no vortex shedding), 2D spacetime Taylor--Green vortices, and 3D forced homogeneous isotropic turbulence problems, and $\lambda=10^{-3}$ for the remaining problems. For the Transolver divergence-loss results, we selected $\lambda=10^{-1}$ for 2D flow past a cylinder (no vortex shedding) and both the 2D spatial and spacetime Taylor--Green vortices, $\lambda=10^{-2}$ for the 3D forced homogeneous isotropic turbulence, and $\lambda=10^{-3}$ for the remaining problems.
%

\clearpage
\bibliographystyle{cas-model2-names}
\bibliography{references}

@phdthesis{kadri2011operator,
  title={Operator-Valued Kernels for Nonparametric Operator Estimation},
  author={Kadri, Hachem and Preux, Philippe and Duflos, Emmanuel and Canu, Stephane},
  year={2011},
  school={INRIA}
}

@article{kadri2016operator,
  title={Operator-valued kernels for learning from functional response data},
  author={Kadri, Hachem and Duflos, Emmanuel and Preux, Philippe and Canu, St{\'e}phane and Rakotomamonjy, Alain and Audiffren, Julien},
  journal={Journal of Machine Learning Research},
  volume={17},
  number={20},
  pages={1--54},
  year={2016}
}

@article{economon2020simulation,
  title={Simulation and adjoint-based design for variable density incompressible flows with heat transfer},
  author={Economon, Thomas D},
  journal={AIAA Journal},
  volume={58},
  number={2},
  pages={757--769},
  year={2020},
  publisher={American Institute of Aeronautics and Astronautics}
}

@article{chiang1999numerical,
  title={Numerical investigation of vortical evolution in a backward-facing step expansion flow},
  author={Chiang, TP and Sheu, Tony WH and Fang, CC},
  journal={Applied Mathematical Modelling},
  volume={23},
  number={12},
  pages={915--932},
  year={1999},
  publisher={Elsevier}
}

@article{economon2016su2,
  title={SU2: An open-source suite for multiphysics simulation and design},
  author={Economon, Thomas D and Palacios, Francisco and Copeland, Sean R and Lukaczyk, Trent W and Alonso, Juan J},
  journal={Aiaa Journal},
  volume={54},
  number={3},
  pages={828--846},
  year={2016},
  publisher={American Institute of Aeronautics and Astronautics}
}

@article{khorrami2024physics,
  title={A physics-encoded Fourier neural operator approach for surrogate modeling of divergence-free stress fields in solids},
  author={Khorrami, Mohammad S and Goyal, Pawan and Mianroodi, Jaber R and Svendsen, Bob and Benner, Peter and Raabe, Dierk},
  journal={arXiv preprint arXiv:2408.15408},
  year={2024}
}

@article{richter2022neural,
  title={Neural conservation laws: A divergence-free perspective},
  author={Richter-Powell, Jack and Lipman, Yaron and Chen, Ricky TQ},
  journal={Advances in Neural Information Processing Systems},
  volume={35},
  pages={38075--38088},
  year={2022}
}

@article{batlle2024kernel,
  title={Kernel methods are competitive for operator learning},
  author={Batlle, Pau and Darcy, Matthieu and Hosseini, Bamdad and Owhadi, Houman},
  journal={Journal of Computational Physics},
  volume={496},
  pages={112549},
  year={2024},
  publisher={Elsevier}
}

@article{li2024physics,
  title={Physics-informed neural operator for learning partial differential equations},
  author={Li, Zongyi and Zheng, Hongkai and Kovachki, Nikola and Jin, David and Chen, Haoxuan and Liu, Burigede and Azizzadenesheli, Kamyar and Anandkumar, Anima},
  journal={ACM/IMS Journal of Data Science},
  volume={1},
  number={3},
  pages={1--27},
  year={2024},
  publisher={ACM New York, NY}
}

@misc{goswami2022physics,
  title={Physics-informed deep neural operator networks},
  author={Goswami, S and Bora, A and Yu, Y and Karniadakis, GE},
  year={2022}
}

@article{li2023phase,
  title={Phase-Field DeepONet: Physics-informed deep operator neural network for fast simulations of pattern formation governed by gradient flows of free-energy functionals},
  author={Li, Wei and Bazant, Martin Z and Zhu, Juner},
  journal={Computer Methods in Applied Mechanics and Engineering},
  volume={416},
  pages={116299},
  year={2023},
  publisher={Elsevier}
}

@article{fuselier2016high,
  title={A high-order radial basis function (RBF) Leray projection method for the solution of the incompressible unsteady Stokes equations},
  author={Fuselier, Edward J and Shankar, Varun and Wright, Grady B},
  journal={Computers \& Fluids},
  volume={128},
  pages={41--52},
  year={2016},
  publisher={Elsevier}
}

@article{fuselier2017radial,
  title={A radial basis function method for computing Helmholtz--Hodge decompositions},
  author={Fuselier, Edward J and Wright, Grady B},
  journal={IMA Journal of Numerical Analysis},
  volume={37},
  number={2},
  pages={774--797},
  year={2017},
  publisher={Oxford University Press}
}

@article{shankar2015augmenting,
  title={Augmenting the immersed boundary method with Radial Basis Functions (RBFs) for the modeling of platelets in hemodynamic flows},
  author={Shankar, Varun and Wright, Grady B and Kirby, Robert M and Fogelson, Aaron L},
  journal={International Journal for Numerical Methods in Fluids},
  volume={79},
  number={10},
  pages={536--557},
  year={2015},
  publisher={Wiley Online Library}
}

@article{owhadi2023gaussian,
  title={Gaussian process hydrodynamics},
  author={Owhadi, Houman},
  journal={Applied Mathematics and Mechanics},
  volume={44},
  number={7},
  pages={1175--1198},
  year={2023},
  publisher={Springer}
}

@article{sommariva2009computing,
  title={Computing approximate Fekete points by QR factorizations of Vandermonde matrices},
  author={Sommariva, Alvise and Vianello, Marco},
  journal={Computers \& Mathematics with Applications},
  volume={57},
  number={8},
  pages={1324--1336},
  year={2009},
  publisher={Elsevier}
}

@article{li2023fourier,
  title={Fourier neural operator with learned deformations for pdes on general geometries},
  author={Li, Zongyi and Huang, Daniel Zhengyu and Liu, Burigede and Anandkumar, Anima},
  journal={Journal of Machine Learning Research},
  volume={24},
  number={388},
  pages={1--26},
  year={2023}
}

@article{wu2024transolver,
  title={Transolver: A fast transformer solver for PDEs on general geometries},
  author={Wu, Haixu and Luo, Huakun and Wang, Haowen and Wang, Jianmin and Long, Mingsheng},
  journal={arXiv preprint arXiv:2402.02366},
  year={2024}
}

@book{Frisch1995,
  title     = {Turbulence: The Legacy of A.\,N.\ Kolmogorov},
  author    = {Frisch, Uriel},
  year      = {1995},
  publisher = {Cambridge University Press},
  isbn      = {9780521451031},
  address   = {Cambridge, UK}
}

@article{Runge1901,
  author  = {Runge, Carl},
  title   = {{\"U}ber empirische Funktionen und die Interpolation zwischen {\"a}quidistanten Ordinaten},
  journal = {Zeitschrift f{\"u}r Mathematik und Physik},
  volume  = {46},
  pages   = {224--243},
  year    = {1901}
}

@article{epperson1987runge,
  title={On the Runge example},
  author={Epperson, James F},
  journal={The American Mathematical Monthly},
  volume={94},
  number={4},
  pages={329--341},
  year={1987},
  publisher={Taylor \& Francis}
}

@article{fornberg2007runge,
  title={The Runge phenomenon and spatially variable shape parameters in RBF interpolation},
  author={Fornberg, Bengt and Zuev, Julia},
  journal={Computers \& Mathematics with Applications},
  volume={54},
  number={3},
  pages={379--398},
  year={2007},
  publisher={Elsevier}
}

@article{platte2005polynomials,
  title={Polynomials and potential theory for Gaussian radial basis function interpolation},
  author={Platte, Rodrigo B and Driscoll, Tobin A},
  journal={SIAM Journal on Numerical Analysis},
  volume={43},
  number={2},
  pages={750--766},
  year={2005},
  publisher={SIAM}
}

@article{berrut2004barycentric,
  title={Barycentric lagrange interpolation},
  author={Berrut, Jean-Paul and Trefethen, Lloyd N},
  journal={SIAM review},
  volume={46},
  number={3},
  pages={501--517},
  year={2004},
  publisher={SIAM}
}

@article{boyd2010six,
  title={Six strategies for defeating the Runge Phenomenon in Gaussian radial basis functions on a finite interval},
  author={Boyd, John P},
  journal={Computers \& Mathematics with Applications},
  volume={60},
  number={12},
  pages={3108--3122},
  year={2010},
  publisher={Elsevier}
}

@article{boyd2009exponentially,
  title={Exponentially-convergent strategies for defeating the Runge phenomenon for the approximation of non-periodic functions, part I: single-interval schemes},
  author={Boyd, John P and Ong, Jun Rong},
  journal={Comput. Phys},
  volume={5},
  number={2-4},
  pages={484--497},
  year={2009}
}

@article{boyd2002comparison,
  title={A comparison of numerical algorithms for Fourier extension of the first, second, and third kinds},
  author={Boyd, John P},
  journal={Journal of Computational Physics},
  volume={178},
  number={1},
  pages={118--160},
  year={2002},
  publisher={Elsevier}
}

@article{bos2008calculation,
  title={On the calculation of approximate Fekete points: the univariate case},
  author={Bos, LEONARD PETER and Levenberg, Norman and others},
  journal={Electron. Trans. Numer. Anal},
  volume={30},
  pages={377--397},
  year={2008}
}

@article{de2012polynomial,
  title={Polynomial approximation and cubature at approximate Fekete and Leja points of the cylinder},
  author={De Marchi, Stefano and Marchioro, Martina and Sommariva, Alvise},
  journal={Applied Mathematics and Computation},
  volume={218},
  number={21},
  pages={10617--10629},
  year={2012},
  publisher={Elsevier}
}

@article{sommariva2010approximate,
  title={Approximate Fekete points for weighted polynomial interpolation},
  author={Sommariva, Alvise and Vianello, Marco and others},
  journal={Electron. Trans. Numer. Anal},
  volume={37},
  pages={1--22},
  year={2010}
}

@article{karvonen2021kernel,
  title={Kernel-based interpolation at approximate Fekete points},
  author={Karvonen, Toni and S{\"a}rkk{\"a}, Simo and Tanaka, Ken’ichiro},
  journal={Numerical Algorithms},
  volume={87},
  number={1},
  pages={445--468},
  year={2021},
  publisher={Springer}
}

@article{schaback2000adaptive,
  title={Adaptive greedy techniques for approximate solution of large RBF systems},
  author={Schaback, Robert and Wendland, Holger},
  journal={Numerical Algorithms},
  volume={24},
  number={3},
  pages={239--254},
  year={2000},
  publisher={Springer}
}

@article{de2005near,
  title={Near-optimal data-independent point locations for radial basis function interpolation},
  author={De Marchi, Stefano and Schaback, Robert and Wendland, Holger},
  journal={Advances in Computational Mathematics},
  volume={23},
  number={3},
  pages={317--330},
  year={2005},
  publisher={Springer}
}

@article{galy2021adaptive,
  title={Adaptive inducing points selection for gaussian processes},
  author={Galy-Fajou, Th{\'e}o and Opper, Manfred},
  journal={arXiv preprint arXiv:2107.10066},
  year={2021}
}

@article{fasshauer2007choosing,
  title={On choosing “optimal” shape parameters for RBF approximation},
  author={Fasshauer, Gregory E and Zhang, Jack G},
  journal={Numerical Algorithms},
  volume={45},
  number={1},
  pages={345--368},
  year={2007},
  publisher={Springer}
}

@article{cavoretto2024comparing,
  title={Comparing deterministic and statistical optimization techniques for the shape parameter selection in RBF interpolation},
  author={Cavoretto, Roberto and De Rossi, Alessandra and Haider, Adeeba and Lancellotti, Sandro},
  journal={Dolomites Research Notes on Approximation},
  volume={17},
  number={DRNA Volume 17.3},
  pages={48--55},
  year={2024},
  publisher={Padova University Press}
}

@article{ling2022stochastic,
  title={A stochastic extended Rippa’s algorithm for LpOCV},
  author={Ling, Leevan and Marchetti, Francesco},
  journal={Applied Mathematics Letters},
  volume={129},
  pages={107955},
  year={2022},
  publisher={Elsevier}
}

@article{rippa1999algorithm,
  title={An algorithm for selecting a good value for the parameter c in radial basis function interpolation},
  author={Rippa, Shmuel},
  journal={Advances in Computational Mathematics},
  volume={11},
  number={2},
  pages={193--210},
  year={1999},
  publisher={Springer}
}

@book{brent1973algorithms,
  author    = {Brent, Richard P.},
  title     = {Algorithms for Minimization Without Derivatives},
  publisher = {Prentice-Hall},
  address   = {Englewood Cliffs, NJ},
  year      = {1973}
}

@incollection{dekker1969finding,
  author    = {Dekker, T. J.},
  title     = {Finding a Zero by Means of Successive Linear Interpolation},
  booktitle = {Constructive Aspects of the Fundamental Theorem of Algebra},
  editor    = {Dejon, B. and Henrici, P.},
  publisher = {Wiley--Interscience},
  address   = {London},
  year      = {1969}
}

@misc{jax2018github,
  author = {James Bradbury and Roy Frostig and Peter Hawkins and Matthew James Johnson and Chris Leary and Dougal Maclaurin and George Necula and Adam Paszke and Jake Vander{P}las and Skye Wanderman-{M}ilne and Qiao Zhang},
  title = {{JAX}: composable transformations of {P}ython+{N}um{P}y programs},
  version = {0.3.13},
  year = {2018},
}

@article{provansal1987benard,
  title={B{\'e}nard-von K{\'a}rm{\'a}n instability: transient and forced regimes},
  author={Provansal, M and Mathis, C and Boyer, L},
  journal={Journal of Fluid Mechanics},
  volume={182},
  pages={1--22},
  year={1987},
  publisher={Cambridge University Press}
}

@book{chorin1993mathematical,
  title     = {A Mathematical Introduction to Fluid Mechanics},
  author    = {Chorin, Alexandre J. and Marsden, Jerrold E.},
  series    = {Texts in Applied Mathematics},
  volume    = {4},
  year      = {1993},
  publisher = {Springer},
  address   = {New York, NY},
  isbn      = {978-0-387-97918-2}
}

@article{goda1979multistep,
  title={A multistep technique with implicit difference schemes for calculating two-or three-dimensional cavity flows},
  author={Goda, Katuhiko},
  journal={Journal of computational physics},
  volume={30},
  number={1},
  pages={76--95},
  year={1979},
  publisher={Elsevier}
}

@article{chorin1967numerical,
  title={A numerical method for solving incompressible viscous flow problems},
  author={Chorin, Alexandre Joel},
  journal={Journal of computational physics},
  volume={2},
  number={1},
  pages={12--26},
  year={1967},
  publisher={Elsevier}
}

@article{johnston2004accurate,
  title={Accurate, stable and efficient Navier--Stokes solvers based on explicit treatment of the pressure term},
  author={Johnston, Hans and Liu, Jian-Guo},
  journal={Journal of Computational Physics},
  volume={199},
  number={1},
  pages={221--259},
  year={2004},
  publisher={Elsevier}
}

@book{holmes2012turbulence,
  title     = {Turbulence, Coherent Structures, Dynamical Systems and Symmetry},
  editor    = {Philip Holmes and John L. Lumley and Gahl Berkooz and Clarence W. Rowley},
  publisher = {Cambridge University Press},
address   = {Cambridge, UK; New York, NY, USA},
  year      = {2012},
  isbn      = {9781107008250},
}

@book{quarteroni2015reduced,
  title     = {Reduced Basis Methods for Partial Differential Equations: An Introduction},
  author    = {Alfio Quarteroni and Andrea Manzoni and Federico Negri},
  year      = {2016},
  address   = {Cham, Switzerland},
  publisher = {Springer},
  volume    = {92},
  isbn      = {978-3-319-15430-5},
}

@article{lee2020model,
  title={Model reduction of dynamical systems on nonlinear manifolds using deep convolutional autoencoders},
  author={Lee, Kookjin and Carlberg, Kevin T},
  journal={Journal of Computational Physics},
  volume={404},
  pages={108973},
  year={2020},
  publisher={Elsevier}
}

@article{fresca2020deep,
  title={Deep learning-based reduced order models in cardiac electrophysiology},
  author={Fresca, Stefania and Manzoni, Andrea and Ded{\`e}, Luca and Quarteroni, Alfio},
  journal={PloS one},
  volume={15},
  number={10},
  pages={e0239416},
  year={2020},
  publisher={Public Library of Science San Francisco, CA USA}
}

@article{sirignano2018dgm,
  title={DGM: A deep learning algorithm for solving partial differential equations},
  author={Sirignano, Justin and Spiliopoulos, Konstantinos},
  journal={Journal of computational physics},
  volume={375},
  pages={1339--1364},
  year={2018},
  publisher={Elsevier}
}

@article{thuerey2020deep,
  title={Deep learning methods for Reynolds-averaged Navier--Stokes simulations of airfoil flows},
  author={Thuerey, Nils and Wei{\ss}enow, Konstantin and Prantl, Lukas and Hu, Xiangyu},
  journal={AIAA journal},
  volume={58},
  number={1},
  pages={25--36},
  year={2020},
  publisher={American Institute of Aeronautics and Astronautics}
}

@inproceedings{sanchez2020learning,
  title={Learning to simulate complex physics with graph networks},
  author={Sanchez-Gonzalez, Alvaro and Godwin, Jonathan and Pfaff, Tobias and Ying, Rex and Leskovec, Jure and Battaglia, Peter},
  booktitle={International conference on machine learning},
  pages={8459--8468},
  year={2020},
  organization={PMLR}
}

@inproceedings{pfaff2020learning,
  title={Learning mesh-based simulation with graph networks},
  author={Pfaff, Tobias and Fortunato, Meire and Sanchez-Gonzalez, Alvaro and Battaglia, Peter},
  booktitle={International conference on learning representations},
  year={2020}
}

@article{PINNs1,
title = {Physics-informed neural networks: A deep learning framework for solving forward and inverse problems involving nonlinear partial differential equations},
journal = {Journal of Computational Physics},
volume = {378},
pages = {686-707},
year = {2019},
issn = {0021-9991},
author = {M. Raissi and P. Perdikaris and G.E. Karniadakis}
}

@incollection{boulle2024mathematical,
  title={A mathematical guide to operator learning},
  author={Boull{\'e}, Nicolas and Townsend, Alex},
  booktitle={Handbook of Numerical Analysis},
  volume={25},
  pages={83--125},
  year={2024},
  address={Amsterdam, The Netherlands},
  publisher={Elsevier}
}

@article{lu2019deeponet,
   title={Learning nonlinear operators via DeepONet based on the universal approximation theorem of operators},
   volume={3},
   ISSN={2522-5839},
   number={3},
   journal={Nature Machine Intelligence},
   publisher={Springer Science and Business Media LLC},
   author={Lu, Lu and Jin, Pengzhan and Pang, Guofei and Zhang, Zhongqiang and Karniadakis, George Em},
   year={2021},
   month=mar, pages={218–229} }

@article{li2020fourier,
  title={Fourier neural operator for parametric partial differential equations},
  author={Li, Zongyi and Kovachki, Nikola and Azizzadenesheli, Kamyar and Liu, Burigede and Bhattacharya, Kaushik and Stuart, Andrew and Anandkumar, Anima},
  journal={arXiv preprint arXiv:2010.08895},
  year={2020}
}

@article{gin2021deepgreen,
  title={DeepGreen: deep learning of Green’s functions for nonlinear boundary value problems},
  author={Gin, Craig R and Shea, Daniel E and Brunton, Steven L and Kutz, J Nathan},
  journal={Scientific reports},
  volume={11},
  number={1},
  pages={21614},
  year={2021},
  publisher={Nature Publishing Group UK London}
}

@article{boulle2022data,
  title={Data-driven discovery of Green’s functions with human-understandable deep learning},
  author={Boull{\'e}, Nicolas and Earls, Christopher J and Townsend, Alex},
  journal={Scientific reports},
  volume={12},
  number={1},
  pages={4824},
  year={2022},
  publisher={Nature Publishing Group UK London}
}

@misc{hao2023gnot,
      title={GNOT: A General Neural Operator Transformer for Operator Learning},
      author={Zhongkai Hao and Zhengyi Wang and Hang Su and Chengyang Ying and Yinpeng Dong and Songming Liu and Ze Cheng and Jian Song and Jun Zhu},
      year={2023},
      eprint={2302.14376},
}

@article{liu2025geometry,
  title={Geometry-informed neural operator transformer},
  author={Liu, Qibang and Zhong, Weiheng and Meidani, Hadi and Abueidda, Diab and Koric, Seid and Geubelle, Philippe},
  journal={arXiv preprint arXiv:2504.19452},
  year={2025}
}

@article{liu2024mitigating,
  title={Mitigating spectral bias for the multiscale operator learning},
  author={Liu, Xinliang and Xu, Bo and Cao, Shuhao and Zhang, Lei},
  journal={Journal of Computational Physics},
  volume={506},
  pages={112944},
  year={2024},
  publisher={Elsevier}
}

@article{mora2025operator,
  title={Operator learning with Gaussian processes},
  author={Mora, Carlos and Yousefpour, Amin and Hosseinmardi, Shirin and Owhadi, Houman and Bostanabad, Ramin},
  journal={Computer Methods in Applied Mechanics and Engineering},
  volume={434},
  pages={117581},
  year={2025},
  publisher={Elsevier}
}

@article{li2020neural,
  title={Neural operator: Graph kernel network for partial differential equations},
  author={Li, Zongyi and Kovachki, Nikola and Azizzadenesheli, Kamyar and Liu, Burigede and Bhattacharya, Kaushik and Stuart, Andrew and Anandkumar, Anima},
  journal={arXiv preprint arXiv:2003.03485},
  year={2020}
}

@article{zhang2025floating,
  title={Floating-Body Hydrodynamic Neural Networks},
  author={Zhang, Tianshuo and Zhai, Wenzhe and Yann, Rui and Gao, Jia and Cao, He and Xing, Xianglei},
  journal={arXiv preprint arXiv:2509.13783},
  year={2025}
}

@article{alfonsi2009reynolds,
    author = {Alfonsi, Giancarlo},
    title = {Reynolds-Averaged Navier–Stokes Equations for Turbulence Modeling},
    journal = {Applied Mechanics Reviews},
    volume = {62},
    number = {4},
    pages = {040802},
    year = {2009},
    month = {06},
}

@article{durbin2018some,
  title={Some recent developments in turbulence closure modeling},
  author={Durbin, Paul A},
  journal={Annual Review of Fluid Mechanics},
  volume={50},
  number={1},
  pages={77--103},
  year={2018},
  publisher={Annual Reviews}
}

@article{spalart2000strategies,
  title={Strategies for turbulence modelling and simulations},
  author={Spalart, Philippe R},
  journal={International journal of heat and fluid flow},
  volume={21},
  number={3},
  pages={252--263},
  year={2000},
  publisher={Elsevier}
}

@article{alam2015hybrid,
  title={Hybrid Reynolds-Averaged Navier--Stokes/Large-Eddy simulation models for flow around an iced wing},
  author={Alam, MF and Thompson, David S and Walters, D Keith},
  journal={Journal of aircraft},
  volume={52},
  number={1},
  pages={244--256},
  year={2015},
  publisher={American Institute of Aeronautics and Astronautics}
}

@article{cambon1999linear,
  title={Linear and nonlinear models of anisotropic turbulence},
  author={Cambon, Claude and Scott, Julian F},
  journal={Annual review of fluid mechanics},
  volume={31},
  number={1},
  pages={1--53},
  year={1999},
  publisher={Annual Reviews 4139 El Camino Way, PO Box 10139, Palo Alto, CA 94303-0139, USA}
}

@article{geuzaine2009gmsh,
  title={Gmsh: A 3-D finite element mesh generator with built-in pre-and post-processing facilities},
  author={Geuzaine, Christophe and Remacle, Jean-Fran{\c{c}}ois},
  journal={International journal for numerical methods in engineering},
  volume={79},
  number={11},
  pages={1309--1331},
  year={2009},
  publisher={Wiley Online Library}
}

@article{secomb2016hemodynamics,
  title={Hemodynamics},
  author={Secomb, Timothy W},
  journal={Comprehensive physiology},
  volume={6},
  number={2},
  pages={975--1003},
  year={2016},
  publisher={Wiley Online Library}
}

@article{castorrini2023investigations,
  title={Investigations on offshore wind turbine inflow modelling using numerical weather prediction coupled with local-scale computational fluid dynamics},
  author={Castorrini, Alessio and Gentile, Sabrina and Geraldi, Edoardo and Bonfiglioli, Aldo},
  journal={Renewable and Sustainable Energy Reviews},
  volume={171},
  pages={113008},
  year={2023},
  publisher={Elsevier}
}

@article{kwak2009cfd,
  title={CFD for incompressible flows at NASA Ames},
  author={Kwak, Dochan and Kiris, Cetin},
  journal={Computers \& fluids},
  volume={38},
  number={3},
  pages={504--510},
  year={2009},
  publisher={Elsevier}
}

@article{perktold1995computer,
  title={Computer simulation of local blood flow and vessel mechanics in a compliant carotid artery bifurcation model},
  author={Perktold, Karl and Rappitsch, Gerhard},
  journal={Journal of biomechanics},
  volume={28},
  number={7},
  pages={845--856},
  year={1995},
  publisher={Elsevier}
}

@article{thwaites1960incompressible,
  title={Incompressible aerodynamics: an account of the theory and observation of the steady flow of incompressible fluid past aerofoils, wings, and other bodies},
  author={Thwaites, Bryan and Street, RE},
  journal={Physics Today},
  volume={13},
  number={12},
  pages={60--61},
  year={1960},
  publisher={American Institute of Physics}
}

@article{womersley1955method,
  title={Method for the calculation of velocity, rate of flow and viscous drag in arteries when the pressure gradient is known},
  author={Womersley, John R},
  journal={The Journal of physiology},
  volume={127},
  number={3},
  pages={553},
  year={1955}
}

@article{mcdonald1955relation,
  title={The relation of pulsatile pressure to flow in arteries},
  author={McDonald, DA},
  journal={The Journal of physiology},
  volume={127},
  number={3},
  pages={533},
  year={1955}
}

@article{kwak1986three,
  title={A three-dimensional incompressible Navier-Stokes flow solver using primitive variables},
  author={Kwak, Dochan and Chang, James LC and Shanks, Samuel P and Chakravarthy, Sukumar R},
  journal={AIAA journal},
  volume={24},
  number={3},
  pages={390--396},
  year={1986}
}

@article{chang1988numerical,
  title={Numerical simulation methods of incompressible flows and an application to the space shuttle main engine},
  author={Chang, JLC and Kwak, D and Rogers, SE and Yang, R-J},
  journal={International Journal for numerical methods in fluids},
  volume={8},
  number={10},
  pages={1241--1268},
  year={1988},
  publisher={Wiley Online Library}
}

@article{najm2005modeling,
  title={Modeling low Mach number reacting flow with detailed chemistry and transport},
  author={Najm, Habib N and Knio, Omar M},
  journal={Journal of Scientific Computing},
  volume={25},
  number={1},
  pages={263--287},
  year={2005},
  publisher={Springer}
}

@article{brown2001accurate,
  title={Accurate projection methods for the incompressible Navier--Stokes equations},
  author={Brown, David L and Cortez, Ricardo and Minion, Michael L},
  journal={Journal of computational physics},
  volume={168},
  number={2},
  pages={464--499},
  year={2001},
  publisher={Elsevier}
}

@article{guermond2006overview,
  title={An overview of projection methods for incompressible flows},
  author={Guermond, Jean-Luc and Minev, Peter and Shen, Jie},
  journal={Computer methods in applied mechanics and engineering},
  volume={195},
  number={44-47},
  pages={6011--6045},
  year={2006},
  publisher={Elsevier}
}

@article{orszag1986boundary,
  title={Boundary conditions for incompressible flows},
  author={Orszag, Steven A and Israeli, Moshe and Deville, Michel O},
  journal={Journal of Scientific Computing},
  volume={1},
  number={1},
  pages={75--111},
  year={1986},
  publisher={Springer}
}

@article{karniadakis1991high,
  title={High-order splitting methods for the incompressible Navier-Stokes equations},
  author={Karniadakis, George Em and Israeli, Moshe and Orszag, Steven A},
  journal={Journal of computational physics},
  volume={97},
  number={2},
  pages={414--443},
  year={1991},
  publisher={Elsevier}
}

@article{kan1986secondorder,
    author = {van Kan, J.},
    title = {A Second-Order Accurate Pressure-Correction Scheme for Viscous Incompressible Flow},
    journal = {SIAM Journal on Scientific and Statistical Computing},
    volume = {7},
    number = {3},
    pages = {870-891},
    year = {1986},
}

@article{ahmed2018assessment,
  title        = {An assessment of solvers for saddle point problems emerging from the incompressible Navier--Stokes equations},
  author       = {Ahmed, Naveed and Bartsch, Clemens and John, Volker and Wilbrandt, Ulrich},
  journal      = {Computational Methods in Applied Mechanics and Engineering},
  volume       = {331},
  pages        = {492--513},
  year         = {2018},
  publisher    = {Elsevier},
  note         = {Based on WIAS Preprint 2408 (2017)},
}

@article{john2002higher,
  title={Higher order finite element methods and multigrid solvers in a benchmark problem for the 3D Navier--Stokes equations},
  author={John, Volker},
  journal={International Journal for Numerical Methods in Fluids},
  volume={40},
  number={6},
  pages={775--798},
  year={2002},
  publisher={Wiley Online Library}
}

@article{elman2002preconditioners,
  title={Preconditioners for saddle point problems arising in computational fluid dynamics},
  author={Elman, Howard C},
  journal={Applied Numerical Mathematics},
  volume={43},
  number={1-2},
  pages={75--89},
  year={2002},
  publisher={Elsevier}
}

@article{dick2015solving,
  title={Solving the fluid pressure Poisson equation using multigrid—evaluation and improvements},
  author={Dick, Christian and Rogowsky, Marcus and Westermann, R{\"u}diger},
  journal={IEEE transactions on visualization and computer graphics},
  volume={22},
  number={11},
  pages={2480--2492},
  year={2015},
  publisher={IEEE}
}

@article{jung2013heterogenous,
    author = {Jung, Hwi-Ryong and Kim, Sun-Tae and Noh, Junyong and Hong, Jeong-Mo},
    title = {A heterogeneous CPU–GPU parallel approach to a multigrid Poisson solver for incompressible fluid simulation},
    journal = {Computer Animation and Virtual Worlds},
    volume = {24},
    number = {3-4},
    pages = {185-193},
    year = {2013}
}

@article{kress2006time,
  title={Time step restrictions using semi-explicit methods for the incompressible Navier--Stokes equations},
  author={Kress, Wendy and L{\"o}tstedt, Per},
  journal={Computer methods in applied mechanics and engineering},
  volume={195},
  number={33-36},
  pages={4433--4447},
  year={2006},
  publisher={Elsevier}
}

@article{ohlberger2015reduced,
  title={Reduced basis methods: Success, limitations and future challenges},
  author={Ohlberger, Mario and Rave, Stephan},
  journal={arXiv preprint arXiv:1511.02021},
  year={2015}
}

@article{ionita2014data,
  title={Data-driven parametrized model reduction in the Loewner framework},
  author={Ionita, Antonio Cosmin and Antoulas, Athanasios C},
  journal={SIAM Journal on Scientific Computing},
  volume={36},
  number={3},
  pages={A984--A1007},
  year={2014},
  publisher={SIAM}
}

@article{tezzele2020enhancing,
  title={Enhancing CFD predictions in shape design problems by model and parameter space reduction},
  author={Tezzele, Marco and Demo, Nicola and Stabile, Giovanni and Mola, Andrea and Rozza, Gianluigi},
  journal={Advanced Modeling and Simulation in Engineering Sciences},
  volume={7},
  number={1},
  pages={40},
  year={2020},
  publisher={Springer}
}

@article{bhattacharya2021model,
  title={Model reduction and neural networks for parametric PDEs},
  author={Bhattacharya, Kaushik and Hosseini, Bamdad and Kovachki, Nikola B and Stuart, Andrew M},
  journal={The SMAI journal of computational mathematics},
  volume={7},
  pages={121--157},
  year={2021}
}

@article{suncica2003effective,
    author = {Canic, Suncica and Mikelic, Andro},
    title = {Effective Equations Modeling the Flow of a Viscous Incompressible Fluid through a Long Elastic Tube Arising in the Study of Blood Flow through Small Arteries},
    journal = {SIAM Journal on Applied Dynamical Systems},
    volume = {2},
    number = {3},
    pages = {431-463},
    year = {2003}
}

@article{janela20102783,
    title = {A 3D non-Newtonian fluid–structure interaction model for blood flow in arteries},
    journal = {Journal of Computational and Applied Mathematics},
    volume = {234},
    number = {9},
    pages = {2783-2791},
    year = {2010},
    author = {João Janela and Alexandra Moura and Adélia Sequeira},
}

@article{guy2005stability,
    title = {Stability of approximate projection methods on cell-centered grids},
    journal = {Journal of Computational Physics},
    volume = {203},
    number = {2},
    pages = {517-538},
    year = {2005},
    issn = {0021-9991},
    author = {Robert D. Guy and Aaron L. Fogelson}
}

@article{bell1989secondorder,
    title = {A second-order projection method for the incompressible navier-stokes equations},
    journal = {Journal of Computational Physics},
    volume = {85},
    number = {2},
    pages = {257-283},
    year = {1989},
    issn = {0021-9991},
    author = {John B Bell and Phillip Colella and Harland M Glaz}
}

@InProceedings{raviartthomas1977,
    author="Raviart, P. A.
        and Thomas, J. M.",
    editor="Galligani, Ilio
        and Magenes, Enrico",
    title="A mixed finite element method for 2-nd order elliptic problems",
    booktitle="Mathematical Aspects of Finite Element Methods",
    year="1977",
    publisher="Springer Berlin Heidelberg",
    address="Berlin, Heidelberg",
    pages="292--315",
}

@article{brezzi1985two,
  title={Two families of mixed finite elements for second order elliptic problems},
  author={Brezzi, Franco and Douglas Jr, Jim and Marini, L Donatella},
  journal={Numerische Mathematik},
  volume={47},
  number={2},
  pages={217--235},
  year={1985},
  publisher={Springer}
}

@article{nedelec1980mixed,
    title = {Mixed finite elements in {$\mathbb{R}^3$}},
    author={N{\'e}d{\'e}lec, Jean-Claude},
    journal={Numerische Mathematik},
    volume={35},
    number={3},
    pages={315--341},
    year={1980},
    publisher={Springer}
}

@article{narcowich1994generalized,
  title={Generalized Hermite interpolation via matrix-valued conditionally positive definite functions},
  author={Narcowich, Francis J and Ward, Joseph D},
  journal={Mathematics of Computation},
  volume={63},
  number={208},
  pages={661--687},
  year={1994}
}

@article{narcowich2007divergence,
  title={Divergence-free RBFs on surfaces},
  author={Narcowich, Francis J and Ward, Joseph D and Wright, Grady B},
  journal={Journal of Fourier Analysis and Applications},
  volume={13},
  number={6},
  pages={643--663},
  year={2007},
  publisher={Springer}
}

@article{fuselier2009error,
  title={Error and stability estimates for surface-divergence free RBF interpolants on the sphere},
  author={Fuselier, Edward and Narcowich, Francis and Ward, Joseph and Wright, Grady},
  journal={Mathematics of Computation},
  volume={78},
  number={268},
  pages={2157--2186},
  year={2009}
}

@article{fuselier2008sobolev,
  title={Sobolev-type approximation rates for divergence-free and curl-free RBF interpolants},
  author={Fuselier, Edward},
  journal={Mathematics of Computation},
  volume={77},
  number={263},
  pages={1407--1423},
  year={2008}
}

@article{wendland2009divergence,
  title={Divergence-free kernel methods for approximating the Stokes problem},
  author={Wendland, Holger},
  journal={SIAM Journal on Numerical Analysis},
  volume={47},
  number={4},
  pages={3158--3179},
  year={2009},
  publisher={SIAM}
}

@article{jacob2002sampling,
  title={Sampling of periodic signals: A quantitative error analysis},
  author={Jacob, Mathews and Blu, Thierry and Unser, Michael},
  journal={IEEE Transactions on Signal Processing},
  volume={50},
  number={5},
  pages={1153--1159},
  year={2002},
  publisher={IEEE}
}

@article{nguyen2012imposing,
  title={Imposing periodic boundary condition on arbitrary meshes by polynomial interpolation},
  author={Nguyen, V-D and B{\'e}chet, Eric and Geuzaine, Christophe and Noels, Ludovic},
  journal={Computational Materials Science},
  volume={55},
  pages={390--406},
  year={2012},
  publisher={Elsevier}
}

@book{boissonnat2006effective,
  editor    = {Boissonnat, Jean‐Daniel and Teillaud, Monique},
  title     = {Effective Computational Geometry for Curves and Surfaces},
  series    = {Mathematics and Visualization},
  publisher = {Springer Berlin Heidelberg},
  address   = {Berlin, Heidelberg},
  year      = {2006},
  edition   = {1st},
  isbn      = {978-3-540-33258-9},
  ebookisbn = {978-3-540-33259-6},
}

@book{canuto2006spectral,
  author    = {Canuto, Claudio and Hussaini, M. Youssuff and Quarteroni, Alfio and Zang, Thomas A.},
  title     = {Spectral Methods: Fundamentals in Single Domains},
  series    = {Scientific Computation},
  publisher = {Springer Berlin Heidelberg},
  address   = {Berlin, Heidelberg},
  year      = {2006},
  isbn      = {978-3-540-30725-9},
  ebookisbn = {978-3-540-30726-6},
}

@book{boyd2001chebyshev,
    title={Chebyshev and Fourier spectral methods},
    author={Boyd, John P},
    year={2001},
    address   = {Berlin, Heidelberg},
    publisher={Courier Corporation}
}

@book{fornberg_1996,
    address   = {Cambridge, United Kingdom},
    series={Cambridge Monographs on Applied and Computational Mathematics},
    title={A Practical Guide to Pseudospectral Methods},
    publisher={Cambridge University Press},
    author={Fornberg, Bengt},
    year={1996},
    collection={Cambridge Monographs on Applied and Computational Mathematics}
}

@book{fasshauer2007meshfree,
    title={Meshfree approximation methods with Matlab (With Cd-rom)},
    author={Fasshauer, Gregory E},
    volume={6},
    year={2007},
    address   = {Singapore; Hackensack, NJ, USA},
    publisher={World Scientific Publishing Company}
}

@article{xiao2014periodized,
  title={Periodized radial basis functions, part I: theory},
  author={Xiao, Jianping and Boyd, John P},
  journal={Applied Numerical Mathematics},
  volume={86},
  pages={43--73},
  year={2014},
  publisher={Elsevier}
}

@article{kaarnioja2022fast,
  title={Fast approximation by periodic kernel-based lattice-point interpolation with application in uncertainty quantification},
  author={Kaarnioja, Vesa and Kazashi, Yoshihito and Kuo, Frances Y and Nobile, Fabio and Sloan, Ian H},
  journal={Numerische Mathematik},
  volume={150},
  number={1},
  pages={33--77},
  year={2022},
  publisher={Springer}
}

@article{cools2019lattice,
  title={Lattice algorithms for multivariate approximation in periodic spaces with general weight parameters},
  author={Cools, Ronald and Kuo, Frances Y and Nuyens, Dirk and Sloan, Ian H},
  journal={arXiv preprint arXiv:1910.06604},
  year={2019}
}

@article{shankar2015radial,
  title={Radial basis function (RBF)-based parametric models for closed and open curves within the method of regularized stokeslets},
  author={Shankar, Varun and Olson, Sarah D},
  journal={International Journal for Numerical Methods in Fluids},
  volume={79},
  number={6},
  pages={269--289},
  year={2015},
  publisher={Wiley Online Library}
}

@article{fuselier2015order,
  title={Order-preserving derivative approximation with periodic radial basis functions},
  author={Fuselier, Edward and Wright, Grady B},
  journal={Advances in Computational Mathematics},
  volume={41},
  number={1},
  pages={23--53},
  year={2015},
  publisher={Springer}
}

@book{owhadi2021kernel,
  title={Kernel Mode Decomposition and the programming of kernels},
  author={Owhadi, Houman and Scovel, Clint and Yoo, Gene Ryan},
  address   = {Cham, Switzerland},
  year={2021},
  publisher={Springer}
}

@article{romero2017kernel,
  title={Kernel-based reconstruction of space-time functions on dynamic graphs},
  author={Romero, Daniel and Ioannidis, Vassilis N and Giannakis, Georgios B},
  journal={IEEE Journal of Selected Topics in Signal Processing},
  volume={11},
  number={6},
  pages={856--869},
  year={2017},
  publisher={IEEE}
}

@article{cressie01121999,
    author = {Noel Cressie and Hsin-Cheng Huang},
    title = {Classes of Nonseparable, Spatio-Temporal Stationary Covariance Functions},
    journal = {Journal of the American Statistical Association},
    volume = {94},
    number = {448},
    pages = {1330--1339},
    year = {1999},
    publisher = {Taylor \& Francis},
}

@article{stein01032005,
    author = {Michael L Stein},
    title = {Space–Time Covariance Functions},
    journal = {Journal of the American Statistical Association},
    volume = {100},
    number = {469},
    pages = {310--321},
    year = {2005},
    publisher = {Taylor \& Francis},
}

@article{wendland1995piecewise,
  author  = {Wendland, Holger},
  title   = {Piecewise Polynomial, Positive Definite and Compactly Supported Radial Functions of Minimal Degree},
  journal = {Advances in Computational Mathematics},
  volume  = {4},
  number  = {1},
  pages   = {389--396},
  year    = {1995},
}

@manual{demarchi2018lectures,
  title        = {Lectures on Radial Basis Functions},
  author       = {De Marchi, Stefano and Perracchione, Emma},
  organization = {Department of Mathematics ``Tullio Levi-Civita'', University of Padua},
  address      = {Padua, Italy},
  year         = {2018},
}

@article{chen2023selection,
    title = {On the selection of a better radial basis function and its shape parameter in interpolation problems},
    journal = {Applied Mathematics and Computation},
    volume = {442},
    pages = {127713},
    year = {2023},
    issn = {0096-3003},
    author = {Chuin-Shan Chen and Amir Noorizadegan and D.L. Young and C.S. Chen},
}

@article{wang2002optimal,
    title = {On the optimal shape parameters of radial basis functions used for 2-D meshless methods},
    journal = {Computer Methods in Applied Mechanics and Engineering},
    volume = {191},
    number = {23},
    pages = {2611-2630},
    year = {2002},
    issn = {0045-7825},
    author = {J.G. Wang and G.R. Liu},
}

@article{krowiak2019choosing,
    author = {Krowiak, Artur and Podgórski, Jordan},
    title = {On choosing a value of shape parameter in radial basis function collocation methods},
    journal = {AIP Conference Proceedings},
    volume = {2116},
    number = {1},
    pages = {450020},
    year = {2019},
    month = {07},
    issn = {0094-243X},
}

@book{fasshauer2015kernel,
    author = {Fasshauer, Gregory and McCourt, Michael},
    title = {Kernel-based Approximation Methods using MATLAB},
    publisher = {World Scientific},
    year = {2015},
    address   = {Singapore; Hackensack, NJ, USA},
}

@book{williams2006gaussian,
    author = {Rasmussen, Carl Edward and Williams, Christopher K. I.},
    title = {Gaussian Processes for Machine Learning},
    publisher = {The MIT Press},
    year = {2005},
    month = {11},
    address   = {Cambridge, MA},
    isbn = {9780262256834},
}

@article{jing2024modified,
  title={Modified space-time radial basis function collocation method for solving three-dimensional transient elastodynamic problems},
  author={Jing, Xiaohan and Qiu, Lin and Wang, Fajie and Gu, Yan},
  journal={Engineering Analysis with Boundary Elements},
  volume={169},
  pages={106027},
  year={2024},
  publisher={Elsevier}
}

@article{ku2022space,
  title={Space--time polyharmonic radial polynomial basis functions for modeling saturated and unsaturated flows},
  author={Ku, Cheng-Yu and Hong, Li-Dan and Liu, Chih-Yu and Xiao, Jing-En},
  journal={Engineering with Computers},
  volume={38},
  number={6},
  pages={4947--4960},
  year={2022},
  publisher={Springer}
}

@article{ku2020space,
  title={Space--time radial basis function--based meshless approach for solving convection--diffusion equations},
  author={Ku, Cheng-Yu and Xiao, Jing-En and Liu, Chih-Yu},
  journal={Mathematics},
  volume={8},
  number={10},
  pages={1735},
  year={2020},
  publisher={MDPI}
}

@article{yue2019novel,
  title={A novel space--time meshless method for nonhomogeneous convection--diffusion equations with variable coefficients},
  author={Yue, Xingxing and Wang, Fajie and Hua, Qingsong and Qiu, Xiang-Yun},
  journal={Applied Mathematics Letters},
  volume={92},
  pages={144--150},
  year={2019},
  publisher={Elsevier}
}

@article{li2016estimation,
  title={Estimation of river pollution source using the space-time radial basis collocation method},
  author={Li, Zi and Mao, Xian-Zhong and Li, Tak Sing and Zhang, Shiyan},
  journal={Advances in Water Resources},
  volume={88},
  pages={68--79},
  year={2016},
  publisher={Elsevier}
}

@article{li2011global,
  title={Global multiquadric collocation method for groundwater contaminant source identification},
  author={Li, Zi and Mao, Xian-zhong},
  journal={Environmental modelling \& software},
  volume={26},
  number={12},
  pages={1611--1621},
  year={2011},
  publisher={Elsevier}
}

@article{posa1993simple,
  title={A simple description of spatial-temporal processes},
  author={Posa, Donato},
  journal={Computational Statistics \& Data Analysis},
  volume={15},
  number={4},
  pages={425--437},
  year={1993},
  publisher={Elsevier}
}

@article{myers2002space,
  title={Space-time radial basis functions},
  author={Myers, DE and De Iaco, Sandra and Posa, Donato and De Cesare, Luigi},
  journal={Computers \& Mathematics with Applications},
  volume={43},
  number={3-5},
  pages={539--549},
  year={2002},
  publisher={Elsevier}
}

@article{chen2005scientific,
  title={Scientific computing with radial basis functions},
  author={Chen, CS and Hon, YC and Schaback, RA},
  journal={Department of Mathematics, University of Southern Mississippi, Hattiesburg, MS},
  volume={39406},
  year={2005}
}

@article{baxter2010interpolation,
  title={The interpolation theory of radial basis functions},
  author={Baxter, Brad},
  journal={arXiv preprint arXiv:1006.2443},
  year={2010}
}

@inproceedings{sudret2017surrogate,
  title={Surrogate models for uncertainty quantification: An overview},
  author={Sudret, Bruno and Marelli, Stefano and Wiart, Joe},
  booktitle={2017 11th European conference on antennas and propagation (EUCAP)},
  pages={793--797},
  year={2017},
  organization={IEEE}
}

@article{asouti2023radial,
  title={Radial basis function surrogates for uncertainty quantification and aerodynamic shape optimization under uncertainties},
  author={Asouti, Varvara and Kontou, Marina and Giannakoglou, Kyriakos},
  journal={Fluids},
  volume={8},
  number={11},
  pages={292},
  year={2023},
  publisher={MDPI}
}

@article{zhu2018bayesian,
  title={Bayesian deep convolutional encoder--decoder networks for surrogate modeling and uncertainty quantification},
  author={Zhu, Yinhao and Zabaras, Nicholas},
  journal={Journal of Computational Physics},
  volume={366},
  pages={415--447},
  year={2018},
  publisher={Elsevier}
}

@article{kovachki2022multiscale,
  title={Multiscale modeling of materials: Computing, data science, uncertainty and goal-oriented optimization},
  author={Kovachki, Nikola and Liu, Burigede and Sun, Xingsheng and Zhou, Hao and Bhattacharya, Kaushik and Ortiz, Michael and Stuart, Andrew},
  journal={Mechanics of Materials},
  volume={165},
  pages={104156},
  year={2022},
  publisher={Elsevier}
}

@article{huang2022iterated,
  title={Iterated Kalman methodology for inverse problems},
  author={Huang, Daniel Zhengyu and Schneider, Tapio and Stuart, Andrew M},
  journal={Journal of Computational Physics},
  volume={463},
  pages={111262},
  year={2022},
  publisher={Elsevier}
}

@article{martin2012stochastic,
  title={A stochastic Newton MCMC method for large-scale statistical inverse problems with application to seismic inversion},
  author={Martin, James and Wilcox, Lucas C and Burstedde, Carsten and Ghattas, Omar},
  journal={SIAM Journal on Scientific Computing},
  volume={34},
  number={3},
  pages={A1460--A1487},
  year={2012},
  publisher={SIAM}
}

@article{narayan2021optimal,
  title={Optimal design for kernel interpolation: Applications to uncertainty quantification},
  author={Narayan, Akil and Yan, Liang and Zhou, Tao},
  journal={Journal of Computational Physics},
  volume={430},
  pages={110094},
  year={2021},
  publisher={Elsevier}
}

@article{tharakan2015scalable,
  title={Scalable kernel methods for uncertainty quantification},
  author={Tharakan, S and March, WB and Biros, G},
  journal={Recent Trends in Computational Engineering-CE2014: Optimization, Uncertainty, Parallel Algorithms, Coupled and Complex Problems},
  pages={3--28},
  year={2015},
  publisher={Springer}
}

@article{csaji2019distribution,
  title={Distribution-free uncertainty quantification for kernel methods by gradient perturbations},
  author={Cs{\'a}ji, Bal{\'a}zs Cs and Kis, Kriszti{\'a}n B},
  journal={Machine Learning},
  volume={108},
  number={8},
  pages={1677--1699},
  year={2019},
  publisher={Springer}
}

@article{rosales2021high,
  title={High-order finite element methods for a pressure Poisson equation reformulation of the Navier--Stokes equations with electric boundary conditions},
  author={Rosales, Rodolfo Ruben and Seibold, Benjamin and Shirokoff, David and Zhou, Dong},
  journal={Computer Methods in Applied Mechanics and Engineering},
  volume={373},
  pages={113451},
  year={2021},
  publisher={Elsevier}
}

@article{guermond2015high,
  title={High-order time stepping for the incompressible Navier--Stokes equations},
  author={Guermond, Jean-Luc and Minev, Peter},
  journal={SIAM Journal on Scientific Computing},
  volume={37},
  number={6},
  pages={A2656--A2681},
  year={2015},
  publisher={SIAM}
}

@article{guermond2017high,
  title={High-order time stepping for the Navier--Stokes equations with minimal computational complexity},
  author={Guermond, Jean-Luc and Minev, Peter D},
  journal={Journal of Computational and Applied Mathematics},
  volume={310},
  pages={92--103},
  year={2017},
  publisher={Elsevier}
}

@article{lowery2024kernel,
  title={Kernel neural operators (KNOs) for scalable, memory-efficient, geometrically-flexible operator learning},
  author={Lowery, Matthew and Turnage, John and Morrow, Zachary and Jakeman, John D and Narayan, Akil and Zhe, Shandian and Shankar, Varun},
  journal={arXiv preprint arXiv:2407.00809},
  year={2024}
}

@article{tritton1959experiments,
  title={Experiments on the flow past a circular cylinder at low Reynolds numbers},
  author={Tritton, David J},
  journal={Journal of Fluid Mechanics},
  volume={6},
  number={4},
  pages={547--567},
  year={1959},
  publisher={Cambridge University Press}
}

@article{rajani2009numerical,
  title={Numerical simulation of laminar flow past a circular cylinder},
  author={Rajani, BN and Kandasamy, A and Majumdar, Sekhar},
  journal={Applied Mathematical Modelling},
  volume={33},
  number={3},
  pages={1228--1247},
  year={2009},
  publisher={Elsevier}
}

@article{li1991numerical,
  title={Numerical study of laminar flow past one and two circular cylinders},
  author={Li, J and Chambarel, A and Donneaud, M and Martin, R},
  journal={Computers \& fluids},
  volume={19},
  number={2},
  pages={155--170},
  year={1991},
  publisher={Elsevier}
}

@article{jackson1987finite,
  title={A finite-element study of the onset of vortex shedding in flow past variously shaped bodies},
  author={Jackson, CP},
  journal={Journal of fluid Mechanics},
  volume={182},
  pages={23--45},
  year={1987},
  publisher={Cambridge University Press}
}

@article{taylor1937mechanism,
  title={Mechanism of the production of small eddies from large ones},
  author={Taylor, Geoffrey Ingram and Green, Albert Edward},
  journal={Proceedings of the Royal Society of London. Series A-Mathematical and Physical Sciences},
  volume={158},
  number={895},
  pages={499--521},
  year={1937},
  publisher={The Royal Society London}
}

@article{kim1985application,
  title={Application of a fractional-step method to incompressible Navier-Stokes equations},
  author={Kim, John and Moin, Parviz},
  journal={Journal of computational physics},
  volume={59},
  number={2},
  pages={308--323},
  year={1985},
  publisher={Elsevier}
}

@article{ubal2024adjoint,
  title={Adjoint-based design optimization of a Kenics static mixer},
  author={Ubal, Cristopher Morales and Beishuizen, Nijso and Kusch, Lisa and van Oijen, Jeroen},
  journal={Results in Engineering},
  volume={21},
  pages={101856},
  year={2024},
  publisher={Elsevier}
}

@article{sockol2003multigrid,
  title={Multigrid solution of the Navier--Stokes equations at low speeds with large temperature variations},
  author={Sockol, Peter M},
  journal={Journal of Computational Physics},
  volume={192},
  number={2},
  pages={570--592},
  year={2003},
  publisher={Elsevier}
}

@article{peng2003transition,
  title={Transition in a 2-D lid-driven cavity flow},
  author={Peng, Yih-Ferng and Shiau, Yuo-Hsien and Hwang, Robert R},
  journal={Computers \& Fluids},
  volume={32},
  number={3},
  pages={337--352},
  year={2003},
  publisher={Elsevier}
}

@article{ramanan1994linear,
  title={Linear stability of lid-driven cavity flow},
  author={Ramanan, Natarajan and Homsy, George M},
  journal={Physics of Fluids},
  volume={6},
  number={8},
  pages={2690--2701},
  year={1994},
  publisher={American Institute of Physics}
}

@article{sahin2003novel,
  title={A novel fully-implicit finite volume method applied to the lid-driven cavity problem—Part II: Linear stability analysis},
  author={Sahin, Mehmet and Owens, Robert G},
  journal={International journal for numerical methods in fluids},
  volume={42},
  number={1},
  pages={79--88},
  year={2003},
  publisher={Wiley Online Library}
}

@incollection{kuhlmann2019,
    author="Kuhlmann, Hendrik C.
        and Roman{\`o}, Francesco",
    title="The Lid-Driven Cavity",
    bookTitle="Computational Modelling of Bifurcations and Instabilities in Fluid Dynamics",
    year="2019",
    publisher="Springer International Publishing",
    address="Cham",
    pages="233--309"
}

@article{dritschel1995general,
  title={A general theory for two-dimensional vortex interactions},
  author={Dritschel, David G},
  journal={Journal of Fluid Mechanics},
  volume={293},
  pages={269--303},
  year={1995},
  publisher={Cambridge University Press}
}

@article{josserand2007merging,
  title={The merging of two co-rotating vortices: a numerical study},
  author={Josserand, Ch and Rossi, M},
  journal={European Journal of Mechanics-B/Fluids},
  volume={26},
  number={6},
  pages={779--794},
  year={2007},
  publisher={Elsevier}
}

@article{trieling2010dynamics,
  title={Dynamics of two identical vortices in linear shear},
  author={Trieling, RR and Dam, CEC and van Heijst, GJF},
  journal={Physics of Fluids},
  volume={22},
  number={11},
  year={2010},
  publisher={AIP Publishing}
}

@book{Pope2000,
    author    = {S. B. Pope},
    title     = {Turbulent Flows},
    publisher = {Cambridge University Press},
    year      = {2000},
    address   = {Cambridge, UK},
}

@article{Kolmogorov1941,
  author  = {A. N. Kolmogorov},
  title   = {The Local Structure of Turbulence in Incompressible Viscous Fluid for Very Large Reynolds Numbers},
  journal = {Doklady Akademii Nauk SSSR},
  volume  = {30},
  pages   = {301--305},
  year    = {1941}
}

@book{Stull1988,
    author    = {R. B. Stull},
    title     = {An Introduction to Boundary Layer Meteorology},
    publisher = {Kluwer Academic Publishers},
    year      = {1988},
    address   = {Dordrecht, The Netherlands},
}

@book{Vallis2017,
    author    = {G. K. Vallis},
    title     = {Atmospheric and Oceanic Fluid Dynamics},
    edition   = {2},
    publisher = {Cambridge University Press},
    year      = {2017},
    address   = {Cambridge, UK},
}

@article{Nezu2005,
  author  = {I. Nezu and H. Nakagawa},
  title   = {Turbulence in Open-Channel Flows},
  journal = {IAHR Monograph Series},
  year    = {2005}
}

@book{Schlichting2000,
    author    = {H. Schlichting and K. Gersten},
    title     = {Boundary-Layer Theory},
    edition   = {8},
    publisher = {Springer},
    year      = {2000},
    address   = {Berlin, Heidelberg},
}

@book{Richardson1922,
    author    = {L. F. Richardson},
    title     = {Weather Prediction by Numerical Process},
    publisher = {Cambridge University Press},
    year      = {1922},
    address   = {Cambridge, UK},
}

@book{Stein1999,
    author    = {Michael L. Stein},
    title     = {Interpolation of Spatial Data: Some Theory for Kriging},
    publisher = {Springer},
    year      = {1999},
    address   = {New York, NY},
}

@book{Wendland2004,
    title        = {Scattered Data Approximation},
    author       = {Wendland, Holger},
    publisher    = {Cambridge University Press},
    year         = {2005},
    address      = {Cambridge, UK},
}

@inproceedings{drineas2005approximating,
  title={Approximating a gram matrix for improved kernel-based learning},
  author={Drineas, Petros and Mahoney, Michael W},
  booktitle={International Conference on Computational Learning Theory},
  pages={323--337},
  year={2005},
  organization={Springer}
}

@inproceedings{gittens2013revisiting,
  title={Revisiting the nystrom method for improved large-scale machine learning},
  author={Gittens, Alex and Mahoney, Michael},
  booktitle={International Conference on Machine Learning},
  pages={567--575},
  year={2013},
  organization={PMLR}
}

@incollection{yokota2011treecode,
    title={Treecode and fast multipole method for N-body simulation with CUDA},
    author={Yokota, Rio and Barba, Lorena A},
    booktitle={GPU Computing Gems Emerald Edition},
    pages={113--132},
    year={2011},
    publisher={Elsevier},
    address      = {Boston, MA}
}

@article{rokhlin1985rapid,
  title={Rapid solution of integral equations of classical potential theory},
  author={Rokhlin, Vladimir},
  journal={Journal of computational physics},
  volume={60},
  number={2},
  pages={187--207},
  year={1985},
  publisher={Elsevier}
}

@article{micchelli2005learning,
  title={On learning vector-valued functions},
  author={Micchelli, Charles A and Pontil, Massimiliano},
  journal={Neural computation},
  volume={17},
  number={1},
  pages={177--204},
  year={2005},
  publisher={MIT Press}
}

@inproceedings{duchon1976splinesmr,
  title={Splines minimizing rotation-invariant semi-norms in Sobolev spaces},
  author={Duchon, Jean},
  booktitle={Constructive theory of functions of several variables: proceedings of a conference held at oberwolfach April 25--May 1, 1976},
  pages={85--100},
  year={2006},
  organization={Springer}
}

@book{leveque2002finite,
    title={Finite volume methods for hyperbolic problems},
    author={LeVeque, Randall J},
    volume={31},
    year={2002},
    publisher={Cambridge university press},
    address = {Cambridge, UK},
}

@book{shankar2014radial,
  title={Radial basis function-based numerical methods for the simulation of platelet aggregation},
  author={Shankar, Varun},
  year={2014},
  publisher={The University of Utah},
  address={Utah, US}
}

@article{flyer2011radial,
  title={Radial basis functions: Developments and applications to planetary scale flows},
  author={Flyer, Natasha and Fornberg, Bengt},
  journal={Computers \& Fluids},
  volume={46},
  number={1},
  pages={23--32},
  year={2011},
  publisher={Elsevier}
}

@article{flyer2009radial,
  title={A radial basis function method for the shallow water equations on a sphere},
  author={Flyer, Natasha and Wright, Grady B},
  journal={Proceedings of the Royal Society A: Mathematical, Physical and Engineering Sciences},
  volume={465},
  number={2106},
  pages={1949--1976},
  year={2009},
  publisher={The Royal Society}
}

@article{flyer2007transport,
  title={Transport schemes on a sphere using radial basis functions},
  author={Flyer, Natasha and Wright, Grady B},
  journal={Journal of Computational Physics},
  volume={226},
  number={1},
  pages={1059--1084},
  year={2007},
  publisher={Elsevier}
}

@article{kassen2021immersed,
    title = {Immersed boundary simulations of cell-cell interactions in whole blood},
    journal = {Journal of Computational Physics},
    volume = {469},
    pages = {111499},
    year = {2022},
    issn = {0021-9991},
    author = {Andrew Kassen and Aaron Barrett and Varun Shankar and Aaron L. Fogelson},
}

@book{hastie2009elements,
    title={The elements of statistical learning: data mining, inference, and prediction},
    author={Hastie, Trevor and Tibshirani, Robert and Friedman, Jerome H and Friedman, Jerome H},
    volume={2},
    year={2009},
    publisher={Springer},
    address = {New York, NY}
}

@article{drake2021partition,
  title={A partition of unity method for divergence-free or curl-free radial basis function approximation},
  author={Drake, Kathryn P and Fuselier, Edward J and Wright, Grady B},
  journal={SIAM Journal on Scientific Computing},
  volume={43},
  number={3},
  pages={A1950--A1974},
  year={2021},
  publisher={SIAM}
}

@article{flyer2016role,
  title={On the role of polynomials in RBF-FD approximations: I. Interpolation and accuracy},
  author={Flyer, Natasha and Fornberg, Bengt and Bayona, Victor and Barnett, Gregory A},
  journal={Journal of Computational Physics},
  volume={321},
  pages={21--38},
  year={2016},
  publisher={Elsevier}
}

@article{sriperumbudur2011universality,
  title   = {Universality, Characteristic Kernels and {RKHS} Embedding of Measures},
  author  = {Sriperumbudur, Bharath K. and Fukumizu, Kenji and Lanckriet, Gert R. G.},
  journal = {Journal of Machine Learning Research},
  volume  = {12},
  pages   = {2389--2410},
  year    = {2011}
}

@book{runst1996sobolev,
  title     = {Sobolev Spaces of Fractional Order, Nemytskij Operators, and Nonlinear Partial Differential Equations},
  author    = {Runst, Thomas and Sickel, Winfried},
  publisher = {Walter de Gruyter},
  year      = {1996},
  isbn      = {9783110151138}
}

@article{bourdaud2023composition,
  title   = {An Introduction to Composition Operators in Sobolev Spaces},
  author  = {Bourdaud, G{\'e}rard},
  journal = {Eurasian Mathematical Journal},
  volume  = {14},
  number  = {1},
  pages   = {39--54},
  year    = {2023},
  doi     = {10.32523/2077-9879-2023-14-1-39-54}
}

@article{bacho2026krom,
  title={KROM: Kernelized Reduced Order Modeling},
  author={Bacho, Aras and Lee, Jonghyeon and Owhadi, Houman},
  journal={arXiv preprint arXiv:2603.00360},
  year={2026}
}

@article{smagorinsky1963general,
  title={General circulation experiments with the primitive equations: I. The basic experiment},
  author={Smagorinsky, Joseph},
  journal={Monthly weather review},
  volume={91},
  number={3},
  pages={99--164},
  year={1963}
}

@article{meneveau2000scale,
  title={Scale-invariance and turbulence models for large-eddy simulation},
  author={Meneveau, Charles and Katz, Joseph},
  journal={Annual Review of Fluid Mechanics},
  volume={32},
  number={1},
  pages={1--32},
  year={2000},
  publisher={Annual Reviews 4139 El Camino Way, PO Box 10139, Palo Alto, CA 94303-0139, USA}
}

@article{arenas2006existence,
  title={On the existence and scaling of structure functions in turbulence according to the data},
  author={Arenas, Alexandre and Chorin, Alexandre J},
  journal={Proceedings of the National Academy of Sciences},
  volume={103},
  number={12},
  pages={4352--4355},
  year={2006},
  publisher={National Academy of Sciences}
}

@article{canatar2021out,
  title={Out-of-distribution generalization in kernel regression},
  author={Canatar, Abdulkadir and Bordelon, Blake and Pehlevan, Cengiz},
  journal={Advances in Neural Information Processing Systems},
  volume={34},
  pages={12600--12612},
  year={2021}
}

@inproceedings{wilson2013gaussian,
  title={Gaussian process kernels for pattern discovery and extrapolation},
  author={Wilson, Andrew and Adams, Ryan},
  booktitle={International conference on machine learning},
  pages={1067--1075},
  year={2013},
  organization={PMLR}
}
\end{document}